\documentclass[11pt, svgnames]{article}

\usepackage[margin=1in]{geometry}

\usepackage{comment}
\usepackage{cite}

\usepackage{hyperref}
\hypersetup{
colorlinks=true,
linkcolor=blue,
filecolor=blue,
citecolor=blue,  
urlcolor=blue,
}

\usepackage{amsmath,amssymb,amsthm,mathrsfs,amsfonts,dsfont,amscd,keyval}
\usepackage{mathtools,mathrsfs}
\usepackage{yfonts} 
\usepackage{ytableau}
\usepackage{subfigure, epsfig,graphicx}
\usepackage{braket}
\usepackage{tensor}
\usepackage{bm}
\usepackage{enumerate}
\usepackage{color}
\usepackage{hyperref}
\usepackage{tabularx}
\hypersetup{colorlinks=true, linkcolor=blue, citecolor=blue, urlcolor=black, linktocpage=true}
\usepackage{multirow}
\usepackage{float}
\usepackage{tikz}
\usepackage{tikz-cd}
\usepackage{quantikz}
\usepackage{adjustbox}
\usepackage{pgfplots}
\pgfplotsset{compat=1.12}
\usetikzlibrary{patterns}
\usepackage{quiver} 
\usepackage{ytableau}
\usepackage{nicematrix}

\theoremstyle{definition}
\newtheorem{definition}{Definition}[section]
\newtheorem{example}{Example}

\newtheorem{proposition}[definition]{Proposition}

\theoremstyle{plain}
\newtheorem{theorem}[definition]{Theorem}
\newtheorem{lemma}[definition]{Lemma}
\newtheorem{corollary}[definition]{Corollary}

\newtheorem{conjecture}{Conjecture}

\theoremstyle{remark}
\newtheorem*{remark}{Remark}

\newcommand{\comments}[1]{} 

\newcommand{\tr}{\mathrm{Tr}}

\DeclareMathOperator{\U}{U}
\DeclareMathOperator{\SU}{SU}
\DeclareMathOperator{\Gadget}{IG}
\DeclareMathOperator{\Ima}{Im}
\DeclareMathOperator{\CNOT}{\mathsf{CNOT}}

\DeclareMathOperator{\SQSW}{\mathsf{SQSW}}
\DeclareMathOperator{\QFT}{\mathsf{QFT}}
\DeclareMathOperator{\SWAP}{\mathsf{SWAP}}
\DeclareMathOperator{\B}{\mathsf{B}}
\DeclareMathOperator{\iSWAP}{\mathsf{iSWAP}}
\DeclareMathOperator{\SQiSW}{\mathsf{SQiSW}}
\DeclareMathOperator{\CCZ}{\mathsf{CCZ}}
\DeclareMathOperator{\Toffoli}{\mathsf{Toffoli}}
\DeclareMathOperator{\CSWAP}{\mathsf{CSWAP}}
\DeclareMathOperator{\CiSWAP}{\mathsf{CiSWAP}}
\DeclareMathOperator{\Peres}{\mathsf{Peres}}
\DeclareMathOperator{\Margolus}{\mathsf{Margolus}}
\DeclareMathOperator{\T}{\mathcal{T}}

\usepackage{amsmath,amsthm,amsfonts,graphicx}
\usepackage{amssymb}
\usepackage{epsfig,graphicx,color}
\usepackage{mathrsfs}
\usepackage[active]{srcltx}
\usepackage{url}
\usepackage{enumerate}
\usepackage{multirow}
\usepackage{mathtools}
\usepackage{authblk}

\newcommand{\F}{\mathbb{F}}

\def\>{\rangle}
\def\<{\langle}

\newcommand{\done}[1]{ { \color{green} done }}

\begin{document}

\title{Convergence efficiency of quantum gates and circuits}

\newcommand{\cC}{\mathcal{C}}

\author[1]{Linghang Kong}
\author[2]{Zimu Li}
\author[2]{Zi-Wen Liu}
\affil[1]{Zhongguancun Laboratory}
\affil[2]{Yau Mathematical Sciences Center, Tsinghua University}

\date{}

\maketitle

\begin{abstract} 
We consider quantum circuit models where the gates are drawn from arbitrary gate ensembles given by probabilistic distributions over certain gate sets and circuit architectures, which we call {stochastic quantum circuits}. Of main interest in this work is the speed of convergence of stochastic circuits with different gate ensembles and circuit architectures to unitary $t$-designs, which we analyze through the spectral gaps of the associated moment operators. A key motivation for this theory is the varying preference for different gates and circuit architectures in different practical scenarios, such as different experimental platforms. In particular, it provides a versatile framework for devising efficient circuits for implementing $t$-designs and relevant applications including random circuit and scrambling experiments, as well as evaluating and optimizing the efficiency of gates and circuit architectures. We examine various important settings in depth, showcasing numerous useful analyses and findings along the way. 

A key aspect of our study is an ``ironed gadget'' model made up of entangling gates and random single-qubit gates, which allows us to systematically evaluate and compare the convergence efficiency of entangling gates and circuit architectures. Particularly notable results include i) gadgets of two-qubit gates with KAK coefficients $\left(\frac{\pi}{4}-\frac{1}{8}\arccos(\frac{1}{5}),\frac{\pi}{8},\frac{1}{8}\arccos(\frac{1}{5})\right)$ (which we call $\chi$ gates) directly form exact 2- and 3-designs and are thus optimal in the sense of autoconvolution; ii) the $\iSWAP$ gate family achieves the best efficiency for convergence to 2-designs under mild conjectures with numerical evidence, even outperforming the Haar-random gate, for generic many-body circuits; iii) $\iSWAP$ + complete graph achieve the best efficiency for convergence to 2-designs among all 2-local random ensembles built on graphs. 
A variety of numerical results are provided to complement our analysis. We also derive robustness guarantees for our analysis against gate perturbations.
Additionally, we provide cursory analysis on gates with higher locality and found, for instance, that the $\Margolus$ gate outperforms several other well-known gates and 2-local gates.
\end{abstract}

\tableofcontents

\section{Introduction}

Quantum circuits represent a fundamental model for quantum computation, in which certain sequences of local quantum gates acting on a limited number of subsystems (like qubits) generate the global evolution of many-body quantum systems. These local gates describe the elementary operations that can be done in a time step in practical setups, which serve as building blocks of global quantum computation. The fact that local quantum circuits are capable of generating arbitrary global evolution of quantum systems underpins the feasibility of quantum computation and other technologies. 
As a framework that naturally captures the locality of physical interactions and thereby induces a time scale, quantum circuits have also emerged as a suitable model for complex many-body quantum systems which plays key roles in recent developments of physics.

Modeling how global quantum randomness is generated from random local dynamics, the random quantum circuits are of particular importance from both practical and theoretical perspectives.
First of all, they have found a broad array of technological applications across areas including quantum device benchmarking~\cite{Randomized2005,knill2008randomized,Bannai2021,Elben_2022}, information theory~\cite{hayden2004randomizing,Hayden_2006,Hayden-Preskill2007,dupuis2010decouplingapproachquantuminformation,Dupuis_2014,HastingsSuperAdd,Haferkamp2024linear}, error correction~\cite{brown2013short,Brown_2015,preskill2021approxeastinknill,kong2022near},  tomography~\cite{Scott2008,Huang2020shadow,Bertoni2024,Elben_2022,Zhu2024Clifford}, machine learning~\cite{biamonte2017quantum,NetKet,Zheng2021SpeedingUL,Liu2022QNTK,Liu2023QNTK,Elben_2022}, as well as the demonstration of quantum computational supremacy~\cite{arute2019quantum,Bertoni2024,decross2024}.
On the other hand, the random circuit models have also attracted major interest and extensive investigations in physics contexts in recent years as canonical models of quantum dynamics, offering critical insights into the nature of complex quantum many-body systems. This has significantly advanced our understanding of various related  dynamical phenomena under the names of information scrambling, thermalization, chaos, complexity growth etc.~\cite{Hayden-Preskill2007,SekinoSusskind2008,Hayden2013FastScrambe,nahum2018operator,Liu_2018,Swingle2022,SUd-k-Design2023Application,khemani2018operator,von2018operator,Huang2019OTOC,Motta_2019,roberts2016chaos,junyu2017chaos,RobertsChaos2017,Brown-Susskind2018,haferkamp2022linear}, which have been a key driving force of recent developments at the frontiers of quantum many-body physics and gravity.
Furthermore, underlying the study of random quantum circuits revolves are profound connections with wide-ranging areas of pure mathematics including probability theory, stochastic processes, dynamical systems, geometry, representation theory etc.~(see e.g.~the discussions on design generation in the following), which have brought forth various problems and directions of mathematical interest and significantly enriched the interaction between quantum information and mathematics.
Conceptually, random ensembles provide us with a powerful yet often tractable lens for understanding the generic features and behaviors of quantum systems which are of fundamental importance from diverse perspectives.

A central problem in the study of random quantum circuits is to understand and optimize the rate, or the requisite circuit depth (time), for local circuit ensembles to converge to unitary $t$-designs (we may simply refer to them as $t$-designs when there is no ambiguity hereafter), i.e., statistically ``pseudorandom'' distributions of unitaries that mimic the Haar (uniform) measure up to $t$-th moments. The study of this design convergence problem has a long history~\cite{Dankert2026PRA,Gross2006,Oliveira2design2007a,Oliveira2design2007b,Scott2008,HarrowTEP08,HarrowTEP09,Harrow2design2009,Arnold2009,Diniz_2011,szehr2013decoupling,webb2015clifford,zhu2016clifford,zhu2017multiqubit,BHH16,BHH16prl,hunter2019unitary}, and has seen extensive progress lately~\cite{brandao2021models,Bannai2021,harrow2023approximate,U(1)Design2023,MarvianSUd,SUd-k-Design2023,Marvian3local,MarvianDesign,mittal2023local,Haah2024Pauli,Metger2024,Haferkamp2024linear,mitsuhashi2024Designs,Schuster2024lowDepth}. Instead of the full Haar measure which requires exponential circuit depths to approximate~\cite{knill1995approximationquantumcircuits}, the technological applications of random ensembles listed before generally only require low-order design properties. Moreover, the  most important physical phenomena of interest to the study of the physics of complex quantum systems, such as quantum information scrambling~\cite{roberts2016chaos,RobertsChaos2017} and global entanglement generation~\cite{OliveiraDahlstenPlenio08,roberts2016chaos,Liu_2018,PhysRevLett.120.130502} already manifest themselves at the level of $t$-designs with the critical landmark being $t=2$.

In random circuit constructions, it is common to take Haar-random local gates as the building blocks. Here, we consider and explore a more general scheme of random circuits, which we call \emph{stochastic quantum circuits}, where the gates can be drawn from generic ensembles. This provides a versatile framework not only for studying the generation of quantum (pseudo)randomness and scrambling, but furthermore, for benchmarking and optimizing the efficiency of quantum gates and circuits. Particularly, this theory has evident motivations from experiment perspectives. First, certain gates may be easier to implement than others for a certain experimental setup or platform, or more formally, different gates are associated with different weights or costs in particular setups, naturally linking to the general stochastic quantum circuit model. Therefore, a systematic framework for evaluating  the efficiency and power of different gates and circuit architectures would be highly desirable for experimental efforts, offering a guideline for optimizing the experimental designs for random circuit experiments and quantum tasks in general. There are several works in the literature that have considered the comparison of different gates and circuit architectures, mostly based on numerical observations~\cite{Znidaric07:2,WBV08,BWV08:cluster}. However, a rigorous and systematic theory remains to be established.  

Here we examine the convergence speed of stochastic quantum circuits with different gate ensembles and circuit architectures to unitary $t$-designs through the spectral gaps of the associated moment operators. To begin with, we  consider various natural gate ensembles such as those composed by Hadamard and phase gates and standard entangling gates, and showcase various basic types of analysis valuable for gate and circuit design that can be done using the stochastic circuit framework. It is found that these simple examples already exhibit surprisingly complicated features even when considering the convergence to low-order designs, signifying the profound behaviors of stochastic circuits.

A major part of this work is based on what we call the ``ironed gadget'' model, motivated by the ``easiness'' of single-qubit gates. In particular, single-qubit gates can be rather easily implemented to high precision and commonly induce negligible cost compared to entangling gates in experiments nowadays.
Specifically, the ingredients of this model are gate gadgets given by a certain entangling gate supplemented by single-qubit gates set to be Haar-random and thus effectively averaged out.
The motivations can be perceived from various perspectives. For instance, to evaluate gate and circuit efficiency, it is desirable to ``mod out'' the effects of single-qubit gates due to their low cost.  As such, this model lays the ground for a rigorous, refined theory for comparing the efficiency of entangling gates and identifying particularly good ones. Moreover, 
the gadgets provide a practically relevant scheme for implementing random entangling gates,  potentially yielding particularly favorable constructions for experiments combining with the efficiency results.
Mathematically, the model also aligns with Cartan's KAK decomposition of 2-qubit gates~\cite{Khaneja2001,Zhang2003KAK} into 1-qubit gates and 2-qubit Pauli rotations through angles referred to as KAK coefficients. Our key findings based on the ironed gadget theory are summarized as follows.
First, we find a new two-qubit gate families with KAK coefficients $(\frac{\pi}{4}-\frac{1}{8}\arccos(\frac{1}{5})$, $\frac{\pi}{8}$, $\frac{1}{8}\arccos(\frac{1}{5}))$, which we call $\chi$ gates, that are solutions for gadgets to form exact 2- and 3-designs on 2-qubit systems. That is, we only need to apply a gate from this family once together with single-qubit Haar-random gates to construct an exact 2- and 3-design. 
That is, they produce exactly the same moment operators as the Haar-random gate and achieve the best efficiency among all gadgets.
The existence of such exact solutions for order 2 and 3 turns out to be a fortuity: we can prove that there does not exist any solution of the KAK coefficients that can form exact designs of higher orders.
Then, for many-body circuits defined on e.g., the typical 1D chain and all-to-all graphs, we find that the $\iSWAP$ gate family is particularly efficient, even outperforming the 2-qubit Haar-random gate,  for convergence to 2-designs. More specifically, we prove that $\iSWAP$  achieves the highest efficiency within a significant portion of the Weyl chamber and provide evidence for its highest efficiency among all 2-qubit gates.  In particular, we show that $\iSWAP$ + complete graph architecture achieves the best efficiency for convergence to 2-designs among all 2-local random circuits defined on arbitrary graphs. 
Moreover, we also provide numerical results for an array of important gates and graphs as important complements to the theoretical analysis.
For higher designs, we provide numerical comparisons and find that neither the local Haar-random gates nor  $\iSWAP$ remains the best choice, which demonstrates the fact that the behaviors of different moments can exhibit significant distinctions. The theoretical exploration of  efficient generation of high-order designs is left as an avenue for future work.
Additionally, we provide cursory analysis on gates with higher locality and found, for instance, that the  $\Margolus$ gate is the most efficient among a set of well-known 3-qubit and outperforms 2-local gates in forming 2-designs. 
We also study the Clifford gates plus phase gates model. Notably, for varying  phase angle and probability distribution, we  derive analytical results for the fastest convergence to 4- and 5-designs on any $n$-qubit system.
{In addition to the main findings outlines here, there are a variety of detailed theoretical and numerical results that can be found in the paper.}

This paper is organized as follows. In Section~\ref{sec:framework}, we define the central concepts and introduce the theory for benchmarking the efficiency of a random ensemble in design generation through the spectral gap of the associated moment operators. In Section~\ref{sec:one-qubit}, as heuristic examples for the stochastic circuit framework, we present numerical results for certain natural ensembles on 1- and 2-qubit systems. Section~\ref{sec:two-qubit} is devoted to the ironed gadget model. We first present relevant definitions and then the multifaceted mathematical theory for evaluating and comparing the convergence efficiency of gadgets based on different entangling gates. Besides circuits defined on graphs, we numerically study other architectures such as the brickwork model in Section \ref{sec:architectures}. Then we verify in Section \ref{sec:Perturbation} the robustness of the gadget models when the KAK coefficients are perturbed in real experiments. We provide a summary of the main results for readers' convenience in Section~\ref{sec:summary}, and leave several conjectures in Section~\ref{sec:conjectures} based on our numerical and theoretical findings. In Section~\ref{sec:multiqubit}, we extend our study to multiqubit gates. Especially, we provide numerical analysis for several typical 3-qubit gates. In Section~\ref{sec:large-n}, we study the Clifford plus diagonal gate set and provide solutions for its fast convergence to 4- and 5-designs.

\section{Framework: stochastic quantum circuits and efficiency of design generation}\label{sec:framework}

We start by formally introducing our framework and key definitions. Consider a general distribution/ensemble $\mathcal{E}$ of unitary quantum gates given by some gate set $\mathcal{U}$ and a probability measure $\nu$ over $\mathcal{U}$. A \emph{stochastic quantum circuit} is defined by applying a gate $U$ drawn randomly from the ensemble $\mathcal{E}$, denoted as $U \sim \nu$, at each time step.  

We are interested in the speed at which certain stochastic quantum circuits converges to unitary $t$-designs. Starting with formal definitions of the central concepts, we shall formally introduce the mathematical foundation of our study in this section.

\begin{definition}[$t$-fold channels and $t$-th moment (super-)operators]
	Let $\mathcal{H} = (\mathbb{C}^2)^{\otimes n}$ be the $n$-qubit Hilbert space. Given an arbitrary operator $M \in \operatorname{End}(\mathcal{H}^{\otimes t})$, the \emph{$t$-fold (twirling) channel} associated with unitary ensemble $\mathcal{E}$ is given by \begin{align}\label{eq:MomentOperator}
		\mathscr{T}_t^{\mathcal{E}}(M) = \int_{\mathcal{E}}  U^{\otimes t} M U^{\dagger \otimes t} dU.
	\end{align}
	As a linear map acting on $\operatorname{End}(\mathcal{H}^{\otimes t})$, it can be reformulated as the \emph{$t$-th moment (super-)operator}:
	\begin{align}
		\T_t^{\mathcal{E}} = \int_{\mathcal{E}} U^{\otimes t } \otimes \bar{U}^{ \otimes t} dU,
	\end{align}
	where $\bar U$ is the complex conjugate of $U$. Particularly, when $\mathcal{E}$ is taken to be the Haar measure of the $n$-qubit unitary group $\U(\mathcal{H}) \cong \U(2^n)$, the corresponding moment operator is simply denoted by $\T_t^{\U(2^n)}$.
\end{definition}

\begin{definition}[Unitary $t$-designs]
    An ensemble $\mathcal{E}$ of $n$-qubit unitaries is an \emph{(exact) unitary $t$-design}  if $\T_t^{\mathcal{E}} = \T_t^{\U(2^n)}$. That is, the $t$-th moment of the ensemble $\mathcal{E}$ matches that of the Haar measure. 
\end{definition}

The bi-invariance of Haar measure implies that $\T_t^{\U(2^n)}$ is a \emph{projector} from $\operatorname{End}(\mathcal{H}^{\otimes t})$ onto the \emph{commutant algebra} 
\begin{align}
	\text{Comm}_t( \U(2^n) ) \vcentcolon = \{M \in \operatorname{End}(\mathcal{H}^{\otimes t}); U^{\otimes t} M = M U^{\otimes t} \}, 
\end{align}
Therefore, eigenvalues of $\T_t^{\U(2^n)}$ are either 0 or 1. Also see Refs.~\cite{Dankert2026PRA,Gross2006,HarrowTEP08,HarrowTEP09} for more details.

Generally, for any $M \in \operatorname{End}( ((\mathbb{C}^2)^{\otimes n})^{\otimes t} ) \cong ((\mathbb{C}^2)^{\otimes n})^{\otimes 2t}$, \emph{Schur--Weyl duality}~\cite{Goodman2009,Tolli2009} indicates that $\T_t^{\U(2^n)}$ projects $M$ into the subspace spanned by the action of the symmetric group $S_t$ permuting tensors from $((\mathbb{C}^2)^{\otimes n})^{\otimes t}$. That is,
\begin{align}
	\text{Comm}_t( \U(2^n) ) = \mathrm{span} \{\sigma \in S_t\} \subset \operatorname{End}( ((\mathbb{C}^2)^{\otimes n})^{\otimes t} ),
\end{align}
where $\sigma$ is a permutation on the $t$-fold Hilbert spaces.
Given any permutation $\sigma \in S_t$ and integers $1 \leq i_1 < \cdots < i_l \leq t$, they form an \emph{increasing subsequence} of $\sigma$ if $\sigma(i_1) < \cdots < \sigma(i_l)$. It is proved in Ref.~\cite{Rains1998} that the number of linearly independent $S_t$ actions equals the number of permutations with no increasing subsequence of length greater than the dimension, which is $2^n$ here. 
Let $D \coloneqq \dim \text{Comm}_t( \U(2^n) )$ denote the dimension of the commutant.
Obviously, if $t \leq 2^n$, all those permutations are linearly independent and hence  $D = t!$. However, for the simplest $n=1$ (single-qubit) system~\cite{Rains1998}, 
\begin{align}\label{eq:1qDimension}
	D = \dim \text{Comm}_t( \text{U}(2) ) = \dim \Ima \T_t^{\U(2)} = \frac{(2t)!}{t!(t+1)!},
\end{align} 
which is strictly less than $t!$ when $t > 2$.

Assume $T_t^{\mathcal{E}}$ is Hermitian, which can be easily fulfilled by incorporating the gate and its inverse with equal probability when defining $\mathcal{E}$. Due to the bi-invariance of Haar measure, $\T_t^{\mathcal{E}}, \T_t^{\U(2^n)}$ are commutative and thus can be simultaneously diagonalized. It is straightforward to see that as long as the infinity norm $\Vert \T_t^{\mathcal{E}} - \T_t^{\U(2^n)} \Vert_{\infty} < 1$, 
\begin{align}
	\Vert (\T_t^{\mathcal{E}})^p - \T_t^{\U(2^n)} \Vert_{\infty} \to 0
\end{align} 
as we take larger and larger $p$ which characterizes the number of steps or depth of the random circuit. In this case,  $\T_t^{\mathcal{E}}$ and $\T_t^{\U(2^n)}$ must share the same unit eigenspace, which is exactly $\operatorname{Comm}_t( \U(2^n) )$. Then $\Vert \T_t^{\mathcal{E}} - \T_t^{\U(2^n)} \Vert_{\infty}$ equals the second largest eigenvalue $\lambda_2(\T_t^{\mathcal{E}})$ or the absolute value of the smallest eigenvalue $\lambda_{\min}(\T_t^{\mathcal{E}})$ of $\T_t^{\mathcal{E}}$, depending on which one is larger and closer to $1$. Therefore, we see that the efficiency or rate at which an ensemble converges to $t$-designs is determined by the spectral gap of the $t$-th moment operators, formally defined as follows.
\begin{definition}[Spectral gap]
	The \emph{spectral gap} of the Hermitian $t$-th moment operator $T_t^{\mathcal{E}}$ is defined as
	\begin{align}\label{eq:Gap}
		\Delta( \T_t^{\mathcal{E}}) \coloneqq 1 - \max \{\lambda_2(\T_t^{\mathcal{E}}), \vert \lambda_{\min}(\T_t^{\mathcal{E}}) \vert \}.
	\end{align}
\end{definition}

To be more precise, consider the conventional way to define approximate $t$-designs based on the closeness of the two moment operators $\T_t^{\mathcal{E}}$, $\T_t^{\U(2^n)}$ in terms of complete positivity. we say $\T_t^{\mathcal{E}}$  
 That is, if
\begin{align}
	(1 - \epsilon) \T_t^{\U(2^n)} \leq_{\mathrm{cp}} \T_t^{\mathcal{E}} \leq_{\mathrm{cp}} (1 + \epsilon) \T_t^{\U(2^n)},
\end{align} 
where $A \leq_{\mathrm{cp}} B$ means $B - A$ is completely positive, then $\T_t^{\mathcal{E}}$ is said to form an \emph{$\epsilon$-approximate $t$-design}~\cite{BHH16,Gao2022,Metger2024,Haferkamp2024linear}). Denoting by $c_{\mathrm{cp}}(\mathcal{E}, t)$ the smallest constant $\epsilon$ achieving the above bound, it is proved in Ref.~\cite{BHH16} that when the circuit depth $p$ is no smaller than
\begin{align}\label{eq:convergence}
	\frac{1}{\Delta( \T_t^{\mathcal{E}})} \log \frac{2^{2nt}}{\epsilon} = \frac{1}{\Delta( \T_t^{\mathcal{E}})} \Big( 2nt \log 2 + \log \frac{1}{\epsilon} \Big),
\end{align}
the generated ensemble forms an $\epsilon$-approximate $t$-designs within precision $\epsilon$ (also see Ref.\cite{Schuster2024lowDepth} for an improvement on the dependence of $t$).


\section{Heuristic examples of convergence analysis}\label{sec:one-qubit}

We start with discussing some basic examples of the application of the stochastic circuit framework. The main purpose is to demonstrate various basic types of analysis that can be done using the stochastic circuit framework and the complexity of its behaviors even for simple cases. 
 
\subsection{Single-qubit gates}

We first consider the simplest quantum system with only one qubit as a warmup.

As an example, we consider the gate set consisting of the Hadamard gate $H$ and a diagonal gate $\operatorname{diag}(1,e^{i\theta})$. A discrete probability distribution is defined as follows: let $p/2$ be the probability for $H$, $(1-p)/4$ for $\operatorname{diag}(1,e^{i\theta})$, $(1-p)/4$ for $\operatorname{diag}(1,e^{-i\theta})$ and $1/2$ for identity. Recall that sampling both $\operatorname{diag}(1,e^{i\theta})$ and $\operatorname{diag}(1,e^{-i\theta})$ with the same probability guarantees that the corresponding $t$-th moment operator $\T_t^{\mathcal{E}}$ is Hermitian. Adding the identity shifts every eigenvalue $x$ to $(x+1)/2$, so that all eigenvalues of $\T_t^{\mathcal{E}}$ become non-negative and that the  convergence rate is purely limited by the gap between 1 and the second largest eigenvalue. We should stress that the negative eigenvalues can play key roles in the convergence theory which will be utilized later, but for now we focus on the simplest case.

For the $T$ gate with $\theta=\pi/4$, a figure of gap as a function of $p$ is show in Fig.~\ref{fig:pi-over-4} with $t=3$. The maximum gap is 0.0433879, which corresponds to $m \ge 107$ in Eq.~\eqref{eq:convergence} for $\epsilon = 0.01$.

\begin{figure}[H]
 \centering
 \includegraphics[width=0.4\textwidth]{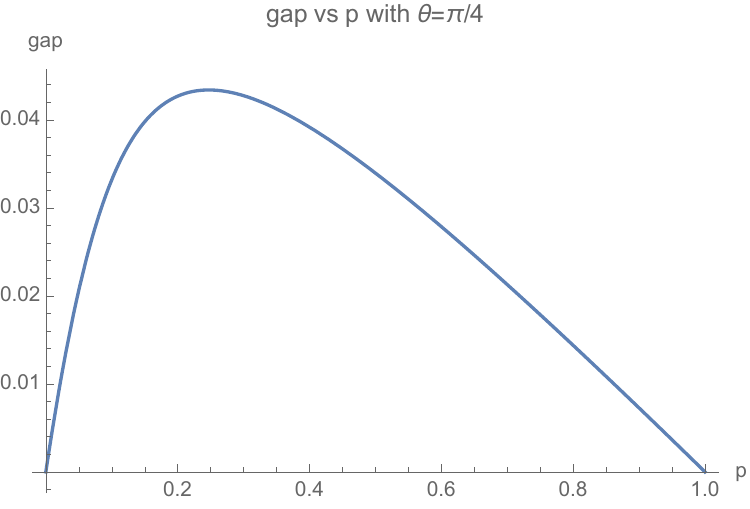}
 \caption{The gap vs $p$ for the gate set $\{H,T\}$. $t$ is set to be 3.}
 \label{fig:pi-over-4}
\end{figure}

Now we treat $\theta$ as a variable and see how the largest gap over all $p$ and the corresponding optimal $p$ changes with $\theta$. This is shown in Fig.~\ref{fig:1-qubit}.

\begin{figure}[H]
 \centering
 \includegraphics[width=0.4\textwidth]{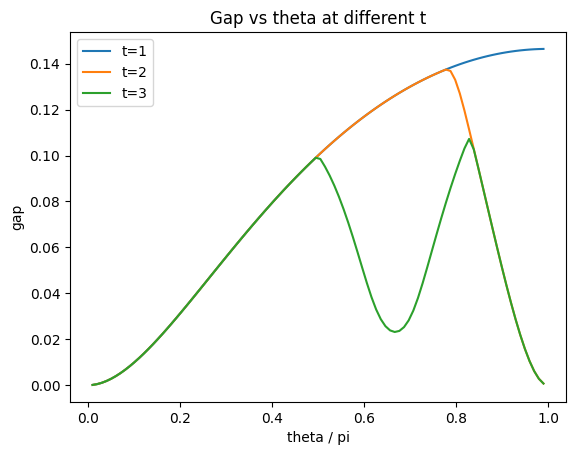}
 \includegraphics[width=0.4\textwidth]{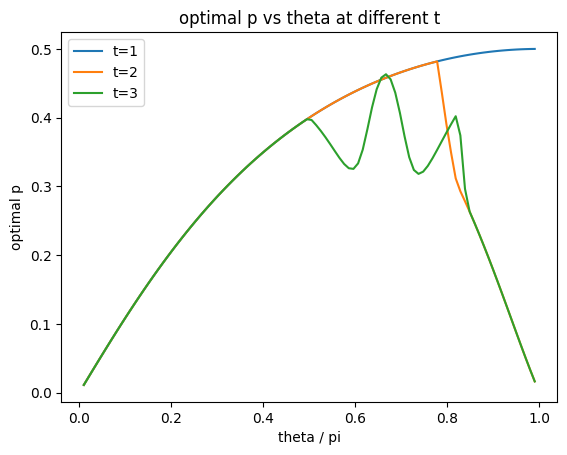}
 \caption{The largest gap over all $p$ and the corresponding optimal $p$ as a function of $\theta$. The data for $t=1,2,3$ are shown.}
 \label{fig:1-qubit}
\end{figure}

It is remarkable that the behaviors associated with the properties of this elementary model already exhibits highly nontrivial behaviors. For example, from Fig.~\ref{fig:1-qubit} we see that the optimal gate parameters can vary in an unintuitive manner for the generation of designs of different orders.

We can also use 1-qubit Haar random unitaries. For each value of $t$ we sample 50 sets of gates, each containing two Haar random unitaries. The average gap and the standard deviation over these samples can be found in Table~\ref{tab:1qubit}. The optimum gap from Hadamard and diagonal gates has also been listed for comparison. {One can see that the Hadamard + diagonal gates has a larger gap than the set of two typical Haar random gates (except for the case of $t=1$), but there still exists better choices than this gate set.}

\begin{table}[H]
 \centering
 \begin{tabular}{|c|c|c|c|c|}
  \hline
  $t$ & Mean Gap & Standard Deviation of Gap & Largest Gap & Previous Optimum \\
  \hline
  1 & 0.177569 & 0.139433 & 0.455504 & 0.146429 \\
  \hline
  2 & 0.0990918 & 0.0673355 & 0.238451 & 0.137656\\
  \hline
  3 & 0.0596353 & 0.0437368 & 0.169486 & 0.107030\\
  \hline
  4 & 0.0598917 & 0.0333851 & 0.118764 & N/A\\
  \hline
 \end{tabular}
 \caption{The mean, standard deviation and the largest gap among the samples. The ``Previous Optimum'' column refers to the largest gap obtained from Hadamard and diagonal gates, maximized over all possible $p$ and $\theta$. 
}
 \label{tab:1qubit}
\end{table}


\subsection{Two-qubit gates}
Now we extend the analysis to two qubits. For practical reason, we also consider the case when sampling one and two-qubit gates independently. Like in Section \ref{sec:one-qubit}, we here consider the gate set consisting of the $\theta$-phase gate, the Hadamard gate plus the $\CNOT$ gate as a natural example. Suppose there is $p_1/4$ probability for Hadamard to act on each of the qubits respectively. Similarly, there is $p_2/8$ probability for $\operatorname{diag}(1,e^{i\theta})$, and $p_2/8$ probability for $\operatorname{diag}(1,e^{-i\theta})$ on each qubit. Then we sample by $(1-p_1-p_2)/4$ probability for $\CNOT$ with each qubit as the control. Finally there is $1/2$ probability for identity. We show the gap and optimizing $p_1$ and $p_2$ in Fig.~\ref{fig:2-qubit}. {Even such a simple case could exhibit highly nontrivial behaviors.}

\begin{figure}[!ht]
	\centering
	\includegraphics[width=0.3\textwidth]{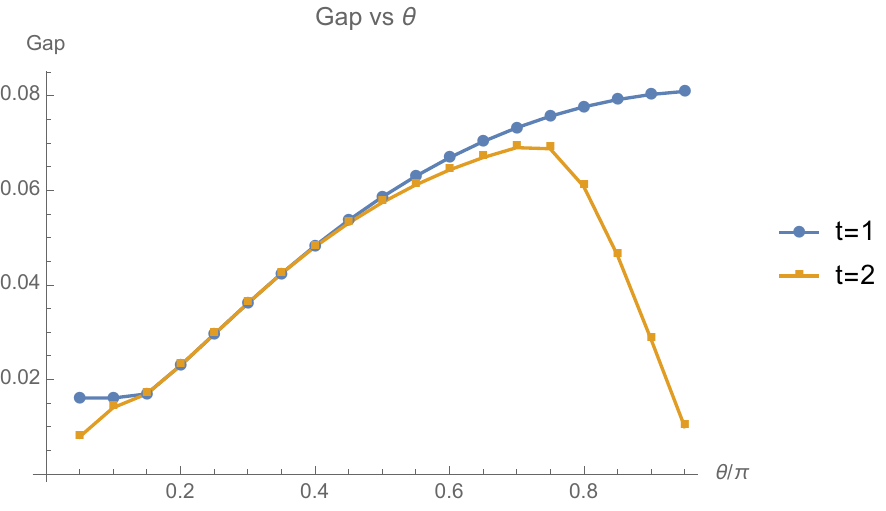}
	\includegraphics[width=0.3\textwidth]{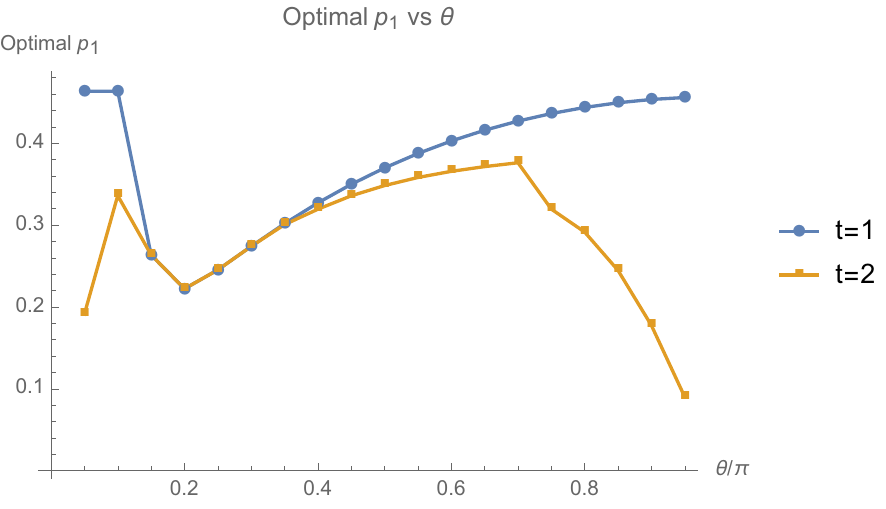}
	\includegraphics[width=0.3\textwidth]{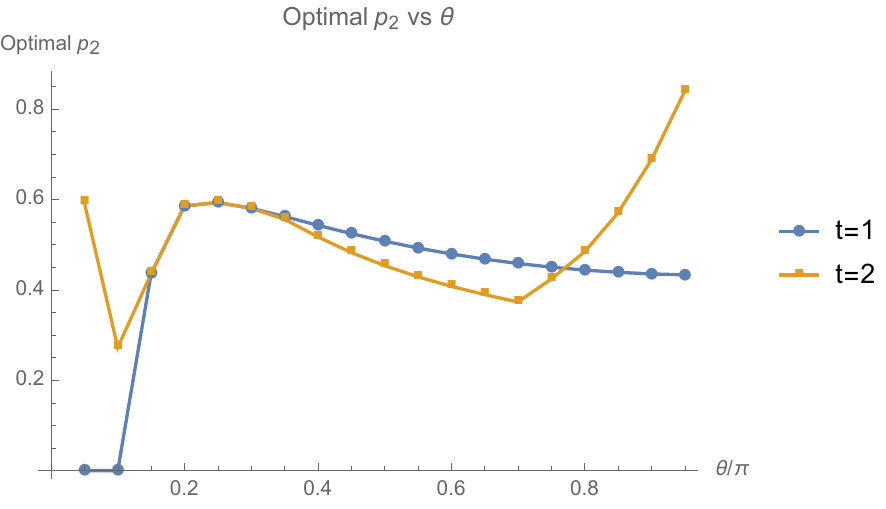}
	\label{fig:2-qubit}
	\caption{The largest gap over all $(p_1,p_2)$ and the corresponding optimal $(p_1,p_2)$ as a function of $\theta$. The data for $t=1,2$ are shown.}
\end{figure}

We could also consider using Haar-random unitaries as the single qubit gate set. 50 samples are run for $t=1$ and $t=2$ respectively. The data can be found in Table~\ref{tab:2qubit}. {We can see that there exists some gate sets that are much better than $\CNOT$ with Hadamard and diagonal gates.}

\begin{table}[!ht]
	\centering
	\begin{tabular}{|c|c|c|c|c|}
		\hline
		$t$ & Mean Gap & Standard Deviation of Gap & Largest Gap & Previous Optimum \\
		\hline
		1 & 0.119186 & 0.0581807 & 0.232504 & 0.0809092 \\
		\hline
		2 & 0.0717063 & 0.0309854 & 0.126088 & 0.0690015\\
		\hline
	\end{tabular}
	\caption{The mean, standard deviation and the largest gap among the samples. The ``Previous Optimum'' column refers to the largest gap obtained from Hadamard and diagonal gates, maximized over all possible $(p_1,p_2,\theta)$.}
	\label{tab:2qubit}
\end{table}

\section{Ironed gadget model and efficiency of 2-qubit gates}\label{sec:two-qubit}

Now we focus on the efficiency of different 2-qubit gates in the stochastic circuit models in forming approximate 2-designs. We are going to introduce the so-called \emph{ironed gadget model} which is defined by sampling 1-qubit Haar random gates with a fixed 2-qubit gate. It reflects the relative easiness of experimentally implementing any single-qubit gates and mathematically enable us to obtain numerous rigorous results summarized in Section \ref{sec:summary}.

We may interchangeably use the terms ``2-qubit'' and ``2-local'' when no ambiguity can arise. Circuits acting on many-body systems composed of multiple qubits are naturally associated with graph structures. We here consider circuit architectures built on a connected graph $\mathcal{G}$ composed of $n$ vertices, each of which represents a qubit, and edges $(i,j) \in E(\mathcal{G})$, each of which represents a 2-qubit gates. Accordingly, the stochastic quantum circuit models considered here are defined by 
\begin{enumerate}
	\item Sampling an edge from the given graph $\mathcal{G}$ uniformly;
	\item Sampling a 2-qubit gate from a prescribed gate ensemble and applying it on this edge.
\end{enumerate}

An illustration on a 5-qubit system with $\mathcal{G}$ be a 1D chain where we only sample nearest-neighbor 2-local gates is given in the following. For conciseness, we call this model a \emph{graph circuit}.  There are variant models, e.g. the brickwork model. In Section \ref{sec:architectures}, we also provide our theoretical and numerical analysis on these models.
\begin{align*}
	\begin{quantikz}
		& \gate[2][2cm]{U \sim \nu} & & & & \\
		&  & & \gate[2][2cm]{U \sim \nu} & & \\
		&  &  \gate[2][2cm]{U \sim \nu}  & & & \\
		&  &  & & \gate[2][2cm]{U \sim \nu}  & \\
		& & & & &
	\end{quantikz}
\end{align*}

As mentioned in Section \ref{sec:framework}, the moment operator of the ensemble $\mathcal{E}$ is supposed to be Hermitian in order to study its spectrum and hence the convergence time with no ambiguity. We emphasize here that there is no further requirements. For instance, the operator $T_t^{\mathcal{E}}$ needs \emph{not} being positive semidefinite. Making use of negative eigenvalues turns out to be one crucial step to find the ensemble converging towards 2-designs with the highest rate. 

\subsection{KAK decomposition and Weyl chamber}\label{sec:KAK}

Prior to the definition of gadgets of 2-local gates, we briefly review Cartan's KAK decomposition of 2-qubit operators~\cite{Zhang2003KAK,PhysRevLett.130.070601}:

\begin{theorem}[Cartan's KAK decomposition]\label{thm:KAK}
	Given any $U \in \U(4)$, there exist $A_1,A_2,B_1,B_2 \in \U(2)$, a triplet $(k_x, k_y, k_z)$ of real numbers such that
	\begin{align}
		U = (A_1 \otimes A_2) \exp[i(k_x X \otimes X + k_y Y \otimes Y + k_z Z \otimes Z)] (B_1 \otimes B_2),
	\end{align}
	where $XX,YY,ZZ$ are 2-qubit Pauli matrices with KAK coefficients $k_x, k_y, k_z \in \mathbb{R}$.  
\end{theorem}

Two unitaries $U, V \in \SU(4)$ are said to be \emph{KAK-equivalent} if $U = (R_1 \otimes R_2) V (S_1 \otimes S_2)$ for some $R_1,R_2,S_1,S_2 \in \U(2)$. Obviously, unitaries with the same KAK coefficients $k_x, k_y, k_z$ are KAK-equivalent. Moreover,

\begin{proposition}
\label{cor:KAK}
	Given arbitrary KAK coefficients $k_x, k_y, k_z$, the following holds:
	\begin{enumerate}	
		\item Shifting any coefficient by $\frac{\pi}{2}$, unitaries corresponding to these coefficients are still equivalent, e.g., $(k_x,k_y,k_z) \sim (k_x + \frac{\pi}{2},k_y,k_z)$.
		
		\item Reversing the sign any pair of the coefficients, the corresponding unitaries are still KAK-equivalent, e.g., $(k_x,k_y,k_z) \sim (-k_x,-k_y,k_z)$.
		
		\item Swapping any pair of the coefficients, the corresponding unitaries are still KAK-equivalent, e.g., $(k_x,k_y,k_z) \sim (k_y,k_x,k_z)$.
	\end{enumerate}
	Particularly, the KAK coefficients $k_x, k_y, k_z$ of any $U \in \SU(4)$ can be uniquely chosen if requiring:
	\begin{itemize}
		\item $\frac{\pi}{2} > k_x \ge k_y \ge k_z \ge 0$,
		\item $k_x + k_y \le \frac{\pi}{2}$,
		\item and if $k_z = 0$, then $k_x \le \frac{\pi}{4}$.
	\end{itemize}
\end{proposition}

Forgetting the third condition by allowing $k_x \le \frac{\pi}{2}$ when $k_z = 0$, these triplets of numbers compose a closed 3-dimensional tetrahedron called \emph{Weyl chamber} (see Fig.~\ref{fig:Weyl1}). As a reminder, despite the fact that the KAK coefficients can be uniquely selected with conditions given in Corollary \ref{cor:KAK}, the 1-qubit gates $A_1,A_2,B_1,B_2 \in \SU(2)$ from Theorem \ref{thm:KAK} are not unique determined in general. For instance,
\begin{align}
	(X \otimes X) \exp[i \frac{\pi}{4} Z \otimes Z] (X \otimes X) = \exp[i \frac{\pi}{4} ZZ].
\end{align}

\begin{figure}[h]
	\centering
	\subfigure[]{%
        \includegraphics[width=0.42\textwidth]{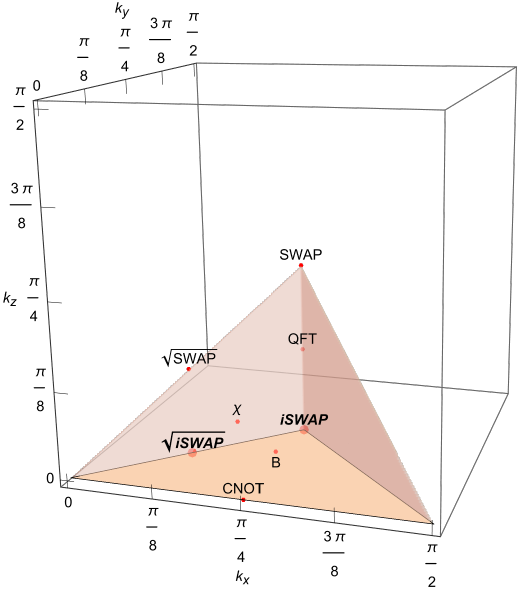} 
    }\hfill
    \subfigure[]{%
        \includegraphics[width=0.45\textwidth]{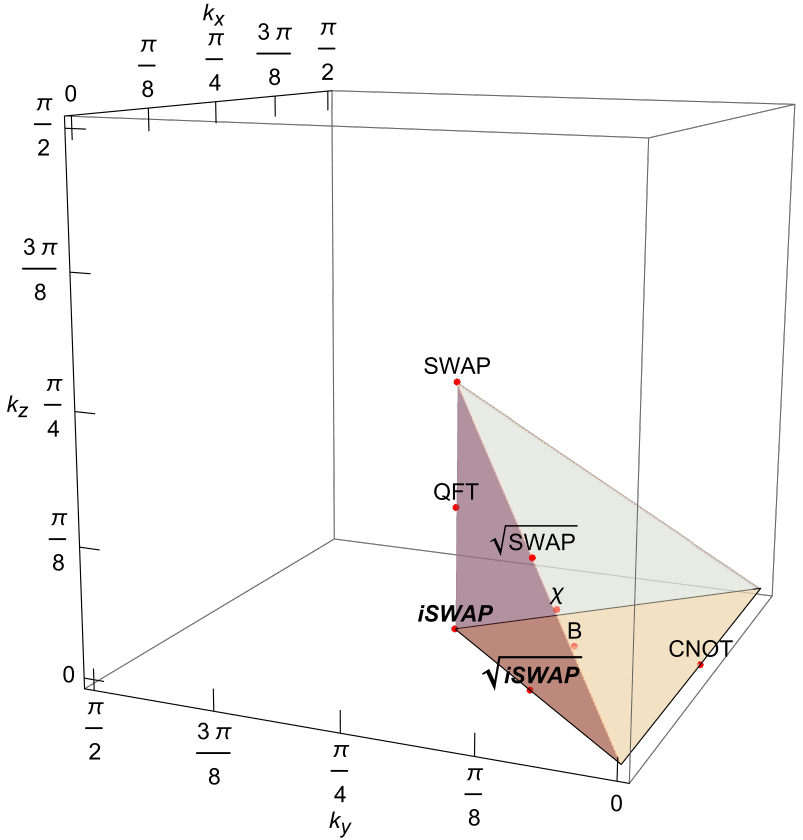}
        }
	\label{fig:Weyl1}
	\caption{Weyl chamber with typical gate families, the $\chi$ gate family is defined in Section \ref{sec:gadget}. Also see Table \ref{tab:Gadgets} for the explicit KAK coefficients and Fig.~\ref{fig:Weyl2}.  }
\end{figure}


\subsection{Ironed gadget model}\label{sec:gadget}

We now introduce the ironed gadget model, which is closely connected to the KAK decomposition. As discussed in the introduction, the conceptual motivation of this model lies with the low price of single-qubit gates.
In more detail, it is reasonable to mod out the effects of single-qubit gates when evalating and comparing gate and circuit efficiency, and furthermore, 
such model represent a rather practical scheme for implementing random entangling gates.

\begin{definition}[Ironed gadgets]\label{def:gadget}
	An \emph{ironed gadget} is a unitary ensemble acting on 2 qubits ($\U(4)$) generated by the following procedure: 
	\begin{enumerate}
		\item Single-qubit Haar-random gates acting independently on each qubit;
		\item A fixed 2-qubit unitary gate $U$;
		\item Single-qubit Haar-random gates acting independently on each qubit again.
	\end{enumerate}
	The following diagram sketches a single gadget with red boxes standing for 1-qubit Haar random gates and $U$ being the 2-qubit gate.
	\begin{align}
		\begin{quantikz}
			& \measure[style={fill=red!20}]{}\gategroup[wires=2,steps=3,style={dashed,rounded corners, inner sep=2pt}]{Ironed Gadget} & \gate[2]{U} & \measure[style={fill=red!20}]{} &  \\
			& \measure[style={fill=red!20}]{} &  & \measure[style={fill=red!20}]{} & 
		\end{quantikz}
	\end{align}
\end{definition}

Intuitively, these single-qubit gates are  ``average out'' so that the properties of the gadgets are solely determined by the 2-qubit gate. Diagrammatically, a stochastic quantum circuit over 1D chain can be depicted as:
\begin{align*}
\begin{quantikz}
	& \measure[style={fill=red!20}]{}\gategroup[wires=2,steps=3,style={dashed,rounded corners, inner sep=2pt}]{} & \gate[2]{U} & \measure[style={fill=red!20}]{}  & & & & & & & & & & \\
	& \measure[style={fill=red!20}]{} &  & \measure[style={fill=red!20}]{} & & & & \measure[style={fill=red!20}]{} \gategroup[wires=2,steps=3,style={dashed,rounded corners, inner sep=2pt}]{} & \gate[2]{U} & \measure[style={fill=red!20}]{} & & & & \\
	& &  & & \measure[style={fill=red!20}]{}\gategroup[wires=2,steps=3,style={dashed,rounded corners, inner sep=2pt}]{}  & \gate[2]{U} & \measure[style={fill=red!20}]{} &  \measure[style={fill=red!20}]{} & & \measure[style={fill=red!20}]{} & & & &  \\
	&  &  & & \measure[style={fill=red!20}]{}  & &  \measure[style={fill=red!20}]{} & & & & \measure[style={fill=red!20}]{}\gategroup[wires=2,steps=3,style={dashed,rounded corners, inner sep=2pt}]{} & \gate[2]{U} & \measure[style={fill=red!20}]{} & \\
	&  &  & &  &  & & & & &  \measure[style={fill=red!20}]{} & & \measure[style={fill=red!20}]{} & 
\end{quantikz}
\end{align*}

In an $n$-qubit system, let us consider a gadget having unitaries acting on the $i$- and $j$-th qubits. By changing the order of integrals, the above definition gives the following $t$-th moment operator of the gadget:
\begin{align}\label{eq:GagdetOperator}
	\begin{aligned}
		\T_{t,(i,j)}^{\Gadget} & = \int \Big( (A_i \otimes A_j) U_{i,j} (B_i \otimes B_j) \Big)^{\otimes t} \otimes \Big( (\bar{A}_i \otimes \bar{A}_j) \bar{U}_{i,j} (\bar{B}_i \otimes \bar{B}_j) \Big)^{\otimes t}  dA_i dA_j dB_i dB_j \\
		& = \T_{t,i}^{\U(2)} \T_{t,j}^{\U(2)} \Big( U_{i,j}^{\otimes t} \otimes \bar{U}_{i,j}^{\otimes t} \Big) \T_{t,i}^{\U(2)} \T_{t,j}^{\U(2)},
	\end{aligned}
\end{align}
where $U_{i,j}$ refers to applying the prescribed 2-local unitary $U$ on the $i$-th and $j$-th qubits while leaving other qubits intact. The operator $\T_{t,i}^{\U(2)}$ represents the 1-qubit Haar projector of the $i$-th qubit. Obviously, $\T_{t,(i,j)}^{\Gadget}, \T_{t,(r,s)}^{\Gadget}$ are similar matrices regardless of the sites they act on. The $t$-moment operator of the ensemble defined by sampling through the gadget acting on edges of a graph $\mathcal{G}$ is
\begin{align}
	\T_{t,\mathcal{G}}^{\Gadget} = \frac{1}{\vert E(\mathcal{G}) \vert} \sum_{(i,j) \in E(\mathcal{G})} \T_{t,(i,j)}^{\Gadget},
\end{align} 
where $\vert E(\mathcal{G}) \vert$ denotes the number of the edges of $\mathcal{G}$. 

\begin{proposition}\label{prop:Gadget}
	The following facts hold:
	\begin{enumerate}
		\item If $U, V$ are KAK-equivalent, the resultant gadgets share the same $t$-th moment operator. Actually, properties of the moment operator are uniquely determined by the KAK coefficients $k_x,k_y,k_z$ of $U$ in the gadget.
		
		\item When $t=2$, the operator $\T_{2,(i,j)}^{\Gadget}$  is always Hermitian.
	\end{enumerate}
\end{proposition}
\begin{proof}
	It is immediate to confirm the first statement by definition. The second one is proved later in Example \ref{example:LocalOperator}.
\end{proof}

When $t \geq 3$, $\T_{t,(i,j)}^{\Gadget}$ may not be Hermitian in general and we present an example in Example \ref{example:LocalOperator2}. In that case, we just adjust the second step in Definition \ref{def:gadget} by sampling from a 2-local unitary and its inverse with equal probability. Then Eq.~\eqref{eq:GagdetOperator} is replaced by
\begin{align}\label{eq:GagdetOperator2}
	\T_{t,(i,j)}^{\Gadget} & = \T_{t,i}^{\U(2)} \T_{t,j}^{\U(2)} \frac{1}{2}\Big( U_{i,j}^{\otimes t} \otimes \bar{U}_{i,j}^{\otimes t} + U_{i,j}^{\dagger \otimes t} \otimes \bar{U}_{i,j}^{\dagger \otimes t}\Big) \T_{t,i}^{\U(2)} \T_{t,j}^{\U(2)}.
\end{align}

In the following context, when we say a gadget of a certain gate, the gate usually means a gate family up to 1q gates and KAK equivalence. The operator $\T_{t,(i,j)}^{\Gadget}$ is referred to as a \emph{2-local $t$-th moment operator} of the gadget defined on site $i$ and $j$. If the 2-local unitary of the gadget or its KAK coefficients is known, the superscript of $\T_{t,(i,j)}^{\Gadget}$ will be replaced accordingly.


\subsection{Summary of key findings}\label{sec:summary}

With the framework formally defined, we first outline the key findings in this subsection, which provides a guide through the extensive body of detailed results and techniques that will be subsequently presented. As claimed in Proposition \ref{prop:Gadget}, the moment operators of gadgets are completely determined by the KAK coefficients, or the corresponding gate families. In what follows, we employ notations like $\T_{2,\mathcal{G}}^{\iSWAP}$, which means the second moment operator defined on the graph $\mathcal{G}$ using gadgets of $\iSWAP$ as gates family.

\begin{itemize}
    \item There exist specific 2-qubit gates whose associated ironed gadgets achieve \emph{exactly} the same second and third moment operators as the 2-qubit Haar measure. A solution is given by what we call the \emph{$\chi$ gate} family with KAK coefficients $(\frac{\pi}{4}-\arccos(\frac{1}{5})$, $\frac{\pi}{8}$, $\frac{1}{8}\arccos(\frac{1}{5}))$. 
    That is, exact 2-qubit 2- and 3-designs can be generated by applying this gate once, which is  evidently the most efficient. However, we also prove that no such solutions exists in the formation of 4-th or higher order exact designs (see Theorem \ref{prop:chi}).
    
    \item The $\iSWAP$ gate is particularly efficient among 2-qubit gates on many-body graph circuits. Specifically, we prove: 
    \begin{itemize}
    \item With respect to any connected graph $\mathcal{G}$, among 2-qubit gate families that occupy a large fraction of the Weyl chamber (satisfying a mild condition on the KAK coefficients illustrated in Fig.~\ref{fig:Weyl2}),  $\T_{2,\mathcal{G}}^{\iSWAP}$ attains the largest spectral gap. (see Theorem \ref{thm:iSWAP1} and Eq.~\eqref{eq:Region}).  We expect this fact to hold for all 2-qubit gates (see Section \ref{sec:conjectures}).

   \item Especially, the circuit given by $\iSWAP$ gadget + complete graph $K_n$ (all-to-all model) achieves the fastest convergence to unitary 2-designs among all possible graph circuits. Namely,
	\begin{align}
		\Delta( \T_{2,K_n}^{\iSWAP}) \geq \Delta( \T_{2,\mathcal{G}}^{\mathcal{E}}),
	\end{align}
	for any ensembles of 2-local unitaries $\mathcal{E}$ (not just gadgets) such that $\T_{2,\mathcal{G}}^{\mathcal{E}}$ is Hermitian (see Theorem \ref{thm:iSWAP2}).
	
	\item Numerical results for various standard graph topologies for small systems are provided (Table \ref{tab:differentGraphs} and \ref{tab:brickwork}), which enrich our theoretical results and reveal interesting finite-size effects. 
\end{itemize}

\item About different many-body circuit architectures (graphs):
\begin{itemize}
	\item Given a fixed ensemble of 2-local unitaries, when defined on the complete graph $K_n$, its moment operator achieves the largest spectral gap among all possible graphs. Particularly, for sufficient large $n$
	\begin{align}
		\Delta( \T_{t,P_n}^{\mathcal{E}} ) \geq \Delta( \T_{t,C_n}^{\mathcal{E}} )  \geq \Delta( \T_{t,K_n}^{\mathcal{E}} ),
	\end{align}
	where $P_n,C_n$ denote path (1D chain) and ring of $n$ vertices respectively (see Proposition \ref{prop:Kn}). 
	
	\item Given different gadgets, with respect to the same graph $\mathcal{G}$, their convergence speed can only differ by at most a constant. For example, the $t$-th moment operators of gadgets defined by 2-local gates: $\iSWAP$, $\B$, $\CNOT$, $\SQSW$, $\SQiSW$ and $\QFT$, can reach the spectral gap $\Theta(1/n)$ scaling with respect to $n$ on 1D chains $P_n$, rings $C_n$ as well as the complete graph $K_n$ (see the remark after Theorem \ref{thm:iSWAP1} and Section \ref{sec:graphs}). 
\end{itemize}

\item Assume the 1-qubit Haar random gates are implemented with negligible error, but there are deviations from the target 2-qubit gates, or simply the KAK coefficients. Let $\T_{t,\mathcal{G}}^{\Gadget}, \tilde{\T}_{t,\mathcal{G}}^{\Gadget}$ be the desired moment operator and the one obtained after deviations. Then
\begin{align}
	(1 - \epsilon) \Delta( \T_{t,\mathcal{G}}^{\Gadget}) \leq  \Delta( \tilde{\T}_{t,\mathcal{G}}^{\Gadget})
\end{align} 
where $\epsilon = O(\delta)$ with $\delta > 0$ stands for the largest possible derivation of the KAK coefficients. The error term $\epsilon$ is also independent of the choice of the graph $\mathcal{G}$ and the number $n$ of qubits in the system, but may depend on $t$ (Theorem \ref{prop:Perturb}). This ensures the robustness of our gadget models in converging to $t$-designs. Particularly, for 2-designs, we can improve the results with $\epsilon = O(\delta^2)$ for gadgets associated with $\CNOT$ or $\iSWAP$.

\item The 3-qubit $\Margolus$ gate outperforms several other well-known 3-qubit gates and 2-qubit gates in the generation of 2-designs (see Section \ref{sec:multiqubit}).

\item There is an upper bound on the spectral gap corresponding to generic $r$-local unitary ensembles built on hypergraphs (see Section \ref{sec:multiqubit}):
\begin{align}
	\Delta( \T_{t,\mathcal{G}}^{\mathcal{E}} ) \leq \frac{2r}{n}.
\end{align}
which can be obtained by using Weyl's inequality~\cite{Horn2017} in a straightforward manner.  
\end{itemize}


\subsection{Theoretical foundation: matrix representation of moment superoperators of gadgets}\label{sec:foundation}

To verify our results rigorously, we first illustrate how to represent the 2-local moment operator $\T_{t,(i,j)}^{\Gadget}$ as a matrix. This step is essential to conducting further theoretical computations. It should be noted that matrices representing $\T_{t,(i,j)}^{\Gadget}$ on different sites $i,j$ are similar to each other. For simplicity, we denote by $T_t^{\Gadget}$ and $T_t^{\mathcal{E}}$ the moment operators of a gadget or a generic ensemble respectively, just on a 2-qubit system. Only when we study the global (super-)moment operator $\T_{t,\mathcal{G}}^{\Gadget}$, we shall explicitly emphasize the interplay of different sites.

By definition, both $\T_{t,(i,j)}^{\Gadget}$ and $\T_{t,\mathcal{G}}^{\Gadget}$ are operators acting on 
\begin{align}
	\operatorname{End}(\mathcal{H}^{\otimes t}) \cong \mathcal{H}^{\otimes 2t} = ((\mathbb{C}^2)^{\otimes n})^{\otimes 2t}) \cong (\mathbb{C}^2)^{\otimes 2t})^{\otimes n} & \cong \Big( \operatorname{End}( (\mathbb{C}^2)^{\otimes t} ) \Big)^{\otimes n}.
\end{align}
Regardless of the large dimension of the space, the action of $\T_{t,(i,j)}^{\Gadget}$ is drastically simplified due to the underlying 1-qubit Haar projectors. According to our discussion in Section \ref{sec:framework}, let $u_{0,i},...,u_{D-1,i}$, with $D$ being given by Eq.~\eqref{eq:1qDimension}, denote an orthonormal basis spanning 
\begin{align}
	\text{Comm}_t( \U(2)) \subset (\mathbb{C}^2)^{\otimes 2t} \cong  \operatorname{End}( (\mathbb{C}^2)^{\otimes t} )
\end{align}
embedded in the $i$-th fold of $( \operatorname{End}( (\mathbb{C}^2)^{\otimes t} ) )^{\otimes n}$. Explicit forms of $u_{r,i}$ can be obtained by the permutation action of the symmetric group $S_t$ and Gram--Schmidt process or Weingarten calculus in general~\cite{Collins2003,RobertsChaos2017,Meckes2019}. Given the $i$- and $j$-th qubits of concern, vectors orthogonal to
\begin{align}\label{eq:basis}
	\mathrm{span}\{ u_{0,i},...,u_{D-1,i} \} \otimes \mathrm{span}\{ u_{0,j},...,u_{D-1,j} \}
\end{align}
are projected to zero under the action of $\T_{t,i}^{\U(2)} \T_{t,j}^{\U(2)}$. Therefore, by Eq.~\eqref{eq:GagdetOperator} and omitting the orders of the sites $i$ and $j$, we see
\begin{align}
	\T_{t,(i,j)}^{\Gadget} = T_t^{\Gadget} \otimes I_{\operatorname{End}( (\mathbb{C}^2)^{\otimes t})^{\otimes n-2}}
\end{align}
with $T_t^{\Gadget}$ been completely determined by its action on vectors from \eqref{eq:basis} tensored with the identity operator acting on $(\operatorname{End}( (\mathbb{C}^2)^{\otimes t})^{\otimes n-2}$ corresponding to left $n-2$ sites.

\begin{example}\label{example:LocalOperator}
	When $t = 2$, $D = 2$ and $\text{Comm}_2( \U(2))$ can be spanned by the identity map and the transposition from $S_2$:
	\begin{align}
		I = \sum_{a,b} E\indices{^a_a} E\indices{^b_b}, \quad S = \sum_{a,b} E\indices{^a_b} E\indices{^b_a} = \frac{1}{2} \sum_P P \otimes P^\dagger,
	\end{align}
	where $E\indices{^a_b} \in \operatorname{End}( \mathbb{C}^2 )$ is a matrix unit and $P = I,X,Y,Z$ are Pauli matrices. As a result, we can construct
	\begin{align}\label{eq:1qBasis2}
		u_0 = \frac{1}{2} I,  \quad  u_1 = \frac{1}{\sqrt{3}}(S - \frac{1}{2} I)
	\end{align}
	by Gram--Schmidt orthogonalization with respect to the Hilbert-Schmidt inner product defined by trace. As discussed earlier, the matrix representation of $\T_{2,(i,j)}^{\Gadget}$ is given by
	\begin{align}
		\T_{2,(i,j)}^{\Gadget} = T_2^{\Gadget} \otimes I \vert_{( \operatorname{End}( (\mathbb{C}^2)^{\otimes 2} ) )^{\otimes n-2} }
	\end{align}
	up to the orders of sites $i$ and $j$. The operator $T_2^{\Gadget}$ can be represented by a $4 \times 4$ matrix under the following basis:
	\begin{align}
		& u_0 \otimes u_0 = \frac{1}{4} II, \label{eq:2qBasisA} \\
		& u_0 \otimes u_1 = \frac{1}{2\sqrt{3}} \Big( IS - \frac{1}{2}II \Big), \label{eq:2qBasisB} \\
		& u_1 \otimes u_0 = \frac{1}{2\sqrt{3}} \Big( SI - \frac{1}{2}II \Big), \label{eq:2qBasisC} \\
		& u_1 \otimes u_1 = \frac{1}{3} \Big( SS - \frac{1}{2}(SI + IS) + \frac{1}{4}II \Big), \label{eq:2qBasisD}
	\end{align}
	where the coefficients appear to normalize the basis elements.

	Specifically, suppose the gadget is defined with
	\begin{align}\label{eq:KAK}
		\begin{aligned}
			U & \sim \exp[i(k_x XX + k_y YY + k_z ZZ)] \\
			& = (\cos k_x II + i \sin k_x XX) (\cos k_y II + i \sin k_y YY)  (\cos k_z II + i \sin k_z ZZ).
		\end{aligned}
	\end{align}
	 Then 
	\begin{align} \label{eq:MatrixEntries}
		\langle u_{s_1} \otimes u_{r_1},  T_2^{\Gadget} (u_{s_2} \otimes u_{r_2})  \rangle = \tr \Big( (u_{s_1} \otimes u_{r_1})^\dagger U_{i,j}^{\otimes 2} (u_{s_2} \otimes u_{r_2}) U_{i,j}^{\dagger \otimes 2} \Big) 
	\end{align}
	can be computed analytically, which yields:
	\begin{align}\label{eq:Gadget4x4}
		T_2^{\Gadget} = \begin{pmatrix}
			1 & 0 & 0 & 0 \\
			0 & 1 - c -3b & c & \sqrt{3}b \\
			0 & c & 1 - c -3b & \sqrt{3}b \\
			0 & \sqrt{3}b & \sqrt{3}b & a
		\end{pmatrix},
	\end{align}
	with
	\begin{align}\label{eq:abc}
		a = & \frac{1}{9} \Big(6 + \cos(4 k_x) \cos(4 k_y) + \cos(4 k_x)\cos(4 k_z) + \cos(4 k_y)\cos(4 k_z) \Big), \\
		b = & \frac{1}{18} \Big( 3 - \cos(4 k_x) \cos(4 k_y) - \cos(4 k_y) \cos(4 k_z) - \cos(4 k_z) \cos(4 k_x) \Big) = \frac{1}{2}(1-a), \\
		c = & \frac{1}{12} \Big( 3 + \cos(4 k_x) \cos(4 k_y) + \cos(4 k_y) \cos(4 k_z) + \cos(4 k_z) \cos(4 k_x) \notag \\
		& \hspace{6cm} - 2 (\cos(4 k_x) + \cos(4 k_y) + \cos(4 k_z) )  \Big). 
	\end{align} 
	 It is straightforward to see that $u_0 \otimes u_0$ is just the identity matrix (with normalization). It is thus a unit eigenvector of both the 1-qubit Haar projector and $U_{i,j} \otimes \bar{U}_{i,j}$. However, other basis elements including cross terms between $I$ and $S$ and cannot be an unit eigenvector. 
	    
	As a side note, there is another canonical way to construct an orthonormal basis for $\text{Comm}_2(\U(2))$:
	\begin{align}\label{eq:1qBasis1}
		u_0' = \Pi_+ = \frac{I + S}{2}, \quad u_1' = \Pi_- = \frac{I - S}{2}.
	\end{align}
	However, none of the four tensor products $u_r' \otimes u_s'$ is a unit eigenvector of the operator $T_2^{\Gadget}$. It will gradually become clear that this choice of basis is inconvenient, comparing to Eq.~\eqref{eq:1qBasis2}, when studying the spectral gap of $\T_{t,\mathcal{G}}^{\Gadget}$ as a summation of $\T_{t,(i,j)}^{\Gadget}$ on edges.
\end{example}

\begin{example}\label{example:LocalOperator2}
	When $t = 3$, according to Eq.~\eqref{eq:1qDimension} $D = 5$. The 5 orthonormal basis elements $\{ u_0,...,u_4 \}$ of $\text{Comm}_2( \U(2))$ can be obtained by the Gram--Schmidt orthogonalization of the $S_3$ actions on $\operatorname{End}( (\mathbb{C}^2)^{\otimes 3} )$. By taking tensor product, $\T_{3,(i,j)}^{\Gadget}$ is completely determined by its action on the 25-dimensional subspace $\mathrm{span}\{ u_0,...,u_4 \}^{\otimes 2}$. Again, we suppose $U_{i,j}$ is given by \eqref{eq:KAK} and then compute
	\begin{align}
		\tr \Big( (u_{s_1} \otimes u_{r_1})^\dagger U_{i,j}^{\otimes 3} (u_{s_2} \otimes u_{r_2}) U_{i,j}^{\dagger \otimes 3} \Big)
	\end{align}
	numerically. Not like the case when $t = 2$, the corresponding matrix is not always Hermitian in general. For instance, when $k_x = k_y = k_z = \frac{\pi}{8}$ which corresponds to $\SQSW$, the $25 \times 25$ matrix representing $T_3^{\SQSW}$ is not Hermitian. Therefore, when $t\geq 3$ we consider the model in which we sample the 2-qubit unitary and its inverse with equal probability.
\end{example}


\subsection{Independent gadgets {and the $\chi$ gate family}}\label{sec:autoconvolution}

Interestingly, there is a solution $(k_x, k_y, k_z) = (\frac{\pi}{4} - \frac{1}{8}\arccos(\frac{1}{5}), \frac{\pi}{8}, \frac{1}{8}\arccos(\frac{1}{5}))$ to the KAK coefficients, corresponding to what we call the \emph{$\chi$ gate family}, for which Eq.~\eqref{eq:Gadget4x4} becomes
\begin{align}\label{eq:G-U(4)}
	T_2^{(\frac{\pi}{4} - \frac{1}{8}\arccos(\frac{1}{5}), \frac{\pi}{8}, \frac{1}{8}\arccos(\frac{1}{5}))} = T_2^\chi = \begin{pmatrix}
		1 & 0 & 0 & 0 \\
		0 & \frac{1}{5} & \frac{1}{5} & \frac{\sqrt{3}}{5} \\
		0 & \frac{1}{5} & \frac{1}{5} & \frac{\sqrt{3}}{5} \\
		0 & \frac{\sqrt{3}}{5}  & \frac{\sqrt{3}}{5}  & \frac{3}{5} 
	\end{pmatrix}
\end{align}
with eigenvalues being $(1,1,0,0)$. By definition, the corresponding second moment operator $\T_{2,(i,j)}^{\chi}$ matches the 2-local Haar projector $\T_{2,(i,j)}^{\U(4)}$ exactly on any $n$-qubit system, even though we do not sample over the whole group $\U(4)$ in defining the gadget. Numerical computation indicates that their third moments also matches, but it ceases to hold for higher moments. In fact, we can show that there is no such solution  for 4-designs and above:
\begin{theorem}
\label{prop:chi}
	When $t \geq 4$, there is no solution of the KAK coefficients such that $T_{t}^{\Gadget} = T_{t}^{\U(4)}$. As a result, $\T_{t,(i,j)}^{\Gadget} \neq \T_{t,(i,j)}^{\U(4)}$ for any pair of sites in a general $n$-qubit system. 
\end{theorem}
\begin{proof}
	We only need to prove that there is no such solutions when $t = 4$. In order to match the Haar projector $T_{4}^{\U(4)}$, $T_{4}^{\Gadget}$ itself has to be a projector of rank $4! = 24$. Our strategy is to prove that for any KAK coefficients, the rank of $T_{4}^{\Gadget}$ is always larger than $24$.
	
	As introduced in Section \ref{sec:foundation}, the matrix representation of $T_4^{\Gadget}$ can be determined by the orthonormal bases $\{u_{0,1},...,u_{D-1,1}\}, \{u_{0,2},...,u_{D-1,2}\}$ ($1,2$ are labels of the two qubits where we will omit later) that span the unit eigenspaces of $\T_{4,1}^{\U(2)}, \T_{4,2}^{\U(2)}$ respectively. By Eq.~\eqref{eq:1qDimension}, $D = 14$ here. However, we do not work with these orthonormal bases because their tensor products $u_{r} \otimes u_{s}$ are generally not unit eigenvectors of $T_4^{\Gadget}$, except the trivial case when $r = s = 0$ corresponding to the identity matrix.
	
	Instead, we take the complete spanning set $\{v_0,...,v_{23}\}, \{v_{0},...,v_{23}\}$ given by permutations from $S_4$ acting on $(\mathbb{C}^2)^{\otimes 4}$. Then by definition, $v_{r} \otimes v_{r}; r=0,...,23$ are still permutations from $S_4$ but acting on $((\mathbb{C}^2)^{\otimes 2})^{\otimes 4}$ and hence unit eigenvectors of $T_4^{\Gadget}$. After orthogonalize all these tensor products 
	\begin{align}
		\{ v_r \otimes v_r \}_{r = 0}^{23} \cup \{ v_r \otimes v_s \}_{r \neq s},
	\end{align}
	we are left with $14^2 = 196$ orthonormal basis elements, denoted by $\{b_\alpha\}_{\alpha = 0}^{195}$, which completely determine the matrix representation of $\T_4^{\Gadget}$.
	
	It is infeasible to compute this $196 \times 196$ matrix analytically (cf. Eq.~\eqref{eq:Gadget4x4}). However, we only need to find a few number of vectors $b_\alpha$ with $\alpha \geq 23$ and prove that $T_{ij} = \langle b_{\alpha_i}, T_4^{\Gadget} b_{\alpha_j} \rangle \neq 0$, it would be sufficient to reach the conclusion. To see the reason, suppose we take $b_{\alpha_1}, b_{\alpha_2},  b_{\alpha_3}$ for some $\alpha_1,\alpha_2,\alpha_3 \geq 23$, the sub-matrix representation of $T_4^{\Gadget}$ restricted to these vectors and the unit eigenvectors is
	\begin{align}
		\begin{pmatrix}
			1 & 0 & \cdots & 0 & 0 & 0 & 0  \\
			0 & 1 &  \cdots & 0 & 0 & 0 & 0  \\
			\vdots & \vdots & \ddots & \vdots & \vdots & \vdots  & \vdots \\
			0 & 0 &  \cdots & 1 & 0 & 0 & 0 \\
			0 & 0 & \cdots & 0 & T_{11} & T_{12} & T_{13} \\
			0 & 0 & \cdots & 0 & T_{21} & T_{22} & T_{23} \\
			0 & 0 & \cdots & 0 & T_{31} & T_{32} & T_{33}		
		\end{pmatrix}.
	\end{align}
	As long as the right bottom three diagonal elements do not vanish simultaneously, the rank of this restricted matrix, and thus $T_4^{\Gadget}$, must be larger than $24$. Note that with fewer extra basis elements, one can always find certain KAK coefficients making the diagonal entries vanish simultaneously.
	
	Specifically, let
	\begin{align}
		b_{\alpha_1}' = (14)(23) \otimes (1423), \quad b_{\alpha_2}’ = (14)(23) \otimes (132), \quad b_{\alpha_3}' = (132) \otimes (124),  
	\end{align}
	where each permutation acts on the 4-fold tensor product $(\mathbb{C}^2)^{\otimes 4}$ of the 1-qubit Hilbert space. By a Gram--Schmidt process, we orthogonalize them with the first 24 orthonormal basis elements to define $b_{\alpha_1},  b_{\alpha_2},  b_{\alpha_3}$. Obviously, 
	\begin{align}
		\langle b_{\alpha_k}, \T_{4,(i,j)}^{\Gadget} b_{\alpha_k} \rangle = \frac{1}{2} \langle b_{\alpha_k},  U^{\otimes 4} \otimes \bar{U}^{\otimes 4}  + U^{\dagger \otimes 4} \otimes \bar{U}^{\dagger \otimes 4}  b_{\alpha_k} \rangle.
	\end{align}
	Let
	\begin{align}
		f_1(k_x,k_y,k_z) = & \cos(4 k_x)(2 + \cos(4 k_y)) + \cos(4 k_y)(2 + \cos(4 k_z)) + \cos(4 k_z)(2 + \cos(4 k_x) ), \\
		f_2(k_x,k_y,k_z) = & \cos(4 k_x) \cos(4 k_y) \cos(4 k_z) - ( \cos(4 k_x) + \cos(4 k_y) + \cos(4 k_z) ), \\
		f_3(k_x,k_y,k_z) = & \cos(4 k_x) + \cos(4 k_y) + \cos(4 k_z), \\
		f_4(k_x,k_y,k_z) = & \cos^2(4 k_x) \cos^2(4 k_y) + \cos^2(4 k_y) \cos^2(4 k_z) + \cos^2(4 k_z) \cos^2(4 k_x), \\
		f_5(k_x,k_y,k_z) = & \cos^2(4 k_x) (1 + \cos(4 k_y) + \cos(4 k_z)) +  \cos^2(4 k_y) (\cos(4 k_x) + 1 + \cos(4 k_z) ) \notag \\
		& +  \cos^2(4 k_z) (\cos(4 k_x) + \cos(4 k_y) + 1).
	\end{align}
	Then analytically, we can compute
		Then analytically, we can compute
	\begin{align}
		T_{11} = & \frac{1}{48} (3 + 5 f_1), \\
		T_{12} = & T_{21} = T_{13} = T_{31} = 0, \\ 
		T_{22} = & \frac{1}{43} \Big( 3 + 5 f_1 + \frac{5}{2} f_2 \Big), \\
		T_{23} = & T_{32} = \frac{55}{129} \sqrt{\frac{7}{53337}} \Big( 3 + 5 f_1 + \frac{6}{11} f_2 - \frac{172}{11} f_3  \Big), \\
		T_{33} = & \frac{1}{110087568} \Big( 11609360 f_1 + 6226080 f_2 - 2841440 f_3 + 1035440 f_2 f_3 + 1035440 f_3^2 \notag  \\
		& + 3282408 + 517720 f_4 + 2070880 f_5  \Big).
	\end{align}
	We are going to show that it is impossible to find $k_x,k_y,k_y$ making all these entries being zero. Despite the first glance, it is quite simple to see that 
	\begin{align}
		T_{11} = T_{22} = T_{23} = 0 \Leftrightarrow f_1 = -\frac{3}{5}, f_2 = f_3 = 0
	\end{align}
	which gives 
	\begin{align}
		\cos(4 k_x) = \sqrt{\frac{3}{5}}, \quad \cos(4 k_y) = -\sqrt{\frac{3}{5}}, \quad \cos(4 k_z) = 0
	\end{align}
	and there are five extra ways to permute the solutions. Substituting them into $T_{33}$, it is nonzero and hence the rank of $\T_{4,(i,j)}^{\Gadget}$ should always be larger than $24$, which verifies the claim. 
\end{proof}

Substituting different KAK coefficients into Eq.~\eqref{eq:Gadget4x4}, we can write down the corresponding matrix representations of $T_2^{\Gadget}$ with eigenvalues as in Tab.~\ref{tab:Gadgets}. 

\renewcommand\arraystretch{1.3}
\newcommand\scalemath[2]{\scalebox{#1}{\mbox{\ensuremath{\displaystyle #2}}}} 
\begin{table}[]
	\centering
	\resizebox{\columnwidth}{!}{%
		\begin{tabular}{c|c|c|c|c}
			\hline\hline
			Gates family & KAK coefficients & Matrix representation & Spectrum & Spectral gap \\ 
			\hline\hline
			$\SWAP$ & $(\frac{\pi}{4},\frac{\pi}{4},\frac{\pi}{4})$ &  $T_2^{\text{SWAP}} = \resizebox{7em}{3em}{ $\begin{pmatrix}
					1 & 0 & 0 & 0 \\
					0 & 0 & 1 & 0 \\
					0 & 1 & 0 & 0 \\
					0 & 0 & 0 & 1
				\end{pmatrix}$ }$  & $(1,1,1,-1)$ & $0$ \\
			\hline
			$\chi$ & $(\frac{\pi}{4} - \frac{1}{8}\arccos(\frac{1}{5}), \frac{\pi}{8}, \frac{1}{8}\arccos(\frac{1}{5}))$ & $T_2^\chi = \resizebox{7em}{3em}{ $\begin{pmatrix}
				1 & 0 & 0 & 0 \\
				0 & \frac{1}{5} & \frac{1}{5} & \frac{\sqrt{3}}{5} \\
				0 & \frac{1}{5} & \frac{1}{5} & \frac{\sqrt{3}}{5} \\
				0 & \frac{\sqrt{3}}{5}  & \frac{\sqrt{3}}{5}  & \frac{3}{5} 
			\end{pmatrix}$ }$ & $(1,1,0,0)$ & $1$ \\
			\hline
			$\QFT$ & $(\frac{\pi}{4},\frac{\pi}{4},\frac{\pi}{8})$ & $T_2^{\QFT} = \resizebox{7em}{3em}{ $\begin{pmatrix}
				1 & 0 & 0 & 0 \\
				0 & 0 & \frac{2}{3} & \frac{\sqrt{3}}{9} \\
				0 & \frac{2}{3} & 0 & \frac{\sqrt{3}}{9} \\
				0 & \frac{\sqrt{3}}{9}  & \frac{\sqrt{3}}{9}  & \frac{7}{9} 
			\end{pmatrix}$ }$ &  $(1,1,\frac{4}{9},-\frac{2}{3})$ & $\frac{1}{3}$ \\
			\hline 
			$\SQSW$ & $(\frac{\pi}{8},\frac{\pi}{8},\frac{\pi}{8})$ & $T_2^{\SQSW} = \resizebox{7em}{3em}{ $\begin{pmatrix}
				1 & 0 & 0 & 0 \\
				0 & \frac{1}{4} & \frac{1}{4} & \frac{\sqrt{3}}{6} \\
				0 & \frac{1}{4} & \frac{1}{4} & \frac{\sqrt{3}}{6} \\
				0 & \frac{\sqrt{3}}{6}  & \frac{\sqrt{3}}{6}  & \frac{2}{3} 
			\end{pmatrix}$ }$ & $(1,1,\frac{1}{6},0)$ & $\frac{5}{6}$ \\
			\hline
			$\iSWAP$ & $(\frac{\pi}{4},\frac{\pi}{4},0)$ & $T_2^{\iSWAP} = \resizebox{7em}{3em}{ $\begin{pmatrix}
				1 & 0 & 0 & 0 \\
				0 & 0 & \frac{1}{3} & \frac{2\sqrt{3}}{9} \\
				0 & \frac{1}{3} & 0 & \frac{2\sqrt{3}}{9} \\
				0 & \frac{2\sqrt{3}}{9}  & \frac{2\sqrt{3}}{9}  & \frac{5}{9} 
			\end{pmatrix}$ }$ & $(1,1,-\frac{1}{9},-\frac{1}{3})$ & $\frac{2}{3}$ \\
			\hline
			$\B$ & $(\frac{\pi}{4},\frac{\pi}{8},0)$ & $T_2^{\B} = \resizebox{7em}{3em}{ $\begin{pmatrix}
				1 & 0 & 0 & 0 \\
				0 & \frac{1}{6} & \frac{1}{6} & \frac{2\sqrt{3}}{9} \\
				0 & \frac{1}{6} & \frac{1}{6} & \frac{2\sqrt{3}}{9} \\
				0 & \frac{2\sqrt{3}}{9}  & \frac{2\sqrt{3}}{9}  & \frac{5}{9} 
			\end{pmatrix}$ }$ & $(1,1,0,-\frac{1}{9})$ & $\frac{8}{9}$ \\
			\hline 
			$\SQiSW$ & $(\frac{\pi}{8},\frac{\pi}{8},0)$ & $T_2^{\SQiSW} = \resizebox{7em}{3em}{ $\begin{pmatrix}
				1 & 0 & 0 & 0 \\
				0 & \frac{5}{12} & \frac{1}{12} & \frac{\sqrt{3}}{6} \\
				0 & \frac{1}{12} & \frac{5}{12} & \frac{\sqrt{3}}{6} \\
				0 & \frac{\sqrt{3}}{6}  & \frac{\sqrt{3}}{6}  & \frac{2}{3} 
			\end{pmatrix}$ }$ & $(1,1,\frac{1}{3},\frac{1}{6})$ & $\frac{2}{3}$ \\
			\hline
			$\CNOT$ & $(\frac{\pi}{4},0,0)$ & $T_2^{\CNOT} = \resizebox{7em}{3em}{ $\begin{pmatrix}
				1 & 0 & 0 & 0 \\
				0 & \frac{1}{3} & 0 & \frac{2\sqrt{3}}{9} \\
				0 & 0 & \frac{1}{3} & \frac{2\sqrt{3}}{9} \\
				0 & \frac{2\sqrt{3}}{9}  & \frac{2\sqrt{3}}{9}  & \frac{5}{9} 
			\end{pmatrix}$ }$ & $(1,1,\frac{1}{3},-\frac{1}{9})$ & $\frac{2}{3}$ \\
			\hline\hline
		\end{tabular}
	}
	\caption{Matrix representations of $T_2^{\Gadget}$ for a variety of notable 2-local gates. The spectral gap is defined by $\Delta( T_t^{\Gadget}) = 1 - \max \{\lambda_2(\T_t^{\Gadget}), \vert \lambda_{\min}(\T_t^{\Gadget}) \vert \} $ (also see Eq.~\eqref{eq:Gap} for more details).}
	\label{tab:Gadgets}
\end{table}

Note that among all 2-local gates families, only the $\SWAP$ gate is not universal together with single-qubit gates. As shown in Table \ref{tab:Gadgets}, $T_2^{\SWAP}$ has three unit eigenvalues and can never converge to 2-designs. 


\subsection{Efficiency of convergence to 2-designs on many-body circuits}\label{sec:efficiency}

We now systematically study the convergence efficiency of gadget models on general $n$-body many-body systems under arbitrary graphs using a variety of methods.
A key conclusion is that the gadget of $\iSWAP$ exhibits particularly high efficiency in generating 2-designs, provably outperforming a majority of other 2-qubit gates on any circuit graph (Theorem \ref{thm:iSWAP1}) and indeed, expected to achieve the highest efficiency among all 2-qubit gates, for which we provide numerical evidence. Furthermore, when $n \geq 3$, over all possible 2-local unitary circuit ensembles, the $\iSWAP$ gadgets acting on the complete graph (with all-to-all interaction) attain the fastest convergence to  2-designs (Theorem \ref{thm:iSWAP2}). 

Before presenting the results, we should make it clear and rigorous here that when studying the spectral gap, there is in general no need to care about the influence of negative eigenvalues (see Eq.~\eqref{eq:convergence}). Many examples in Table \ref{tab:Gadgets} involve negative eigenvalues. However, after summing over the graph
\begin{align}
	\T_{t,\mathcal{G}}^{\mathcal{E}} = \frac{1}{\vert E(\mathcal{G}) \vert} \sum_{(i,j) \in E(\mathcal{G})} \T_{t,(i,j)}^{\mathcal{E}}
\end{align}
the absolute value of the smallest eigenvalue cannot be larger than second largest eigenvalue for sufficiently large system size $n$. Recall that we denote by $T_t^{\mathcal{E}}$ the moment operator simply defined on a 2-qubit system as each local term $\T_{t,(i,j)}^{\mathcal{E}}$ are similar to each other on different edges.

\begin{lemma}\label{lemma:Eigenvalues1}
	Let $\mathcal{E}$ be any ensemble of 2-local unitaries such that $T_t^{\mathcal{E}}$ is Hermitian. When 
	 \begin{align}
		n \geq 2\frac{1 - \lambda_{\min}(T_t^{\mathcal{E}}) }{1 + \lambda_{\min}(T_t^{\mathcal{E}}) },
	\end{align}
	the second largest eigenvalue of $\T_{t,\mathcal{G}}^{\mathcal{E}}$ is larger than the absolute value of its smallest eigenvalue $\lambda_{\min}(\T_{t,\mathcal{G}}^{\mathcal{E}})$. That is,
	\begin{align}
		\lambda_2( \T_{t,\mathcal{G}}^{\mathcal{E}} ) \geq \vert \lambda_{\min}( \T_{t,\mathcal{G}}^{\mathcal{E}} ) \vert.
	\end{align}
	In other words, the spectral gap is solely determined by the second largest eigenvalue:
	\begin{align}
		\Delta( \T_t^{\mathcal{E}}) = 1 - \max \{\lambda_2(\T_t^{\mathcal{E}}), \vert \lambda_{\min}(\T_t^{\mathcal{E}}) \vert \} = 1 - \lambda_2(\T_t^{\mathcal{E}})
	\end{align}
	for sufficiently large $n$.
\end{lemma}
\begin{proof}
	The proof is straightforward by applying Weyl's inequality~\cite{Horn2017}. Let us list eigenvalues of an operator $M$ acting on an $N$-dimensional space as 
	\begin{align}
		\lambda_{\min}(M) = \lambda_1(M) \leq \cdots \leq \lambda_{N-1}(M) \leq \lambda_N(M).
	\end{align}
	They are different from our previous labels but will only be used when applying Weyl's inequality. We first lower bound the smallest eigenvalue of $\T_{t,\mathcal{G}}^{\mathcal{E}}$:
	\begin{align}
		\lambda_{\min} (\T_{t,\mathcal{G}}^{\mathcal{E}}) \geq \frac{1}{\vert E(\mathcal{G}) \vert} \sum_{(i,j) \in E(\mathcal{G})} \lambda_{\min}( \T_{t,(i,j)}^{\mathcal{E}} ) = \lambda_{\min}( \T_{t,(i,j)}^{\mathcal{E}} )
	\end{align}
	because the 2-local moment operators acting on different sites are assumed to be similar and share the same spectrum.
	
	Let us take the vertex $v$ from the graph $\mathcal{G}$ having the smallest degree. Let $N(v)$ denote edges that joint $v$. By definition, 
	\begin{align}\label{eq:handshake}
		\vert E(\mathcal{G}) \vert \geq \frac{1}{2}n \vert N(v) \vert.
	\end{align}
	The second largest eigenvalue of $\T_{t,\mathcal{G}}^{\mathcal{E}}$ is now labeled by $N-t!$ due to the fact that we have $t!$ unit eigenvalues. Then
	\begin{align}
		\lambda_{N-t!} (\T_{t,\mathcal{G}}^{\mathcal{E}}) \geq \frac{1}{\vert E(\mathcal{G}) \vert} \Big[ \lambda_{N-t!} \Big( \sum_{(i,j) \in E(\mathcal{G}) \setminus N(v) } \T_{t,(i,j)}^{\mathcal{E}} \Big) + \lambda_{\min}\Big( \sum_{k,l \in N(v)} \T_{t,(k,l)}^{\mathcal{E}} \Big) \Big].
	\end{align}
	It should be noted that there are several copies of the largest eigenvalue, equal to $\vert E(\mathcal{G}) \setminus N(v) \vert$ for $\sum_{E(\mathcal{G}) \setminus N(v) } \T_{t,(i,j)}^{\mathcal{E}}$, as its summation is taken over a subgraph with the neighborhood of one site totally deleted. Therefore, by \eqref{eq:handshake},
	\begin{align}
		\begin{aligned}
			\lambda_{N-t!}(\T_{t,\mathcal{G}}^{\mathcal{E}}) + \lambda_{\min}(\T_{t,\mathcal{G}}^{\mathcal{E}}) 
			\geq & \frac{1}{\vert E(\mathcal{G}) \vert} \Big[ \Big( \vert E(\mathcal{G}) \vert - \vert N(v) \vert \Big) +  \Big( \vert E(\mathcal{G}) \vert + \vert N(v) \vert \Big)\lambda_{\min}(\T_{t,(i,j)}^{\mathcal{E}}) \Big] \\
			\geq & \frac{\vert N(v) \vert}{\vert E(\mathcal{G}) \vert} \Big[ ( \frac{1}{2} n - 1 ) + ( \frac{1}{2} n + 1) \lambda_{\min}(\T_{t,(i,j)}^{\mathcal{E}}) \Big] \geq 0
		\end{aligned}
	\end{align}
	when $	n \geq 2\frac{1 - \lambda_{\min}(\T_{t,(i,j)}^{\mathcal{E}}) }{1 + \lambda_{\min}(\T_{t,(i,j)}^{\mathcal{E}}) } = 2\frac{1 - \lambda_{\min}(T_t^{\mathcal{E}}) }{1 + \lambda_{\min}(T_t^{\mathcal{E}}) }$.
\end{proof}

\begin{example}\label{example:Eigenvalue}
For 2-designs and the gadget models, Lemma \ref{lemma:Eigenvalues1} can be improved to Corollary \ref{cor:Eigenvalues3} such that the condition is relaxed to
	\begin{align}
		n \geq \frac{2}{1 + \lambda_{\min}(T_t^{\Gadget}) }.
	\end{align}
	Particularly by Table \ref{tab:Gadgets},
	\begin{align}
		\lambda_{\min}(\T_{2,(i,j)}^{\iSWAP}) = \lambda_{\min}( T_2^{\iSWAP} ) = -\frac{1}{3}, 
	\end{align}
	Therefore, when $n \geq 3$, 
	\begin{align}\label{eq:iSWAPEigenvalue}
		\lambda_2( \T_{2,\mathcal{G}}^{\iSWAP}) \geq \vert \lambda_{\min}( \T_{2,\mathcal{G}}^{\iSWAP}) \vert,
	\end{align} 
	and the spectral gap of $\T_{2,\mathcal{G}}^{\iSWAP}$ is genuinely determined by its second largest eigenvalue only when we have more than two qubits.
\end{example}

As a caveat, the lower bound of $n$ from Lemma \ref{lemma:Eigenvalues1} becomes invalid when $\lambda_{\min}(T_t^{\mathcal{E}}) = -1$. However, this case is not relevant because the underlying ensemble, e.g., the gadget of $\SWAP$, cannot even generate a unitary $t$-design. 


\subsubsection{Comparison via positive semidefiniteness}\label{sec:PSD}

Let $\T_{2,\mathcal{G}}^{\Gadget_1}, \T_{2,\mathcal{G}}^{\Gadget_2}$ be the second moment operators of two gadgets defined on the same graph $\mathcal{G}$. The simplest method to compare their spectral gaps is   to find conditions that determine the positive semidefinite ordering of $\T_{2,\mathcal{G}}^{\Gadget_1}$ and  $\T_{2,\mathcal{G}}^{\Gadget_2}$.

To this end, we consider matrix representations of the local operators $T_2^{\Gadget_1}, T_2^{\Gadget_2}$ with 
\begin{align}
	T_2^{\Gadget_i} = 
	\begin{pmatrix}
		1 & 0 & 0 & 0 \\
		0 & 1 - c_i -\frac{3}{2}(1-a_i) & c_i & \frac{\sqrt{3}}{2}(1-a_i) \\
		0 & c_i & 1 - c_i -\frac{3}{2}(1-a_i) & \frac{\sqrt{3}}{2}(1-a_i) \\
		0 & \frac{\sqrt{3}}{2}(1-a_i) & \frac{\sqrt{3}}{2}(1-a_i) & a_i
	\end{pmatrix}
\end{align}
under the basis \eqref{eq:2qBasisA} - \eqref{eq:2qBasisD} defined in Example \ref{example:LocalOperator}. By taking tensor products of $I,S$ from Eq.~\eqref{eq:1qBasis2} at first and then orthogonalizing $II,SS,SI,IS$, we can further simplify the analysis with the following basis:
\begin{align}
	b_1 = & \frac{1}{4} II, \\
	b_2 = & \frac{1}{\sqrt{15}} \Big(SS - \frac{1}{4}II \Big), \\
	b_3 = & \frac{1}{4\sqrt{\frac{3}{5}}} \Big(IS - \frac{d}{5}(II + SS) \Big), \\
	b_4 = & \frac{1}{3} \Big(SI + \frac{1}{4}IS - \frac{1}{2}(SS + II) \Big).
\end{align}
By a change of basis, we obtain
\begin{align}\label{eq:Gadget2x2}
	\scalemath{0.85}{ 
	\begin{pmatrix}
		1 & 0 & 0 & 0 \\
		0 & 1 - c -\frac{3}{2}(1-a) & c & \frac{\sqrt{3}}{2}(1-a) \\
		0 & c & 1 - c -\frac{3}{2}(1-a) & \frac{\sqrt{3}}{2}(1-a) \\
		0 & \frac{\sqrt{3}}{2}(1-a) & \frac{\sqrt{3}}{2}(1-a) & a
	\end{pmatrix} \mapsto 
	\begin{pmatrix}
		1 & 0 & 0 & 0 \\
		0 & 1 & 0 & 0 \\
		0 & 0 & \frac{1}{8}(15a - 10c - 7) & \frac{\sqrt{15}}{8}( a + 2c - 1) \\
		0 & 0 & \frac{\sqrt{15}}{8}( a + 2c-1) & \frac{1}{8}(17a - 6c - 9) 
	\end{pmatrix}. } 
\end{align}

Let 
\begin{align}
	\Delta a = a_1 - a_2, \quad \Delta c = c_1 - c_2. 
\end{align}
The difference between $T_2^{\Gadget_1}$ and $T_2^{\Gadget_2}$ on right bottom $2 \times 2$ subblock from \eqref{eq:Gadget2x2} is
\begin{align}
	\begin{pmatrix}
		\frac{1}{8}(15\Delta a - 10\Delta c) & \frac{\sqrt{15}}{8}(\Delta a + 2\Delta c) \\
		\frac{\sqrt{15}}{8}(\Delta a + 2\Delta c) & \frac{1}{8}(17\Delta a - 6\Delta c) 
	\end{pmatrix},
\end{align}
whose trace and determinant is
\begin{align}
	2(2\Delta a - \Delta c), \quad \frac{5}{4} \Delta a (3\Delta a - 4\Delta c),
\end{align}
respectively. It suffices to determine the positive semidefiniteness of a $2 \times 2$ matrix by the signs of its trace and determinant. Moreover, arbitrary summations of positive semidefinite matrices is still positive semidefinite. Therefore,
\begin{align}\label{eq:PSD}
	\begin{cases}
		2\Delta a - \Delta c \leq 0, \ \Delta a (3\Delta a - 4\Delta c) \geq 0 
		\implies T_2^{\Gadget_1} \leq T_2^{\Gadget_2} 
		\implies \T_{2,\mathcal{G}}^{\Gadget_1} \leq \T_{2,\mathcal{G}}^{\Gadget_2}, \\
		2\Delta a - \Delta c \geq 0, \ \Delta a (3\Delta a - 4\Delta c) \geq 0 
		\implies T_2^{\Gadget_1} \geq T_2^{\Gadget_2} 
		\implies \T_{2,\mathcal{G}}^{\Gadget_1} \geq \T_{2,\mathcal{G}}^{\Gadget_2}, 
	\end{cases}
\end{align}
where the operator inequalities specify the semidefinite orderings between the operators.

A direct computation shows that the eigenvalues of $T_2^{\Gadget}$, except those equal to one, are given by
\begin{align}\label{eq:GadgetEigenvalues}
	\frac{1}{2} \Big[ 4a - 2c - 2 \pm \sqrt{\Big( \frac{1}{8}(15a - 10c - 7) - \frac{1}{8}(17a - 6c - 9) \Big)^2 + 4\Big(  \frac{\sqrt{15}}{8}( a + 2c - 1) \Big)^2 } \Big].
\end{align}
They can be further simplified as $\lambda_2(T_2^{\Gadget}) = \frac{1}{2}(5a-3) \geq \lambda_3(T_2^{\Gadget}) = \frac{1}{2}(3a-4c-1)$. Despite the fact that a nonnegative spectrum is merely a necessary condition for positive semidefiniteness, in the current special situation, Condition \eqref{eq:PSD} is equivalent to 
\begin{align}
	\begin{cases}
		\lambda_2(T_2^{\Gadget_1}) \leq \lambda_2(T_2^{\Gadget_2}), \ \lambda_3(T_2^{\Gadget_1}) \leq \lambda_3(T_2^{\Gadget_2})
		\implies T_2^{\Gadget_2} \leq T_2^{\Gadget_2} 
		\implies \T_{2,\mathcal{G}}^{\Gadget_1} \leq \T_{2,\mathcal{G}}^{\Gadget_2}, \\
		\lambda_2(T_2^{\Gadget_1}) \geq \lambda_2(T_2^{\Gadget_2}), \ \lambda_3(T_2^{\Gadget_1}) \geq \lambda_3(T_2^{\Gadget_2})
		\implies T_2^{\Gadget_1} \geq T_2^{\Gadget_2} 
		\implies \T_{2,\mathcal{G}}^{\Gadget_1} \geq \T_{2,\mathcal{G}}^{\Gadget_2}. 
	\end{cases} \tag{81$^\ast$}
\end{align}

\begin{theorem}\label{thm:iSWAP1}
	Let $T_2^{\iSWAP}$ be the 2-local second moment operator of the $\iSWAP$ gadget acting on two qubits. Let $T_2^{\Gadget}$ be given by a generic gadget. Then either
	\begin{align}
		 \T_2^{\Gadget} \geq T_2^{\iSWAP}
	\end{align}
	or they are incomparable in the sense of semidefiniteness. The former case corresponds to 2-qubit gates with KAK coefficients satisfying
	\begin{align}\label{eq:Region}
		3\Big( a - \frac{5}{9} \Big) - 4\Big( c-\frac{1}{3} \Big) = \frac{1}{36} & \Big( 25 + 18 (\cos(4 k_x) + \cos(4 k_y) + \cos(4 k_z) ) \\
		& + 7( \cos(4 k_x) \cos(4 k_y) + \cos(4 k_y) \cos(4 k_z) + \cos(4 k_z) \cos(4 k_x) ) \Big) \geq 0. \notag
	\end{align}
  The region identified by this condition is illustrated Fig.~\ref{fig:Weyl2}.
	Within this region, with respect to any connected graph $\mathcal{G}$ on $n$ qubits with $n \geq 3$, $\T_{2,\mathcal{G}}^{\iSWAP}$ always has a larger spectral gap.
\end{theorem}
\begin{proof}
	We prove by contradiction that $T_2^{\Gadget} < T_2^{\iSWAP}$ does not hold for any gadget. Assume there exists a gadget satisfying this inequality, then we must have:
    \begin{align}\label{eq:iSWAPCondition1}
	\tr( T_2^{\Gadget}) < \tr( T_2^{\iSWAP}). 
   \end{align} 
	By Eqs.~\eqref{eq:Gadget4x4} and \eqref{eq:Gadget2x2}, the trace of $H_{\Gadget}$ is
	\begin{align}\label{eq:TraceiSWAP}
		\begin{aligned}
			4a-2c = 2 + \frac{1}{18}\Big( & 3 + 6 (\cos(4 k_x) + \cos(4 k_y) + \cos(4 k_z) ) \\
			& + 5( \cos(4 k_x) \cos(4 k_y) + \cos(4 k_y) \cos(4 k_z) + \cos(4 k_z) \cos(4 k_x) ) \Big). 
		\end{aligned}
	\end{align}
	Its total derivative is
	\begin{align}
		-\frac{10}{9} \begin{pmatrix}
				& \sin(4 k_x) (\cos(4 k_y) + \cos(4 k_z) + \frac{6}{5}) \\
			& \sin(4 k_y) (\cos(4 k_x) + \cos (4 k_z) + \frac{6}{5}) \\
			& \sin(4 k_z) (\cos(4 k_x) + \cos (4 k_y) + \frac{6}{5}) 
		\end{pmatrix}.
	\end{align}	
	Extrema can be achieved when the total derivative vanishes. That is,
\begin{align}
	\sin(4 k_x) = \sin(4 k_y) = \sin(4 k_z) = 0,
\end{align}
or
\begin{align}
	5\cos(4 k_y) + 5\cos(4 k_z) + 6 = 5\cos(4 k_z) + 5\cos(4 k_x) + 6 = 5\cos(4 k_x) + 5\cos(4 k_y) + 6 = 0,
\end{align}
or
\begin{align}
	\sin(4 k_x) = 0, \quad \cos(4 k_x) = -1, \quad 5\cos (4 k_z) + 1 = 5\cos (4 k_y) + 1 = 0,
\end{align}
because other cases like
\begin{align}
	\sin(4 k_x) = 0, \quad \cos(4 k_x) = 1, \quad 5\cos (4 k_z) + 11 = 5\cos (4 k_y) + 11 = 0
\end{align}
do not have solutions for (part of) the KAK coefficients. 

All possible cases corresponds to
\begin{align}
	& \cos(4 k_x) = \mathmakebox[6mm]{1,} \quad \cos(4 k_y) = \mathmakebox[6mm]{1,} \quad \cos(4 k_z) = \mathmakebox[6mm]{1;} \quad 4a-2c = 2. \\
	& \cos(4 k_x) = \mathmakebox[6mm]{1,} \quad \cos(4 k_y) = \mathmakebox[6mm]{1,} \quad \cos(4 k_z) = \mathmakebox[6mm]{-1;} \quad 4a-2c = \frac{2}{9}. \\
	& \cos(4 k_x) = \mathmakebox[6mm]{1,} \quad \cos(4 k_y) = \mathmakebox[6mm]{-1,} \quad \cos(4 k_z) = \mathmakebox[6mm]{-1;} \quad 4a-2c = -\frac{4}{9}.  \\
	& \cos(4 k_x) = \mathmakebox[6mm]{-1,} \quad \cos(4 k_y) = \mathmakebox[6mm]{-1,} \quad \cos(4 k_z) = \mathmakebox[6mm]{-1;} \quad 4a-2c = 0.  \\
	& \cos(4 k_x) = \mathmakebox[6mm]{-\frac{3}{5},} \quad \cos(4 k_y) = \mathmakebox[6mm]{-\frac{3}{5},} \quad \cos(4 k_z) = \mathmakebox[6mm]{-\frac{3}{5};} \quad 4a-2c = -\frac{2}{15}.  \\
	& \cos(4 k_x) = \mathmakebox[6mm]{-1,} \quad \cos(4 k_y) = \mathmakebox[6mm]{-\frac{1}{5},} \quad \cos(4 k_z) = \mathmakebox[6mm]{-\frac{1}{5};} \quad 4a-2c = -\frac{8}{45}. 
\end{align}
As a result, the minima is $-\frac{4}{9}$ with
\begin{align}
	k_x = \frac{\pi}{4}, \quad k_y = \frac{\pi}{4}, \quad k_z = 0.
\end{align} 
We can also exchange the coefficients, but all of them correspond to the same KAK decomposition by Corollary \ref{cor:KAK}. Therefore, when Condition \ref{eq:iSWAPCondition1} holds, $\T_{2,\mathcal{G}}^{\iSWAP}$ is always the smaller one with larger spectral gap. The region from \eqref{eq:Region} can be easily obtained by solving Condition \ref{eq:iSWAPCondition1} for $\iSWAP$.
\end{proof}	
	
\begin{example}\label{example:Gadget2}
	Here we summarize various concrete observations on the comparison between typical 2-qubit gates listed in Table \ref{tab:Gadgets}. Let $\mathcal{G}$ be a fixed connected graph over $n$ qubits. By comparing the matrix entries and eigenvalues from Table \ref{tab:Gadgets} and by Theorem \ref{thm:iSWAP1}, we find: 
	\begin{enumerate}
		\item When $n \geq 3$,
		\begin{align}
			\Delta( \T_{2,\mathcal{G}}^{\iSWAP} ) \geq \Delta(\T_{2,\mathcal{G}}^{\B}) ,\Delta(\T_{2,\mathcal{G}}^{\SQSW}), \Delta(\T_{2,\mathcal{G}}^{\CNOT}).
		\end{align}

		\item For autoconvolution on the two qubits themselves, as discussed earlier, it is clear that the gadgets of $\chi$ gate families form exact 2-designs because the second moment operator $\T_{2,\mathcal{G}}^{\chi}$ matches that of 2-qubit Haar projector $T_{2}^{\U(4)}$. However, when $n \geq 3$, Theorem \ref{thm:iSWAP1} and Eq.~\eqref{eq:iSWAPEigenvalue} implies that $\Delta(\T_{2,\mathcal{G}}^{\iSWAP}) \geq \Delta(\T_{2,\mathcal{G}}^{\chi})$.
		
		\item Similarly, when $n \geq 3$,
		\begin{align}
			\Delta( \T_{2,\mathcal{G}}^{\B} ) \geq \Delta(\T_{2,\mathcal{G}}^{\SQSW}), \Delta(\T_{2,\mathcal{G}}^{\CNOT}), \Delta(\T_{2,\mathcal{G}}^{\chi}).
		\end{align}
		
		\item However, Condition \ref{eq:PSD} is not viable for comparing $\T_{2,\mathcal{G}}^{\iSWAP}$ and $\T_{2,\mathcal{G}}^{\QFT}$. We will deal with the issue in the next subsection.
	\end{enumerate}
\end{example}

\begin{remark}
	It is easy to verify that for each 2-local operator,
	\begin{align}\label{eq:iSWAP-U(4)}
		& (I - T_2^{\chi}) \geq \frac{3}{4} ( I - T_2^{\iSWAP} ) \Leftrightarrow  (I - \T_{2,(i,j)}^{\chi} ) \geq \frac{3}{4}( I - \T_{2,(i,j)}^{\iSWAP} ) \\
		\implies & \Delta(\T_{2,\mathcal{G}}^{\iSWAP}) \geq \Delta(\T_{2,\mathcal{G}}^{\chi}) \geq \frac{3}{4} \Delta(\T_{2,\mathcal{G}}^{\iSWAP}). 
	\end{align}
	That is, despite the fact that $\T_{2,\mathcal{G}}^{\iSWAP}$ has a larger spectral than $\T_{2,\mathcal{G}}^{\chi}$, the scalings of these spectral gaps with respect to the system size $n$ are still the same. 
	
	Similar argument holds for general 2-local ensembles and larger $t$: given arbitrary $\mathcal{E}_1$ and $\mathcal{E}_2$ there should always be constants $c_1,c_2$, independent of $n$ but depends on $t$ and the specific ensembles, such that 
	\begin{align}\label{eq:Ensemble-PSD}
		c_1( I - \T_{t,(i,j)}^{\mathcal{E}_2} ) \leq (I - \T_{t,(i,j)}^{\mathcal{E}_1} ) \leq c_2 ( I - \T_{t,(i,j)}^{\mathcal{E}_2} ).
	\end{align}
	Consequently, on any given graph $\mathcal{G}$, the scaling of $\Delta(\T_{t,\mathcal{G}}^{\mathcal{E}})$ is only determined by the graph. Especially, it is proved in Refs.~\cite{Harrow2design2009,Diniz_2011,BHH16,mittal2023local} that, for any fixed $t$, the spectral gap $\Delta(\T_{t,\mathcal{G}}^{\U(4)}) = \Theta({1}/{n})$ on 1D chains or complete graphs. As a result, this scaling also works for other ensembles, including all gadgets considered here, as long as they are able to form $t$-designs.
\end{remark}

At the end of this subsection, we illustrate in Fig.~\ref{fig:Weyl2} that region of canonical KAK coefficients for which the gadget built from $\iSWAP$ always achieves the optimal convergence speed with respect to \emph{any} connected graph over $n$ qubits when $n \geq 3$. We also conjecture this fact holds for arbitrary 2-local gates with the formal statement being presented in Section \ref{sec:conjectures}.

\begin{figure}[h]
         \subfigure[]{%
        \includegraphics[width=0.42\textwidth]{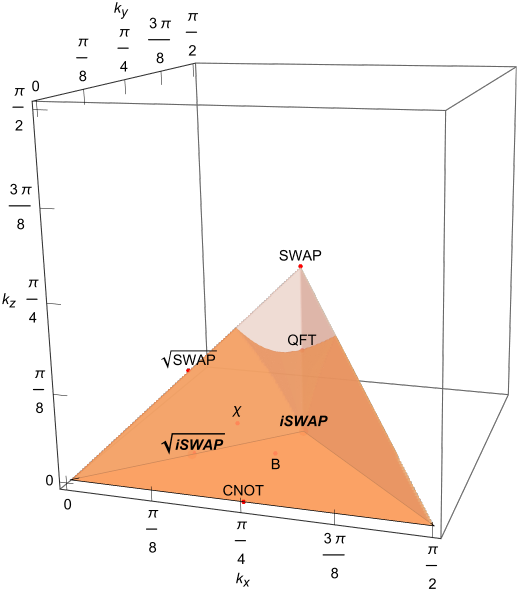} 
    }\hfill
    \subfigure[]{%
        \includegraphics[width=0.45\textwidth]{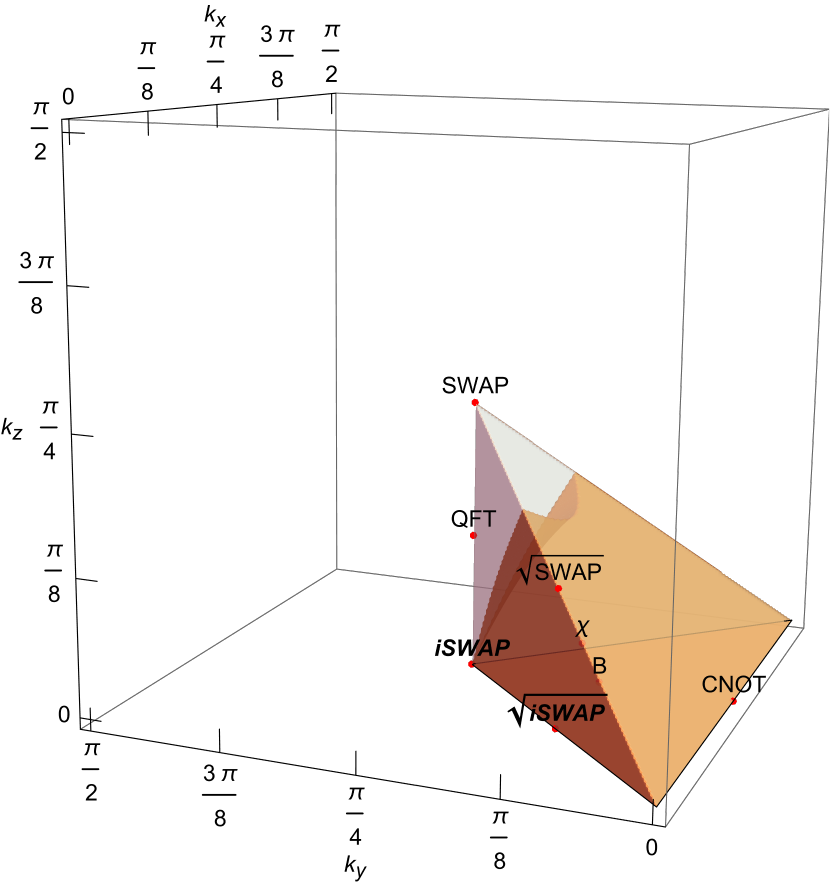}
        }
	\caption{The light yellow tetrahedron represents the Weyl chamber. The orange region  inside, given by \eqref{eq:Region}, encloses the canonical KAK coefficients for which the corresponding gadgets are rigorously comparable to the $\iSWAP$ gadgets in terms of the spectral gaps of their 2-local second moment operators. When $n \geq 3$ and for any connected graph $\mathcal{G}$, the $\iSWAP$ gadget admits the largest spectral gap among all gates within the region.}
 \label{fig:Weyl2}
\end{figure}


\subsubsection{Comparison via Dirichlet forms}\label{sec:Markov} 

The theory of Markov chain is widely useful in the study of designs in the literature (see e.g., Refs.~\cite{Oliveira2design2007a,Oliveira2design2007b,Harrow2design2009,Diniz_2011,Brown_2015,harrow2023approximate}). Roughly speaking, one can translate the moment operator into a transition matrix under a well-prepared basis operator and then apply techniques from Markov chains to bound the spectral gap of that transition matrix. In our case, we notice that a simple rescaling of the basis elements from Eq.~\eqref{eq:2qBasisA} - \eqref{eq:2qBasisD} can transform $T_2^{\Gadget}$ into a row-stochastic transition matrix $P$ as follows:
\begin{align}
	\scalemath{0.8}{ 
	\begin{pmatrix}
		1 & 0 & 0 & 0 \\
		0 & 1 - c -\frac{3}{2}(1-a) & c & \frac{\sqrt{3}}{2}(1-a) \\
		0 & c & 1 - c -\frac{3}{2}(1-a) & \frac{\sqrt{3}}{2}(1-a) \\
		0 & \frac{\sqrt{3}}{2}(1-a) & \frac{\sqrt{3}}{2}(1-a) & a
	\end{pmatrix} \mapsto
	\begin{pmatrix}
		1 & 0 & 0 & 0 \\
		0 & 1 - c -\frac{3}{2}(1-a) & c & \frac{3}{2}(1-a) \\
		0 & c & 1 - c -\frac{3}{2}(1-a) & \frac{3}{2}(1-a) \\
		0 & \frac{1}{2}(1-a) & \frac{1}{2}(1-a) & a
	\end{pmatrix}. }
\end{align}
Despite the fact that the matrix $P$ on the RHS is not Hermitian, this two matrices are similar and thus share the same spectrum. The corresponding moment operators on $n$ qubits can also be transformed and represented by transition matrices. As a basic result of Markov chains, transition matrices having a large \emph{Dirichlet form}~\cite{Levin2009} would possess a large spectral gap.

It turns out that for different gadgets, the corresponding transition matrices always share the same stationary distribution, which is $(\frac{1}{2},\frac{1}{10},\frac{1}{10},\frac{3}{10})$ here. It is well known that in this case, the comparisons via Dirichlet form and positive semidefiniteness are essentially equivalent. In other words, this method does not provide new information in comparing the convergence speeds of gadgets. Nevertheless, viewing these operators as transition matrices still is useful in estimating the influence of perturbations of KAK coefficients on the spectral gaps, which will be elucidated in Section \ref{sec:Perturbation}. 


\subsubsection{Comparison using representation theory}\label{sec:SU(2)}

As discussed above, the positive semidefiniteness condition turns out to be too strong to yield a complete order of all gates. As illustrated in Example \ref{example:Gadget2}, $T_2^{\QFT} - T_2^{\iSWAP}$ is neither positive nor negative semidefinite, indicating that the above methods cannot determine the order of these two gates. Remarkably, when defined on the complete graph $K_n$, using representation theory techniques, we are able to prove that the second moment operator $\T_{2,K_n}^{\iSWAP}$ always achieves the largest possible spectral gap among any generic ensemble $\mathcal{E}$ consisting of 2-local unitaries given that $\T_{2,K_n}^\mathcal{E}$ is Hermitian.

Recall in Example \ref{example:LocalOperator}, we define basis elements $u_0,u_1$ in Eq.~\eqref{eq:1qBasis2} for the 1-qubit Haar projector $T_2^{\U(2)}$ and then write down the matrix representations of $T_2^{\Gadget}$. Let $\mathcal{S} \subset \operatorname{End}( (\mathbb{C}^2)^{\otimes 2} )$ denote the orthogonal complement of $\mathrm{span}\{ u_0, u_1 \}$. Omitting the order of tensor products, we have
\begin{align}\label{eq:decomposition}
	\begin{aligned}
		\mathrm{span}\{ u_0, u_1 \}^{\otimes 2} \otimes \Big( \operatorname{End}( (\mathbb{C}^2)^{\otimes 2} ) \Big)^{\otimes n-2} & = \mathrm{span}\{ u_0, u_1 \}^{\otimes 2} \otimes \Big(  \mathrm{span}\{ u_0, u_1 \} \oplus \mathcal{S} \Big)^{\otimes n-2} \\
		& = \mathrm{span}\{ u_0, u_1 \}^{\otimes n} \oplus ( \mathrm{span}\{ u_0, u_1 \}^{\otimes n-1} \otimes \mathcal{S} ) \oplus \cdots
	\end{aligned}
\end{align}
In the following context, we first analyze the spectral gap of $\T_{2,\mathcal{G}}^{\Gadget} \vert_{\mathrm{span}\{ u_0, u_1 \}^{\otimes n} }$ and then we verify, for any connected graph $\mathcal{G}$,
\begin{align}
	\lambda_2( \T_{2,\mathcal{G}}^{\Gadget} \vert_{\mathrm{span}\{ u_0, u_1 \}^{\otimes n} } ) = \lambda_2( \T_{2,\mathcal{G}}^{\Gadget} ).
\end{align}
Then by Lemma \ref{lemma:Eigenvalues1}, for large $n$, there is no need to consider the orthogonal complement of $\mathrm{span}\{ u_0, u_1 \}^{\otimes n \perp}$ containing subspaces like $\mathrm{span}\{ u_0, u_1 \}^{\otimes n-1} \otimes \mathcal{S}$ from above.

It is natural to note that the standard tensor product basis elements of $\mathcal{V}^{\otimes n} \vcentcolon = \mathrm{span}\{ u_0, u_1 \}^{\otimes n}$ can be represented by binary strings. This simple but insightful observation facilitates our following utilization of group symmetries and representations. Specifically, we can define Pauli matrices $I,X,Y,Z$ (do not get mixed with those acting on the ordinary qubits) acting on $\mathcal{V} = \mathrm{span}\{ u_0, u_1 \}$ and hence Pauli strings on $\mathcal{V}^{\otimes n}$. It is straightforward to consider the representation of the Lie algebra $\mathfrak{su}(2)$, and hence that of the group $\SU(2)$, by
\begin{align}\label{eq:SU(2)action}
	\rho(X) \vcentcolon = X \otimes I \otimes \cdots \otimes I + I \otimes X \otimes \cdots \otimes I + \cdots + I \otimes I \otimes \cdots \otimes X,
\end{align}
and $\rho(Y), \rho(Z)$ defined similarly. The representation on the entire space can be decomposed as follows:
\begin{align}\label{eq:SU(2)irrep}
	\mathcal{V}^{\otimes n} = \mathrm{span}\{ u_0,u_1 \}^{\otimes n} \cong \bigoplus_{r = 0}^{\lfloor n/2 \rfloor} W_{(n-r,r)} \otimes \mathrm{1}^{(n-r,r)},
\end{align}
where $W_{(n-r,r)}$ stands for \emph{the $\SU(2)$ irrep of total spin $j = \frac{1}{2}(n-2r)$} and $\mathrm{1}^{(n-r,r)}$ is the multiplicity space. At the same time, the action of any $\sigma \in S_n$ permuting indices of tensors 
\begin{align}
	u_{i_1} \otimes \cdots \otimes u_{i_n} \equiv \ket{u_{i_1},...,u_{i_n}} \to
	u_{i_{\sigma^{-1}(1)}} \otimes \cdots \otimes 	u_{i_{\sigma^{-1}(n)}} \equiv \ket{u_{i_{\sigma^{-1}(1)}},...,u_{i_{\sigma^{-1}(n)}}}  
\end{align}
is also well-defined and useful later. Symmetrizations of tensors under permutations of indices exactly span the irrep $W_{(n)}$ of the highest total spin. Actually, we have the following lemma:

\begin{lemma}\label{lemma:stationary}
	Let us consider the standard orthonormal basis elements 
	\begin{align}\label{eq:j-m-basis}
		e_r = \ket{j = \frac{n}{2}, m = \frac{n-r}{2}} = \frac{1}{\sqrt{\binom{n}{r}}} \sum_{n-r \text{ many } i_k = 0} \ket{u_{i_1},...,u_{i_k},...,u_{i_n}};  \quad  r = 0, \cdots, n
	\end{align}
	from the highest total spin $j = \frac{n}{2}$ irrep of $\SU(2)$ arisen from the $n$-fold tensor product. Two orthonormal unit eigenvectors of $\T_{2,\mathcal{G}}^{\Gadget}$, defined on the entire operator space with respect to arbitrary connected graph, can be given as
	\begin{align}
		\epsilon_0 & = e_0 = \ket{\frac{n}{2}, \frac{n}{2}} = \ket{u_0,...,u_0,...,u_0}, \\
		\epsilon_1 & = \frac{1}{\sqrt{\frac{1}{3}(4^n- 1)}} \sum_{r=1}^n 3^{\frac{r-1}{2}} \sqrt{\binom{n}{r}} e_r = 	\frac{1}{\sqrt{\frac{1}{3}(4^n - 1)}} \sum_{r=1}^n \sum_{n-r \text{ many } i_k = 0} 3^{\frac{r-1}{2}} \ket{u_{i_1},...,u_{i_k},...,u_{i_n}}.
	\end{align}
\end{lemma}
\begin{proof}
	When $n = 2$, two unit eigenvectors of Eq.~\eqref{eq:Gadget4x4} or, identically, Eq.~\eqref{eq:Gadget4x4-2} in the next lemma are simply
	\begin{align}
		\ket{u_0,u_0}, \quad \frac{1}{\sqrt{5}} \Big( \ket{u_0,u_1} + \ket{u_1,u_0} + \sqrt{3} \ket{u_1,u_1} \Big).
	\end{align}
	The case for arbitrary $n$ is still straightforward by the definition of $u_0,u_1$ in Eq.~\eqref{eq:1qBasis2} and checking that $\epsilon_0, \epsilon_1$ are common unit eigenvectors of $\T_{2,(i,j)}^{\Gadget}$ for any $i \neq j$.
\end{proof}

\begin{lemma}\label{lemma:SU(2)}
	For any gadget, the 2-local second moment operator $\T_{2,(i,j)}^{\Gadget}$ restricted to the two-fold tensor product subspace $\mathrm{span}\{ u_0, u_1 \}^{\otimes 2}$ at sites $i$ and $j$ can be expanded as
	\begin{align}\label{eq:GadgetPauli1}
		\begin{aligned}
			\T_{2,(i,j)}^{\Gadget} \vert_{\mathrm{span}\{ u_0,u_1 \}^{\otimes 2}} = & \Big( (a-c) II + \frac{1}{4} (1-a) (IZ + ZI) + \frac{1}{2} (1-a) ZZ \Big) \\
			& + \Big( \frac{\sqrt{3}}{4} (1-a) (IX + XI) - \frac{\sqrt{3}}{4} (1-a) ( XZ +  ZX) \Big) \\
			& + \frac{1}{2} c (II + XX + YY + ZZ ).
		\end{aligned}
	\end{align}
	As before, we do not explicitly label the sites where Pauli matrices act for conciseness. The local operator is thus $S_2$-symmetric on $\mathrm{span}\{ u_0, u_1 \}^{\otimes 2}$. 
	Let $(i,j)$ is the SWAP operator and let
	\begin{align}
		A = & \Big[ \frac{2}{n(n-1)} \Big( \frac{\sqrt{3}}{4} \Big( (n-1) \rho(X) - \frac{1}{2}(\rho(X) \rho(Z) + \rho(Z) \rho(X) ) \Big) \notag \\
		& \hspace{10mm} + \frac{1}{4} \Big( (n-1) \rho(Z) + \rho(Z) \rho(Z) \Big) \Big) 
		- \Big( 1 + \frac{1}{2(n-1)} \Big) I^{\otimes n} \Big], \\
		B = & \Big[ \frac{2}{n(n-1)} \Big( \sum_{i < j} (i,j) \Big) - I^{\otimes n} \Big].
	\end{align}
    Summing over the complete graph, we have
	\begin{align}\label{eq:GadgetPauli2}
		\T_{2,K_n}^{\Gadget} \vert_{\mathrm{span}\{ u_0,u_1 \}^{\otimes n}} 
		= & (1-a) A + c B + I^{\otimes n},
	\end{align}
  
	which admits the $S_n$-symmetric and can be expressed by $\SU(2)$ representations. 
\end{lemma}
\begin{proof}
	Neglecting the subscripts of sites, by Eq.~\eqref{eq:Gadget4x4}, we have
	\begin{align}\label{eq:Gadget4x4-2}
		& \T_{2,(i,j)}^{\Gadget} \vert_{\mathrm{span}\{ u_0, u_1 \}^{\otimes 2}} = H_{\Gadget} = \begin{pmatrix}
			1 & 0 & 0 & 0 \\
			0 & 1 - c -\frac{3}{2}(1-a) & c & \frac{\sqrt{3}}{2}(1-a) \\
			0 & c & 1 - c -\frac{3}{2}(1-a) & \frac{\sqrt{3}}{2}(1-a) \\
			0 & \frac{\sqrt{3}}{2}(1-a) & \frac{\sqrt{3}}{2}(1-a) & a
		\end{pmatrix}  \\
		= & \begin{pmatrix} 
			1 & 0 & 0 & 0 \\
			0 & 0 & 0 & 0 \\
			0 & 0 & 0 & 0 \\
			0 & 0 & 0 & 0
		\end{pmatrix}
		+ \Big(  \frac{3}{2} a - c - \frac{1}{2} \Big) \begin{pmatrix} 
			0 & 0 & 0 & 0 \\
			0 & 1 & 0 & 0 \\
			0 & 0 & 1 & 0 \\
			0 & 0 & 0 & 0
		\end{pmatrix}
		+ a \begin{pmatrix} 
			0 & 0 & 0 & 0 \\
			0 & 0 & 0 & 0 \\
			0 & 0 & 0 & 0 \\
			0 & 0 & 0 & 1
		\end{pmatrix}
		+ \frac{\sqrt{3}}{2}(1-a) \begin{pmatrix} 
			0 & 0 & 0 & 0 \\
			0 & 0 & 0 & 1 \\
			0 & 0 & 0 & 1 \\
			0 & 1 & 1 & 0
		\end{pmatrix} \notag
	\end{align}
	Expanding with respect to the 2-fold tensor products of Pauli matrices under the Hilbert-Schmidt inner product, we get Eq.~\eqref{eq:GadgetPauli1}. Note that $\frac{1}{2} (II + XX + YY + ZZ )$ is just the SWAP and Eq.~\eqref{eq:GadgetPauli2} can be derived immediately.
\end{proof}

The matrix $A$ from Eq.~\eqref{eq:GadgetPauli2} is negative semidefinite having zero as its largest eigenvalue. This can be seen by analytically solving eigenvalues of the sum of matrices from the first two lines inside Eq.~\eqref{eq:GadgetPauli1}. On the other hand, as a basic result from Schur--Weyl duality and $S_n$ representation theory~\cite{Tolli2009,Sagan01}, $\frac{2}{n(n-1)} \sum_{i < j} (i,j)$ is decomposed into scalar matrices with respect to the direct sum in Eq.~\eqref{eq:SU(2)irrep}. The scalars can also be calculated using $S_n$ characters~\cite{Ingram1950,Roichman1996}. For instance, within the irrep $W_{(n)}$ of the highest total spin $\frac{n}{2}$ and $W_{(n-1,1)}$ of the second highest total spin $\frac{n}{2}-1$, these scalars are
\begin{align}
	& C_{(n)} = \frac{1}{\dim W_{(n)}} \frac{2}{n(n-1)} \tr \Big( \sum_{i < j} (i,j) \vert_{(n)} \Big) = 1, \label{eq:SnCharacter1} \\ 
	& C_{(n-1,1)} = \frac{1}{\dim W_{(n-1,1)}} \frac{2}{n(n-1)} \tr \Big( \sum_{i < j} (i,j) \vert_{(n-1,1)} \Big) = \frac{n-3}{n-1}. \label{eq:SnCharacter2}
\end{align}
As a basic property of these characters~\cite{Sagan01}, $C_{(n-r,r)} > C_{(n-r',r')}$ when $r < r'$. Consequently, the largest eigenvalue, which is one, of
\begin{align}
	cB + I^{\otimes n} = \frac{2}{n(n-1)} \Big( c \sum_{i < j} (i,j) \Big) - cI^{\otimes n} + I^{\otimes n}
\end{align}
is acquired within the irrep $W_{(n)}$ (it is multiplicity-free). It decreases to $1 - c\frac{2}{n-1}$ within the irrep $W_{(n-1,1)} \otimes \mathrm{1}_{(n-1,1)}$. It keeps decreasing as $r$ increases.

By Eq.~\eqref{eq:GadgetPauli2}, there are certain $\lambda(A) < 0, \lambda(B) \leq 0$ for which
\begin{align}
	\lambda_2(\T_{2,K_n}^{\Gadget} \vert_{\mathrm{span}\{ u_0,u_1 \}^{\otimes n}} ) = (1 - a)\lambda(A) + c \lambda(B) + 1.
\end{align}
We are going to prove that $\lambda_2(\T_{2,K_n}^{\iSWAP} \vert_{\mathrm{span}\{ u_0,u_1 \}^{\otimes n}} )$ with the gadget given by $\iSWAP$ do appear when we search within the irrep $W_{(n)}$ of the highest total spin. Therefore, the parameter $c$ defined in Eq.~\eqref{eq:abc} by KAK coefficients becomes insignificant in evaluating that spectral gap.

\begin{lemma}\label{lemma:iSWAP1}
	The second largest eigenvalue of $\T_{2,K_n}^{\iSWAP} \vert_{\mathrm{span}\{ u_0,u_1 \}^{\otimes n}}$ can always be found within the $\SU(2)$ irrep $W_{(n)}$ of the highest total spin.
\end{lemma}
\begin{proof}
	Since it is generally infeasible to solve the second largest eigenvalues as well as the corresponding eigenvector analytically, we turn to find a proper lower bound on the second largest eigenvalue of $\T_{2,K_n}^{\iSWAP} \vert_{W_{(n)}}$ and compare it with eigenvalues of $\T_{2,K_n}^{\iSWAP} \vert_{W_{(n-1,1)}}$ as well as those from irreps with even lower total spin.
	
	We first lower bound the second largest eigenvalue of $\T_{2,K_n}^{\iSWAP} \vert_{W_{(n)}}$. As a simple application of Gram--Schmidt process, we define a unit vector $v$ orthogonal to both $\epsilon_0, \epsilon_1$ (see Lemma \ref{lemma:stationary}) as
	\begin{align}
		v = \frac{1}{1 - \langle e_1, \epsilon_1 \rangle^2} ( e_1 - \langle e_1, \epsilon_1 \rangle \epsilon_1) = \frac{1}{ 1 - \frac{n}{\frac{1}{3}(4^n- 1)} } (e_1 - \sqrt{\frac{n}{\frac{1}{3}(4^n- 1)}} \epsilon_1 ).
	\end{align}
	Using Eq.~\eqref{eq:GadgetPauli2} and $\SU(2)$ representation theory with
	\begin{align}
		& \rho(J_+) = \begin{pmatrix}
			0      & \sqrt{2j} & 0      & \cdots & 0      \\
			0      & 0 & \sqrt{2(2j-1)} & \cdots & 0      \\
			\vdots & \vdots & \vdots & \ddots & \vdots    \\
			0      & 0      & 0      & \cdots & \sqrt{2j} \\
			0      & 0      & 0      & \cdots & 0         \\
		\end{pmatrix}, \\
		& \rho(J_-) = \begin{pmatrix}
			0         & \cdots & 0         & 0      & 0      \\
			\sqrt{2j} & \cdots & 0         & 0      & 0      \\
			\vdots    & \ddots & \vdots    & \vdots & \vdots \\
			0         & \cdots & \sqrt{2(2j-1)} & 0 & 0      \\
			0         & \cdots & \cdots & \sqrt{2j} & 0      \\
		\end{pmatrix}, \\
		& \rho(Z) = 2\begin{pmatrix}
			j      & 0      & \cdots & 0      & 0      \\
			0      & j-1    & \cdots & 0      & 0      \\
			\vdots & \vdots & \ddots & \vdots & \vdots \\
			0      & 0      & \cdots & -j+1   & 0      \\
			0      & 0      & \cdots & 0      & -j     \\
		\end{pmatrix}
	\end{align}
	and $\rho(X) = \rho(J_+) + \rho(J_-)$, we can explicitly write down nontrivial matrix entries of $\T_{2,K_n}^{\Gadget} \vert_{W_{(n)}}$ for $p=0,...,n$ and under the basis $\{e_p\}$ defined in Lemma \ref{lemma:stationary} 
	\begin{align}
		& ( \T_{2,K_n}^{\Gadget} \vert_{W_{(n)}} )\indices{^p_p} = \frac{1-a}{2n(n-1)} (2n-2p-1)(n-2p) - \Big( 1 + \frac{1}{2(n-1) }\Big)(1-a) + 1, \\
		& ( \T_{2,K_n}^{\Gadget} \vert_{W_{(n)}} ) \indices{^p_{p+1}} = \frac{\sqrt{3} (1-a)}{2n(n-1)} 2p \sqrt{(n-p)(p+1)}.
	\end{align}
	\comments{
	original formulas with $p$ starting from $1$
	\begin{align}
		& ( \T_{2,K_n}^{\Gadget} \vert_{W_{(n)}} ) \indices{^p_{p+1}} = \frac{\sqrt{3} (1-a)}{2n(n-1)} \sqrt{(n+1-p)p} (2p-2), \\
		& ( \T_{2,K_n}^{\Gadget} \vert_{W_{(n)}} )\indices{^p_p} = \frac{1-a}{2n(n-1)} (2n-2p+1)(n-2p+2) - \Big( 1 + \frac{1}{2(n-1) }\Big)(1-a) + 1.
	\end{align}
	}
	Evidently, the restricted matrix $\T_{2,K_n}^{\Gadget} \vert_{W_{(n)}}$ is independent of the parameter $c$. For the gadget using $\iSWAP$ with $a = \frac{5}{9}$ (see Table \ref{tab:Gadgets}),
	\begin{align}
		\T_{2,K_n}^{\iSWAP} \vert_{W_{(n)}} (e_1) = \Big( 1 - \frac{4}{3n} \Big) e_1 + \frac{4}{3}\sqrt{\frac{2}{3}} \frac{1}{n\sqrt{n-1}} e_2.
	\end{align}
	Since $\T_{2,K_n}^{\iSWAP} \vert_{W_{(n)}} (\epsilon_1) = \epsilon_1$ and since the operator is Hermitian,
	\begin{align}\label{eq:Bound1}
		\langle v, \T_{2,K_n}^{\iSWAP} \vert_{W_{(n)}} v \rangle & = \Big( 1 - \frac{n}{\frac{1}{3}(4^n- 1)} \Big)^{-2} \Big( \langle e_1,  \T_{2,K_n}^{\iSWAP} \vert_{W_{(n)}} e_1 \rangle -  \langle e_1, \epsilon_1 \rangle \langle \epsilon_1, \T_{2,K_n}^{\iSWAP} \vert_{W_{(n)}} e_1 \rangle \Big) \notag \\
		= & \Big( 1 - \frac{n}{\frac{1}{3}(4^n- 1)} \Big)^{-2}  \Big(  1 - \frac{4}{3n} - \frac{n}{\frac{1}{3}(4^n- 1)} \Big). 
	\end{align}
	Since $v$ is orthogonal to both $\epsilon_0, \epsilon_1$, the above value is a valid lower bound.
	
	To acquire a proper upper bound for the largest eigenvalue of $\T_{2,K_n}^{\iSWAP} \vert_{W_{(n-1,1)}}$, we recall the matrix $A$ defined in Eq.~\eqref{eq:GadgetPauli2} and consider an auxiliary system of $n' = n-2$ qubits. By $\SU(2)$ representation theory, $\dim W_{(n-r,r)} = n - 2r + 1 = \dim W_{(n'-r+1,r-1)}$ and hence $A \vert_{W_{(n-r,r)}}$ and $A \vert_{W_{(n'-r+1,r-1)}}$ are of the same size. Moreover,
	\begin{align}\label{eq:UpperBoundA}
		A \vert_{W_{(n-r,r)}} & = \frac{n'(n'-1)}{(n'+2)(n'+1)} A \vert_{W_{(n'-r+1,r-1)}} + \frac{1}{(n'+2)(n'+1)} \Big( \sqrt{3}\rho(X) + \rho(Z) \Big) - \frac{(4n'+3)}{(n'+2)(n'+1)} I \notag \\
		& = \frac{(n-2)(n-3)}{n(n-1)} A \vert_{W_{(n-r-1,r+1)}} + \frac{1}{n(n-1)} \Big( \sqrt{3}\rho(X) + \rho(Z) \Big) -  \frac{4n-5}{n(n-1)} I.
	\end{align}
	Particularly, we first study the case of $A \vert_{W_{(n-1,1)}}$ and $A \vert_{W_{(n')}} = A \vert_{W_{(n-2)}}$. By the definition of tensor product representation, when $n' = n-2$,
	\begin{align}
	\begin{aligned}
		& \sqrt{3}\rho(X) + \rho(Z) \\
		= & \begin{pmatrix} 1 & \sqrt{3} \\ \sqrt{3} & -1
		\end{pmatrix} \otimes I \otimes \cdots \otimes I + I \otimes \begin{pmatrix} 1 & \sqrt{3} \\ \sqrt{3} & -1
		\end{pmatrix} \otimes \cdots \otimes I + \cdots + I \otimes I \otimes \cdots \otimes \begin{pmatrix} 1 & \sqrt{3} \\ \sqrt{3} & -1
		\end{pmatrix}
	\end{aligned}
	\end{align}
	has eigenvalues $\pm 2(n-2), \pm 2(n-4), \cdots$. On the other hand, the largest eigenvalue of $A \vert_{W_{(n')}} $ is zero. By Lemma \ref{lemma:SU(2)}
	\begin{align}
		\T_{2,K_n}^{\iSWAP} \vert_{W_{(n-1,1)}} = [(1-a) A  + c B + I^{\otimes n} ] \vert_{W_{(n-1,1)}}
	\end{align}
	The eigenvalue of $B \vert_{W_{(n-1,1)}}$ is given by Eq.~\eqref{eq:SnCharacter2}, and thus we obtain
	\begin{align}\label{eq:Bound2}
	\begin{aligned}
		\lambda_1( \T_{2,K_n}^{\iSWAP} \vert_{W_{(n-1,1)}} ) & \leq  (1-a) \Big( \frac{2(n-2)}{n(n-1)}  - \frac{4n-5}{n(n-1)} \Big) - \frac{2c}{n-1} + 1 \\
		& = 1 - \frac{2}{3(n-1)} - \frac{8n - 4}{9n(n-1)},
	\end{aligned}
	\end{align}
	with $a = \frac{5}{9}, c = \frac{1}{3}$ for the $\iSWAP$ given by Table \ref{tab:Gadgets}.
	
	Comparing these two bounds, we find that
	\begin{align}\label{eq:BoundComparison}
	\begin{aligned}
		& \Big( 1 - \frac{n}{\frac{1}{3}(4^n- 1)} \Big)^{-2}  \Big(  1 - \frac{4}{3n} - \frac{n}{\frac{1}{3}(4^n- 1)} \Big) - \Big( 1 - \frac{8n - 4}{9n(n-1)} \Big) \\
		= & \frac{2(n+4)}{9n (n-1)} + \frac{3n-8}{4^n-3n-1} - \frac{12n}{(4^n-3n-1)^2} 
	\end{aligned}
	\end{align}
	is nonnegative when $n \geq 2$. Applying Eq.~\eqref{eq:UpperBoundA} again for $r = 2$, we can expand $A \vert_{W_{(n-2,2)}}$ using $A \vert_{W_{(n'-1,1)}} = A \vert_{W_{(n-3,1)}}$ and achieve a even smaller bound on the largest eigenvalue $\T_{2,K_n}^{\iSWAP} \vert_{W_{(n-2,2)}}$. Arguing in the same way for irreps of fewer total spins, we complete the proof.
\end{proof}

Like Theorem \ref{thm:iSWAP1}, the following lemma provides one more constraints on the parameter $a$ as a function of the KAK coefficients.

\begin{lemma}\label{lemma:iSWAP2}
	For any $k_x,k_y,k_z \in [0, \frac{\pi}{2}]$, the parameter $a(k_x,k_y,k_z)$ defined in Eq.~\eqref{eq:abc} is no less than $a(\frac{\pi}{4},\frac{\pi}{4},0) = \frac{5}{9}$.
\end{lemma}
\begin{proof}
	By definition,
	\begin{align}
		\Delta a = \frac{1}{9} \Big(6 + \cos(4 k_x) \cos(4 k_y) + \cos(4 k_x)\cos(4 k_z) + \cos(4 k_y)\cos(4 k_z) \Big) - \frac{5}{9}.
	\end{align}
	Its total derivative is,
	\begin{align}
		-\frac{4}{9} \begin{pmatrix}
			\sin(4 k_x) (\cos(4 k_y) + \cos(4 k_z) ) \\
			\sin(4 k_y) (\cos(4 k_x) + \cos (4 k_z) ) \\
			\sin(4 k_z) (\cos(4 k_x) + \cos (4 k_y) )
		\end{pmatrix}.
	\end{align}
	Extrema can be achieved when
	\begin{align}
		& \cos(4 k_x) = \mathmakebox[6mm]{1,} \quad \cos(4 k_y) = \mathmakebox[6mm]{1,} \quad \cos(4 k_z) = \mathmakebox[6mm]{1} \implies \Delta a = \frac{4}{9}. \\
		& \cos(4 k_x) = \mathmakebox[6mm]{1,} \quad \cos(4 k_y) = \mathmakebox[6mm]{1,} \quad \cos(4 k_z) = \mathmakebox[6mm]{-1} \implies \Delta a = 0. \\
		& \cos(4 k_x) = \mathmakebox[6mm]{1,} \quad \cos(4 k_y) = \mathmakebox[6mm]{-1,} \quad \cos(4 k_z) = \mathmakebox[6mm]{-1} \implies \Delta a = 0.  \\
		& \cos(4 k_x) = \mathmakebox[6mm]{-1,} \quad \cos(4 k_y) = \mathmakebox[6mm]{-1,} \quad \cos(4 k_z) = \mathmakebox[6mm]{-1} \implies \Delta a = \frac{4}{9}.   
	\end{align}
	Therefore, $\Delta a$ is always nonnegative.
\end{proof}

\begin{lemma}\label{lemma:iSWAP3}
	When defined on the complete graph $K_n$, it holds for all gadgets that
	\begin{align}
		\lambda_2( \T_{2,K_n}^{\Gadget} \vert_{\mathrm{span}\{ u_0,u_1 \}^{\otimes n}} ) \geq \lambda_2( \T_{2,K_n}^{\iSWAP} \vert_{\mathrm{span}\{ u_0,u_1 \}^{\otimes n}} ).
	\end{align}
\end{lemma} 
\begin{proof}
	By Lemma \ref{lemma:iSWAP1},
	\begin{align}
		\lambda_2( \T_{2,K_n}^{\iSWAP} \vert_{\mathrm{span}\{ u_0,u_1 \}^{\otimes n}} ) = 1 + (1- \frac{5}{9})\lambda_2(A \vert_{W_{(n)}})
	\end{align}
	with $\lambda_2(A \vert_{W_{(n)}}) < 0$ being the second largest eigenvalue of $A \vert_{W_{(n)}}$ defined independently of the KAK coefficients in Eq.~\eqref{eq:GadgetPauli2}. On the other hand, 
	\begin{align}
		\lambda_2( \T_{2,K_n}^{\Gadget} \vert_{\mathrm{span}\{ u_0,u_1 \}^{\otimes n}} ) \geq 1 + (1- a)\lambda_2(A \vert_{W_{(n)}})
	\end{align}
	because $a$ is no less than $\frac{5}{9}$ by Lemma \ref{lemma:iSWAP2}. This finishes the proof.
\end{proof}


As promised in the beginning of this subsection, we now reach the position to verify that the second largest eigenvalue of $\T_{2,\mathcal{G}}^{\Gadget}$ can always be found within the $\mathrm{span}\{ u_0, u_1 \}^{\otimes n}$ among all other subspaces from the direct sum in Eq.~\eqref{eq:decomposition}. Together with Lemma \ref{lemma:Eigenvalues1}, Corollary \ref{cor:Eigenvalues3} and Lemma \ref{lemma:iSWAP3}, we can finally conclude that the ensemble corresponding to $\T_{2,K_n}^{\iSWAP}$ achieves the fastest convergence to unitary 2-designs over any ensemble of 2-local unitaries when $n \geq 3$. Note that the lemma below applies to a general graph, though we only need it on $K_n$ at present.

\begin{lemma}\label{lemma:Eigenvalues2}
	When $n \geq \frac{2}{1 + \lambda_{\min}(T_2^{\Gadget} )} $, there is no eigenvalues of $\T_{2,\mathcal{G}}^{\Gadget}$ restricted to the complement of $\mathrm{span}\{ u_0,u_1 \}^{\otimes n}$ that can be larger than the second largest eigenvalue of $\T_{2,\mathcal{G}}^{\Gadget} \vert_{\mathrm{span}\{ u_0,u_1 \}^{\otimes n}}$. In other words,
	\begin{align}
		\lambda_2(\T_{2,\mathcal{G}}^{\Gadget}) = \lambda_2( \T_{2,\mathcal{G}}^{\Gadget} \vert_{\mathrm{span}\{ u_0,u_1 \}^{\otimes n}} ).
	\end{align}
\end{lemma}
\begin{proof}
	Let us recall that $\mathcal{S} \subset \operatorname{End}((\mathcal{C}^2)^{\otimes 2})$ represents the orthogonal complement of $\mathrm{span}\{ u_0, u_1 \}$. We only compare the restrictions of $\T_{2,\mathcal{G}}^{\Gadget}$ to $\mathcal{S} \otimes \mathrm{span}\{ u_0, u_1 \}^{\otimes n-1}$ and $\mathrm{span}\{ u_0, u_1 \}^{\otimes n}$. Other cases can be proved analogously. Within $\mathcal{S} \otimes \mathrm{span}\{ u_0, u_1 \}^{\otimes n-1}$, any 2-local second moment operator acting on the vectors taking tensor product components from $\mathcal{S}$ degenerates to the zero matrix. As a result,
	\begin{align}
		\T_{2,\mathcal{G}}^{\Gadget} \Big\vert_{S \otimes \mathrm{span}\{ u_0, u_1 \}^{\otimes n-1}} = \frac{1}{\vert E(\mathcal{G}) \vert} \sum_{(i,j) \in E(\mathcal{G}) \setminus N(1) } \T_{2,(i,j)}^{\Gadget} \vert_{\mathcal{S} \otimes \mathrm{span}\{ u_0,u_1 \}^{\otimes n-1}},
	\end{align}
	where $N(1)$ is the neighbour of vertex 1.  Its spectrum, if neglecting the multiplicities, is just identical to that of
	\begin{align}
		\frac{1}{\vert E(\mathcal{G}) \vert} \sum_{(i,j) \in E(\mathcal{G}) \setminus N(1) } \T_{2,(i,j)}^{\Gadget} \vert_{\mathrm{span}\{ u_0, u_1 \}^{\otimes n}},
	\end{align}
	because these 2-local moment operators also act trivially on the first component of the tensor product. Its largest eigenvalue is easily obtained by counting the number of edges, which equals $\frac{\vert E(\mathcal{G}) \vert - \vert N(1) \vert }{\vert E(\mathcal{G}) \vert}$.
	
	On the other hand, let $v_1 = \ket{u_1,...,u_0,...,u_0}$ and we consider the normalized vector
	\begin{align}
		v' = \frac{1}{1 - \langle v_1, \epsilon_1 \rangle^2} ( v_1 - \langle v_1, \epsilon_1 \rangle \epsilon_1) = \frac{1}{ 1 - \frac{1}{\frac{1}{3}(4^n- 1)} } (v_1 - \sqrt{\frac{1}{\frac{1}{3}(4^n- 1)}} \epsilon_1 ).
	\end{align}
	Let
	\begin{align}
		\begin{aligned}
			& \frac{1}{\vert E(\mathcal{G}) \vert} \sum_{(i,j) \in E(\mathcal{G}) } \T_{2,(i,j)}^{\Gadget} \vert_{\mathrm{span}\{ u_0, u_1 \}^{\otimes n}} 
			= A' + B' \\
			= & \frac{1}{\vert E(\mathcal{G}) \vert} \sum_{(i,j) \in E(\mathcal{G}) \setminus N(1) } \T_{2,(i,j)}^{\Gadget} \vert_{\mathrm{span}\{ u_0, u_1 \}^{\otimes n}}
			+ \frac{1}{\vert E(\mathcal{G}) \vert} \sum_{(i,j) \in N(1) } \T_{2,(i,j)}^{\Gadget} \vert_{\mathrm{span}\{ u_0, u_1 \}^{\otimes n}} 
		\end{aligned}
	\end{align}
	Directly bounding the eigenvalue of $A'+B'$ is difficult, so we take the square. Like Eq.~\eqref{eq:Bound1}, we have
	\begin{align} 
	& \langle v', (A' + B')^2 v' \rangle \notag \\
	= & \Big( 1 - \frac{1}{\frac{1}{3}(4^n- 1)} \Big)^{-2} \Big( \langle v_1, (A'+B')^2 v_1 \rangle - \frac{1}{\frac{1}{3}(4^n- 1)} \Big) \notag \\
	= & \Big( 1 - \frac{1}{\frac{1}{3}(4^n- 1)} \Big)^{-2} \Big( \langle v_1, A'^2 v_1 \rangle + \langle v_1, (A'B'+B'A') v_1 \rangle + \langle v_1, B'^2 v_1 \rangle - \frac{1}{\frac{1}{3}(4^n- 1)} \Big) \\
	= & \Big( 1 - \frac{1}{\frac{1}{3}(4^n- 1)} \Big)^{-2} \Big( \frac{\vert E(\mathcal{G}) \vert - \vert N(1) \vert }{\vert E(\mathcal{G}) \vert} + \langle v_1, (A'B'+B'A') v_1 \rangle \notag \\ 
	& \hspace{3cm} + \frac{\vert N(1) \vert }{\vert E(\mathcal{G}) \vert^2} \Big( \frac{3}{4}(1-a)^2 + c^2 + \vert N(1) \vert \Big(1 - \frac{3}{2}(1-a) + c \Big)^2 \Big) - \frac{1}{\frac{1}{3}(4^n- 1)} \Big), \notag
	\end{align}
	where we used Eq.~\eqref{eq:Gadget4x4} to compute $\langle v_1, B'^2 v_1 \rangle$. Since $A',B'$ only bear nonnegative entries, $\langle v_1, (A'B'+B'A') v_1 \rangle \geq 0$. Since $\vert N(1) \vert \geq 1$, 
	\begin{align}
	& \langle v', (A+B)^2 v' \rangle \\
	\geq & \Big( 1 - \frac{1}{\frac{1}{3}(4^n- 1)} \Big)^{-2} \Big( \frac{\vert E(\mathcal{G}) \vert - \vert N(1) \vert }{\vert E(\mathcal{G}) \vert} + \frac{\vert N(1) \vert }{\vert E(\mathcal{G}) \vert^2} \Big( \frac{3}{4}(1-a)^2 + c^2 + \Big(1 - \frac{3}{2}(1-a) + c \Big)^2 \Big) - \frac{1}{\frac{1}{3}(4^n- 1)} \Big) \notag \\
	= & \Big( 1 - \frac{1}{\frac{1}{3}(4^n- 1)} \Big)^{-2} \Big( \frac{\vert E(\mathcal{G}) \vert - \vert N(1) \vert }{\vert E(\mathcal{G}) \vert} 
	+ \frac{\vert N(1) \vert }{\vert E(\mathcal{G}) \vert^2} \Big( 3a^2 - 3a + 1 + 2c^2 + 3ac - c \Big) - \frac{1}{\frac{1}{3}(4^n- 1)} \Big). \notag
	\end{align} 
	By Lemma \ref{lemma:iSWAP2}, $a \geq \frac{5}{9}$ and hence the above inequality can be lower bounded by
	\begin{align} 
		\Big( 1 - \frac{1}{\frac{1}{3}(4^n- 1)} \Big)^{-2} \Big( \frac{\vert E(\mathcal{G}) \vert - \vert N(1) \vert }{\vert E(\mathcal{G}) \vert} 
		+ \frac{\vert N(1) \vert }{4 \vert E(\mathcal{G}) \vert^2} - \frac{1}{\frac{1}{3}(4^n- 1)} \Big) \geq \frac{\vert E(\mathcal{G}) \vert - \vert N(1) \vert }{\vert E(\mathcal{G}) \vert}
	\end{align}
	even for $n \geq 2$. 
	Therefore, the second largest eigenvalue of $( \T_{2,\mathcal{G}}^{\Gadget} \vert_{\mathrm{span}\{ u_0, u_1 \}^{\otimes n}} )^2$ is larger than $\frac{\vert E(\mathcal{G}) \vert - \vert N(1) \vert }{\vert E(\mathcal{G}) \vert}$. We also note that
	\begin{align}
		n \geq \frac{2}{1 + \lambda_{\min}(\T_{2,(i,j)}^{\Gadget} )} \implies 
		1 + \lambda_{\min}(\T_{2,\mathcal{G}}^{\Gadget} ) \geq 1 + \lambda_{\min}(\T_{2,(i,j)}^{\Gadget} ) \geq \frac{2}{n} \geq \frac{\vert N(1) \vert }{\vert E(\mathcal{G}) \vert}.
	\end{align} 
	Then by taking square root, there must be a positive eigenvalue of $\T_{2,\mathcal{G}}^{\Gadget} \vert_{\mathrm{span}\{ u_0, u_1 \}^{\otimes n}}$ larger than $\frac{\vert E(\mathcal{G}) \vert - \vert N(1) \vert }{\vert E(\mathcal{G}) \vert}$. The largest eigenvalue of $\T_{2,\mathcal{G}}^{\Gadget}$ restricted to other subspaces, like $\mathcal{S}^{\otimes 2} \otimes \mathrm{span}\{ u_0, u_1 \}^{\otimes n-2}$, are even smaller and cannot exceed the above bound, completing the proof. 
\end{proof}

\begin{example}
	We still take $\iSWAP$ as a example. Lemma \ref{lemma:Eigenvalues2} holds when 
	\begin{align}
		n \geq \frac{2}{1 + \lambda_{\min}(T_2^{\iSWAP} )} = 3 
	\end{align}
	where we recall that $\lambda_{\min}(T_2^{\iSWAP} ) = -\frac{1}{3}$. If $n = 2 < 3$, according to Table \ref{tab:Gadgets} we have
	\begin{align}
		\lambda_2(T_2^{\iSWAP} \vert_{\mathrm{span}\{ u_0,u_1 \}^{\otimes 2}} ) = - \frac{1}{9} < 0 = \lambda_{\max}( T_2^{\iSWAP} \vert_{\mathcal{S} \otimes \mathrm{span}\{ u_0,u_1 \}} )
	\end{align}
	where Lemma \ref{lemma:Eigenvalues2} fails as the condition is not satisfied.
	
	On the other hand, it should be noted that the smallest eigenvalue of $\T_{2,\mathcal{G}}^{\iSWAP} $ may not stay in ${\mathrm{span}\{ u_0,u_1 \}^{\otimes n}}$. For example, when $n = 3$, distinct eigenvalues of
	\begin{align}
		\frac{1}{3} ( \T_{2,(1,2)}^{\iSWAP} + \T_{2,(2,3)}^{\iSWAP} + \T_{2,(1,3)}^{\iSWAP} ) \vert_{\mathcal{S} \otimes \mathrm{span}\{ u_0,u_1 \}^{\otimes 2}}
	\end{align}
	are identical to those of $\frac{1}{3} \T_{2,(2,3)}^{\iSWAP}$, which are just $\frac{1}{3},-\frac{1}{27},-\frac{1}{9}$. However, numerical computation shows that
	\begin{align}
		\frac{1}{3} \lambda_{\min} ( ( \T_{2,(1,2)}^{\iSWAP} + \T_{2,(2,3)}^{\iSWAP} + \T_{2,(1,3)}^{\iSWAP} ) \vert_{ \mathrm{span}\{ u_0,u_1 \}^{\otimes 3}} ) = -\frac{1}{27}
	\end{align}
	which is obviously not the global minimum.
\end{example}

Given any $\T_{2,\mathcal{G}}^{\Gadget}$, Lemma \ref{lemma:Eigenvalues2} conforms that only when $n \geq \frac{2}{1 + \lambda_{\min}(T_2^{\Gadget} )} $, we can find a positive eigenvalue larger than $\vert \lambda_{\min} (\T_{2,\mathcal{G}}^{\Gadget}) \vert$. Specifically, for gadgets and in the formation of 2-designs, this result refines Lemma \ref{lemma:Eigenvalues1}.

\begin{corollary}\label{cor:Eigenvalues3}
	When $n \geq \frac{2}{1 + \lambda_{\min}(T_{2}^{\Gadget} )} $, the spectral gap of $\T_{2,\mathcal{G}}^{\Gadget}$ for any graph $\mathcal{G}$ and any gadget is determined by its second largest eigenvalue.
\end{corollary}


\begin{theorem}\label{thm:iSWAP2}
	Suppose $n \geq 3$. Let $\mathcal{E}$ be an ensemble of 2-local unitaries defined on a connected graph $\mathcal{G}$ over an $n$-qubit system such that on different edges of $\mathcal{G}$, one can even sample from different 2-local gates sets. Suppose $\T_{2,(i,j)}^{\mathcal{E}}$ is Hermitian for each edge $(i,j)$, then
	\begin{align}
		\lambda_2( \T_{2,\mathcal{G}}^{\mathcal{E}} ) \geq \T_{2,K_n}^{\iSWAP}.
	\end{align}
	Regardless of $\T_{2,\mathcal{G}}^{\mathcal{E}}$ bearing negative eigenvalues or not, the ensemble defined by gadgets of $\iSWAP$ on the complete graph attains the fastest convergence to unitary 2-designs.
\end{theorem}
\begin{proof}
	Suppose $\mathcal{E}$ is defined by a fixed type of gadget, the statement has already been proved by combining Lemma \ref{lemma:iSWAP3}, Lemma \ref{lemma:Eigenvalues2} and Proposition \ref{prop:Kn} in Section \ref{sec:architectures}.
	
	Given a general ensemble $\mathcal{E}$ satisfying the assumption, let $\bar{\mathcal{E}}$ be its symmetrization defined by unifying all 2-local gates sets and applying them equally to any pair of qubits (hence the underlying graph becomes complete). Intuitively, $\bar{\mathcal{E}}$ converges faster than $\mathcal{E}$ to unitary 2-designs and we prove this fact first. With respect to any graph $\mathcal{G}$, by definition,
	\begin{align}
		\T_{2,K_n}^{\bar{\mathcal{E}}} = \frac{1}{\vert K_n \vert} \sum_{(i,j) \in E(K_n)} \T_{2,(i,j)}^{\bar{\mathcal{E}}} = \frac{1}{n!} \sum_{\sigma \in S_n} \frac{1}{\vert \mathcal{G} \vert} \sum_{(i,j) \in E(\mathcal{G})} \T_{2,(\sigma(i),\sigma(j))}^{\mathcal{E}} = \frac{1}{n!} \sum_{\sigma \in S_n} \sigma \cdot \T_{2,\mathcal{G}}^{\mathcal{E}}.
	\end{align}
	Suppose $v_1$ is a normalized eigenvector corresponding to the second largest eigenvalue of $\T_{2,K_n}^{\bar{\mathcal{E}}}$. Since $\T_{2,K_n}{\bar{\mathcal{E}}}, \sigma \cdot \T_{2,\mathcal{G}}^{\mathcal{E}}$ always share the unit eigenspace, i.e., $\text{Comm}_2(\U^{2^n})$, $\langle v_1, \sigma \cdot \T_{2,\mathcal{G}}^{\mathcal{E}} v_1 \rangle \leq \lambda_2 (\sigma \cdot \T_{2,\mathcal{G}}^{\mathcal{E}})$. Since $\sigma_1 \cdot \T_{2,\mathcal{G}}^{\mathcal{E}}, \sigma_2 \cdot \T_{2,\mathcal{G}}^{\mathcal{E}}$ are similar operators for any $\sigma_1, \sigma_2 \in S_n$,
	\begin{align}\label{eq:Symmetrize}
		\lambda_2(\T_{2,K_n}^{\bar{\mathcal{E}}}) = \langle v_1, \T_{2,K_n}^{\bar{\mathcal{E}}} v_1 \rangle 
		= \langle v_1, \frac{1}{n!} \sum_{\sigma \in S_n} \sigma \cdot \T_{2,\mathcal{G}}^{\mathcal{E}} v_1 \rangle \leq \lambda_2 (\T_{2,\mathcal{G}}^{\mathcal{E}}).
	\end{align}
	Consequently, we just need to prove the theorem for $\T_{2,K_n}^{\bar{\mathcal{E}}}$.
	
	To this end, let us consider an auxiliary implementation of 1-qubit Haar random matrices over the whole system. By the same argumentation used in Example \ref{example:LocalOperator}, the operator
	\begin{align}
	\begin{aligned}
		& (\T_{2,1}^{\U(2)} \cdots \T_{2,n}^{\U(2)} ) \T_{2,K_n}^{\bar{\mathcal{E}}} (\T_{2,1}^{\U(2)} \cdots \T_{2,n}^{\U(2)} ) \\
		= & \frac{2}{n(n-1)} \sum_{i < j} (\T_{2,1}^{\U(2)} \cdots \T_{2,n}^{\U(2)} ) \T_{2,(i,j)}^{\bar{\mathcal{E}}} (\T_{2,1}^{\U(2)} \cdots \T_{2,n}^{\U(2)} )
	\end{aligned}
	\end{align}
	is identically zero in the orthogonal complement of $\mathrm{span}\{ u_0,u_1 \}^{\otimes n}$. Moreover,
	\begin{align}
	\begin{aligned}
		& (\T_{2,1}^{\U(2)} \cdots \T_{2,n}^{\U(2)} ) \T_{2,(i,j)}^{\bar{\mathcal{E}}} (\T_{2,1}^{\U(2)} \cdots \T_{2,n}^{\U(2)} ) \vert_{\mathrm{span}\{ u_0,u_1 \}^{\otimes n}} \\
		= & \T_{2,i}^{\U(2)} \T_{2,j}^{\U(2)} \T_{2,(i,j)}^{\bar{\mathcal{E}}} \T_{2,i}^{\U(2)} \T_{2,j}^{\U(2)} \vert_{\mathrm{span}\{ u_0,u_1 \}^{\otimes n}}.
	\end{aligned}
	\end{align}
	By Cauchy interlacing theorem~\cite{Horn2017}, we have 
	\begin{align}\label{eq:Cauchy}
		\begin{aligned}
			\lambda_2( \T_{2,K_n}^{\bar{\mathcal{E}}} ) 
			& \geq \lambda_2 \big( (\T_{2,1}^{\U(2)} \cdots \T_{2,n}^{\U(2)} ) \T_{2,K_n}^{\bar{\mathcal{E}}} (\T_{2,1}^{\U(2)} \cdots \T_{2,n}^{\U(2)} ) \big) \\
			& = \frac{2}{n(n-1)} \lambda_2 \big( \sum_{i,j}  \T_{2,i}^{\U(2)} \T_{2,j}^{\U(2)} \T_{2,(i,j)}^{\bar{\mathcal{E}}} \T_{2,i}^{\U(2)} \T_{2,j}^{\U(2)} \vert_{\mathrm{span}\{ u_0,u_1 \}^{\otimes n}}  \big).
		\end{aligned}
	\end{align} 
	By Fubini theorem,
	\begin{align}
		& (\T_{2,i}^{\text{U}(2)} \T_{2,j}^{\text{U}(2)}) \T_{2,(i,j)}^{\bar{\mathcal{E}}} (\T_{2,i}^{\text{U}(2)}  \T_{2,j}^{\text{U}(2)}) \\
		= & \int_{\bar{\mathcal{E}}} (A_i \otimes A_j)^{\otimes 2} \otimes (\bar{A}_i \otimes \bar{A}_j)^{\otimes 2} dA_i dA_j 
		\int U_{i,j}^{\otimes 2} \otimes \bar{U}_{i,j}^{\otimes 2} dU
		\int (B_i \otimes B_j)^{\otimes 2} \otimes (\bar{B}_i \otimes \bar{B}_j)^{\otimes 2} dB_i dB_j  \notag \\
		= & \int_{\bar{\mathcal{E}}} \Big( \int \Big( (A_i \otimes A_j) U_{i,j} (B_i \otimes B_j) \Big)^{\otimes 2} \otimes \Big( (\bar{A}_i \otimes \bar{A}_j) \bar{U}_{i,j} (\bar{B}_i \otimes \bar{B}_j) \Big)^{\otimes 2}  dA_i dA_j dB_i dB_j \Big) dU \notag \\
		\implies & \frac{2}{n(n-1)} \sum_{i < j}  \T_{2,i}^{\U(2)} \T_{2,j}^{\U(2)} \T_{2,(i,j)}^{\bar{\mathcal{E}}} \T_{2,i}^{\U(2)} \T_{2,j}^{\U(2)} \vert_{\mathrm{span}\{ u_0,u_1 \}^{\otimes n}}   
		= \int_{\bar{\mathcal{E}}} \T_{2,K_n}^{U} \vert_{\mathrm{span}\{ u_0,u_1 \}^{\otimes n}} dU
	\end{align}
	where $\T_{2,K_n}^{U}$ now denotes the second moment operator of a gadget defined by applying $U$ from the ensemble $\bar{\mathcal{E}}$ as its 2-local unitary to every edge of the complete graph $K_n$.
	
	Let $a(U)$ denote the parameter $a$ defined in Eq.~\eqref{eq:abc} determined by the KAK coefficients of $U$. It is well-defined because, by Proposition \ref{prop:Gadget}, equivalent KAK coefficients lead to the same second moment operator and the same parameter $a$. By Lemma \ref{lemma:iSWAP3}, there exists a normalized eigenvector $v_2$ such that
	\begin{align}
		\T_{2,K_n}^{U} \ket{v_2} = [  1 + (1- a(U)) \lambda_2(A \vert_{W_{(n)}}) ]  \ket{v_2}, \quad \T_{2,K_n}^{\iSWAP} v_2 = [  1 + (1- \frac{5}{9}) \lambda_2(A \vert_{W_{(n)}}) ]  \ket{v_2},
	\end{align}
	with $\lambda_2(A \vert_{W_{(n)}}) < 0$ and $a(U) \geq \frac{5}{9}$ by Lemma \ref{lemma:iSWAP2}. Together with Eq.~\eqref{eq:Symmetrize} and \eqref{eq:Cauchy}, we conclude that 
	\begin{align}
		\begin{aligned}
			\lambda_2( \T_{2,\mathcal{G}}^{\mathcal{E}} ) 
			\geq \lambda_2( \T_{2,K_n}^{\bar{\mathcal{E}}} ) 
			& \geq \frac{2}{n(n-1)} \lambda_2( \sum_{i,j}  \T_{2,i}^{\U(2)} \T_{2,j}^{\U(2)} \T_{2,(i,j)}^{\bar{\mathcal{E}}} \T_{2,i}^{\U(2)} \T_{2,j}^{\U(2)} \vert_{\mathrm{span}\{ u_0,u_1 \}^{\otimes n}}  ) \\
			& \geq \int_{\bar{\mathcal{E}}}  [1 + (1- a(U)) \lambda_2(A \vert_{W_{(n)}}) ] dU \geq \lambda_2( \T_{2,K_n}^{\iSWAP} ),
		\end{aligned}
	\end{align}
	where we used the same method as \eqref{eq:Symmetrize} in the last step.
\end{proof}


\subsection{Further comments and results on different circuit architectures}\label{sec:architectures}

Here we present miscellaneous discussion and additional results on various other important settings.  

\subsubsection{Different circuit graphs and finite-size effects}\label{sec:graphs}

\begin{proposition}\label{prop:Kn}
	For any fixed ensemble $\mathcal{E}$ of 2-local unitaries, the moment operator {for general $t$} defined on the complete graph $K_n$ achieves the largest spectral gap among all possible graphs. In particular,
	\begin{align}
		\lambda_2( \T_{t,P_n}^{\mathcal{E}} ) \geq \lambda_2( \T_{t,C_n}^{\mathcal{E}} )  \geq \lambda_2( \T_{t,K_n}^{\mathcal{E}} ),
	\end{align}
	where $P_n,C_n$ denote path (1D chain) and ring of $n$ vertices respectively.
\end{proposition}
\begin{proof}
	We only demonstrate that $\lambda_2( \T_{t,P_n}^{\mathcal{E}} ) \geq \lambda_2( \T_{t,C_n}^{\mathcal{E}} )$ as other cases follow immediately by almost the same trick that we used in deriving \eqref{eq:Symmetrize}. 
	
	Let $\sigma \in S_n$ such that $\sigma(i) = i+1$ with $\sigma(n) = 1$. Then
	\begin{align}
		\T_{t,C_n}^{\mathcal{E}} = \frac{1}{n} \sum_{i=1}^{n} \T_{t,i,i+1}^{\mathcal{E}} = \frac{1}{n} \sum_{p=0}^{n-1} \frac{1}{n-1} \sum_{i=1}^{n-1} \T_{t,\sigma^p(i),\sigma^p(i+1)}^{\mathcal{E}} = \frac{1}{n} \sum_{p=0}^{n-1} \sigma^p \cdot \T_{t,P_n}^{\mathcal{E}}.
	\end{align}
	Suppose $v$ is a normalized eigenvector corresponding to the second largest eigenvalue of $\T_{t,C_n}^{\mathcal{E}}$. Since $\T_{t,C_n}^{\mathcal{E}}, \sigma \cdot \T_{t,P_n}^{\mathcal{E}}$ share the same eigenspace of the largest eigenvalue,
	\begin{align}
		\lambda_2( \T_{t,C_n}^{\mathcal{E}} ) = \langle v, \T_{t,C_n}^{\mathcal{E}} v \rangle 
		= \langle v, \frac{1}{n} \sum_{p=0}^{n-1} \sigma \cdot \T_{t,P_n}^{\mathcal{E}} v \rangle 
		= \frac{1}{n} \sum_{p=0}^{n-1} \langle v, \sigma \cdot \T_{t,P_n}^{\mathcal{E}} v \rangle 
		\leq \lambda_2 (\T_{t,P_n}^{\mathcal{E}}),
	\end{align}
	which completes the proof.
\end{proof}

For 2-designs, we present our numerical computations on the spectral gaps relative to different gadgets on $n = 4,5,6,7$ qubits. We also compare their gaps on different topological structures including complete graph, rings and stars. As we can see, $\iSWAP$ achieves the largest gap on all of the graphs as expected. However, note that the ordering of gates in general can change for different graphs and system sizes. Also, the gap decreases as $n$ gets larger. For graphs with the same number of vertices, the complete graph has the largest gap, as it has the best connectivity.

\begin{table}[h]
	\centering
	\begin{tabular}{|c|c|c|c|c|c|c||c|}
		\hline
		& Graph  &  $\iSWAP$ &  $\CNOT$ & $\B$ & $\SQiSW$  & $\QFT$ & $\chi$ \\
		\hline
		\multirow{3}{*}{$n=4$} &  Complete &  {\bf 0.273634} &  0.270939 &  {\bf 0.273634} &  0.205226 &  0.136817 &  0.246271\\
		\cline{2-8}
		&  Star &  {\bf 0.217117} &  0.165037 &  0.199084 &  0.142935 &  0.125694 &  0.185534\\
		\cline{2-8}
		&  Ring &  {\bf 0.247555} &  0.213938 &  0.236944 &  0.173520 &  0.132149 &  0.217157\\
		\cline{2-8}
		\hline
		\multirow{3}{*}{$n=5$} &  Complete &  {\bf 0.203359} &  {\bf 0.203359} &  {\bf 0.203359} &  0.152519 &  0.101679 &  0.183023\\
		\cline{2-8}
		&  Star &  {\bf 0.162714} &  0.128211 &  0.150190 &  0.108219 &  0.093477 &  0.139581\\
		\cline{2-8}
		&  Ring &  {\bf 0.164517} &  0.129655 &  0.151815 &  0.109369 &  0.094177 &  0.141115\\
		\cline{2-8}
		\hline
		\multirow{3}{*}{$n=6$} &  Complete &  {\bf 0.166962} &  {\bf 0.166962} &  {\bf 0.166962} &  0.125222 &  0.083481 &  0.150266\\
		\cline{2-8}
		&  Star &  {\bf 0.134406} &  0.107004 &  0.124479 &  0.089850 &  0.076856 &  0.115526\\
		\cline{2-8}
		&  Ring &  {\bf 0.121396} &  0.090262 &  0.109099 &  0.077785 &  0.074041 &  0.102393\\
		\cline{2-8}
		\hline
		\multirow{3}{*}{$n=7$} &  Complete &  {\bf 0.145163} &  {\bf 0.145163} &  {\bf 0.145163} &  0.108872 &  0.072582 &  0.130647\\
		\cline{2-8}
		&  Star &  {\bf 0.117018} &  0.093400 &  0.108465 &  0.078325 &  0.066839 &  0.100630\\
		\cline{2-8}
		&  Ring &  {\bf 0.095746} &  0.068547 &  0.084448 &  0.059810 &  0.061690 &  0.079779\\
		\cline{2-8}
		\hline

	\end{tabular}
	\caption{Spectral gaps of stochastic circuits over different graphs on different numbers of qubits in the generation of 2-designs. {The largest gap for each graph has been boldfaced for clarity. The $\chi$ gate, which is one of the best gates for autoconvolution, has been listed separately on the right for clearer comparison.}}
	\label{tab:differentGraphs}
\end{table}

Before discussing other circuit architectures, we end up here with a quite straightforward estimation on the largest possible spectral gaps for stochastic circuits on graphs. Precisely, given any connected graph $\mathcal{G}$ over $n$ qubits and any ensemble $\mathcal{E}$ of 2-local unitaries, let $d_{\min}$ denote the minimal degree of the graph and let us denote the corresponding vertex by $n$. Then
\begin{align}\label{eq:TrivialBound}
	\begin{aligned}
		\lambda_2( \T_{t,\mathcal{G}}^{\mathcal{E}} ) 
		& = \frac{1}{\vert E(\mathcal{G}) \vert} \lambda_2( \sum_{(i,j) \in E(\mathcal{G}) \setminus N(n)} \T_{t,(i,j)}^{\mathcal{E}}  +  \sum_{(i,j) \in N(n)} \T_{t,(i,j)}^{\mathcal{E}}  ) \\
		& \geq \frac{ \vert E(\mathcal{G}) \setminus N(n) \vert - \vert N(n) \vert }{\vert E(\mathcal{G}) \vert}  
		= 1 - \frac{2 d_{\min}}{\vert E(\mathcal{G}) \vert},
	\end{aligned}
\end{align} 
where we trivially applied Weyl inequality and $\T_{t,(i,j)}^{\mathcal{E}}$ is allowed to possess negative eigenvalues. It is straightforward to see that RHS from above is universally lower bounded by $1 - \frac{4}{n}$ when $\mathcal{G}$ is taken to be the complete graphs. When $t = 2$, it can be further refined to $1 - \frac{4}{3n} - \frac{n}{\frac{1}{3}(4^n- 1)}$ by Eq.~\eqref{eq:Bound1} in proving Theorem \ref{thm:iSWAP2}. 

In any case, the largest possible spectral gap is $O(1/n)$  and it is proved in Refs.~\cite{Harrow2design2009,BHH16,mittal2023local} that the ensembles built by sampling 2-local Haar random unitaries on 1D chains $P_n$, rings $C_n$ as well as the complete graph $K_n$ achieve this scaling. As mentioned after Theorem \ref{thm:iSWAP1}, our results further demonstrate that ensembles of gadgets associated with different 2-local gates, like  $\iSWAP$, $\B$, $\CNOT$, $\SQSW$, $\SQiSW$, $\QFT$ and etc., can also reach the scaling on the aforementioned graphs.

\begin{remark}
	As a caveat, not all graph structure guarantees the  $O(1/n)$ scaling. Let us consider the graph $\mathcal{G}$ defined by
	\begin{align}
		E(\mathcal{G}) \vcentcolon = \{ (i,j); i < j \leq n-1  \} \cup \{(n-1,n)\}.
	\end{align}
	Since $d_{\min}$ of $\mathcal{G}$ equals $1$ and $\vert E(\mathcal{G}) \vert = \frac{(n-1)(n-2)}{2}+1$, by \eqref{eq:TrivialBound},
	\begin{align}
		\lambda_2( \T_{t,\mathcal{G}}^{\mathcal{E}} ) \geq 1 - \frac{4}{n^2 - 3n + 4}.
	\end{align}
	On the other hand, we can construct $\mathcal{G}$ from the simple 1D chain $P_n$ by adding $\frac{(n-1)(n-2)}{2} + 2 - n$ edges. Applying Weyl inequality again, we obtain
	\begin{align}
		& \lambda_2( \T_{t,\mathcal{G}}^{\mathcal{E}} ) 
		= \frac{1}{\vert E(\mathcal{G}) \vert}\lambda_2( \sum_{(i,j) \in E(P_n) } \T_{t,(i,j)}^{\mathcal{E}}  +  \sum_{(i,j) \in E(\mathcal{G}) \setminus E(P_n) } \T_{t,(i,j)}^{\mathcal{E}}  ) \\
		\leq & \frac{1}{\vert E(\mathcal{G}) \vert} \Big( \vert E(P_n) \vert \lambda_2(\T_{t,P_n}^{\mathcal{E}}) + \vert E(\mathcal{G}) \setminus E(P_n) \vert \Big) 
		= 1 - \frac{\vert E(P_n) \vert}{\vert E(\mathcal{G}) \vert} \Big( 1 - \lambda_2(\T_{t,P_n}^{\mathcal{E}})  \Big) = 1 - \Theta(n^{-2}), \notag
	\end{align}
	given the fact that $\Delta(\T_{t,P_n}^{\mathcal{E}}) = \Theta(1/n)$ by any ensemble used before. The lower and upper bounds together imply 
	\begin{align}
		\Delta(\T_{t,\mathcal{G}}^{\mathcal{E}}) = \Theta(n^{-2}),
	\end{align}
	which confirms that the spectral gap cannot be inverse linear. More comprehensive discussion on the influence of graph structures on the spectral gap can be found in e.g.~Refs.~\cite{harrow2023approximate,mittal2023local}. 
\end{remark}


\subsubsection{Brickwork model and whole-layer model}\label{sec:BW-WL}

The \emph{brickwork (or parallel) model} is also a commonly considered architecture in studies of quantum circuits \cite{BHH16,Haferkamp2021,harrow2023approximate,Schuster2024lowDepth}. It can be depicted in the following diagram:
\begin{align*}
	\begin{quantikz}
		& \measure[style={fill=red!20}]{} & \gate[2]{U} & \measure[style={fill=red!20}]{} \slice{} &  &  & \slice{} & \measure[style={fill=red!20}]{}  & \gate[2]{U} & \measure[style={fill=red!20}]{} & \\
		& \measure[style={fill=red!20}]{} &  & \measure[style={fill=red!20}]{} & \measure[style={fill=red!20}]{}  & \gate[2]{U} & \measure[style={fill=red!20}]{} &  \measure[style={fill=red!20}]{} &  & \measure[style={fill=red!20}]{} & \\
		& \measure[style={fill=red!20}]{} & \gate[2]{U} & \measure[style={fill=red!20}]{} &  \measure[style={fill=red!20}]{} &  & \measure[style={fill=red!20}]{} &  \measure[style={fill=red!20}]{} & \gate[2]{U} & \measure[style={fill=red!20}]{} &   \\
		& \measure[style={fill=red!20}]{} &  & \measure[style={fill=red!20}]{}  & \measure[style={fill=red!20}]{} & \gate[2]{U} & \measure[style={fill=red!20}]{}  &  \measure[style={fill=red!20}]{} & & \measure[style={fill=red!20}]{} & \\
		&  & & & \measure[style={fill=red!20}]{} &  & \measure[style={fill=red!20}]{} & & &  & 
	\end{quantikz}
\end{align*}
Formally, the corresponding random circuit can be defined in the following two ways:
\begin{enumerate}
	\item Sampling $U_{(1,2)}^{\Gadget} \otimes \cdots \otimes U_{(n-1,n)}^{\Gadget}$ or $U_{(2,3)}^{\Gadget} \otimes \cdots \otimes U_{(n-2,n-1)}^{\Gadget}$ with equal probability with $U_{(i,i+1)}^{\Gadget}$ being independently drawn from the gadget. The corresponding $t$-th moment operator is
	\begin{align}
		\T_{t,\mathrm{BW}_1}^{\Gadget} = \frac{1}{2} \Big( L_0 + L_1 \Big) = \frac{1}{2} \Big( \T_{t,(1,2)}^{\Gadget} \cdots 	\T_{t,(n-1,n)}^{\Gadget} + \T_{t,(2,3)}^{\Gadget} \cdots 	\T_{t,(n-2,n-1)}^{\Gadget} \Big).
	\end{align}
	
	\item Implementing a unitary $U_{(1,2)}^{\Gadget} \otimes \cdots \otimes U_{(n-1,n)}^{\Gadget}$ first and then a unitary $U_{(2,3)}^{\Gadget} \otimes \cdots \otimes U_{(n-2,n-1)}^{\Gadget}$ with equal probability, with $U_{(i,i+1)}^{\Gadget}$  independently drawn from the gadget. The corresponding $t$-th moment operator is
	\begin{align}
		\T_{t,\mathrm{BW}_2}^{\Gadget} = L_0 \cdot L_1 = \T_{t,(1,2)}^{\Gadget} \cdots \T_{t,(n-1,n)}^{\Gadget} \cdot \T_{t,(2,3)}^{\Gadget} \cdot \T_{t,(n-2,n-1)}^{\Gadget}.
	\end{align}
\end{enumerate}
When the local moment operators happen to be Haar projectors, e.g., when using gadgets of the $\chi$ gate family for $t = 2,3$, the spectral gap of the brickwork model is proved to be constant~\cite{BHH16,Haferkamp2021} using the so-called \emph{detectability lemma}~\cite{Detectability2009,Detectability2016}. For general gadgets, we provide numerical results. Since $L_1 \times L_0$ may not be Hermitian, the gap studied below corresponds to the operator $\sqrt{L_1} L_0 \sqrt{L_1}$. 

\begin{table}[h]
	\centering
	\begin{tabular}{|c|c|c|c|c|c|c|}
		\hline
		&  $\iSWAP$ &  $\CNOT$ & $\B$ & $\SQiSW$ & $\QFT$ & $\chi$ \\
		\hline
		$n=6$  &  0.476409 &  0.347822 &  \bf{0.530036} &  0.306094 &  0.139783 &  0.520000\\
		\hline
		$n=7$  &  0.487969 &  0.322265 &  \bf{0.490442} &  0.282637 &  0.151591 &  0.480483\\
		\hline
		$n=8$  &  0.449789 &  0.302935 &  \bf{0.463302} &  0.266671 &  0.118617 &  0.453726\\
		\hline
		$n=9$  &  \bf{0.457960} &  0.290412 &  0.444241 &  0.255372 &  0.130793 &  0.434866\\
		\hline
		$n=10$ &  \bf{0.434575} &  0.280106 &  0.429746 &  0.247112 &  0.109966 &  0.421115\\
		\hline    
	\end{tabular}
	\caption{Spectral gaps of brickwork circuits on different numbers of qubits in the generation of 2-designs.}
	\label{tab:brickwork}
\end{table}

We pay particular attention to the performances of gadgets of $\B,\iSWAP$ and $\chi$ gates. As explained in Section \ref{sec:autoconvolution}, $T_{2,(i,j)}^\chi$ and the 2-local Haar projector $T_{2,(i,j)}^{\U(4)}$ are identical. For $n=6,7,8$, $B$ could be best choice in constructing an efficient brickwork circuit; for the smallest $n=6$, the sepctral gap corresponding to $\iSWAP$ is even worse than that of $\chi$. However for $n\geq 9$, $\iSWAP$ becomes the best, which serves as an evidence that Theorem \ref{thm:iSWAP2} also holds for brickwork models. We formally state this conjecture in Section \ref{sec:conjectures}.


There are recent works that generalize the brickwork model by using $D$-dimensional lattices~\cite{harrow2023approximate} or employing small random circuits~\cite{Schuster2024lowDepth} in order to generate designs in lower circuit depths. Taking both the required circuit depth and number of gates applied in each layer into account, we list a very rough comparison among these architectures in the following table for references.

\begin{table}[h]
	\centering
	\begin{tabular}{|c|c|c|c|}
		\hline
		&  Depth & Gates in each layer & Total cost \\
		\hline
		Graph circuits & $O(n^2)$ & $1$ & $O(n^2)$ \\
		\hline
		Brickwork circuits & $O(n)$ & $O(n)$ & $O(n^2)$ \\
            \hline
		Circuits on $D$-dimensional lattices & $O(n^{1/D})$ & $O(n)$ & $O(n^{1+1/D})$ \\
		\hline
		Circuits based on small random ensembles & $O(\log n)$ & $O(n)$ & $O(n \log n)$ \\
		\hline
	\end{tabular}
	\caption{Best known circuit depths required for different circuit architectures in forming $\epsilon$-approximate $t$-designs. The dependence on $\epsilon$ and $t$ are omitted here. }
	\label{tab:architectures}
\end{table}

\comments{
\begin{figure}[!ht]
 \centering
 \includegraphics[width=0.25\textwidth]{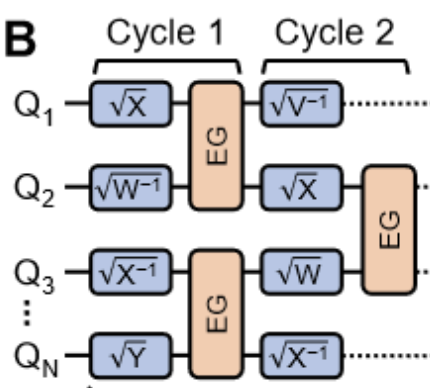}
 \label{fig:brick}
\end{figure}

}


Another natural variant is the \emph{whole-layer model}, in which a full layer of Haar random 1-qubit gates is applied to separate 2-qubit gates as illustrated in the following figure.  This type of circuits  has been used in experimental studies such as Ref.~\cite{decross2024}. 
\begin{align*}
	\begin{quantikz}
		& \measure[style={fill=red!20}]{} & \gate[2]{U} & & & & \measure[style={fill=red!20}]{} & \\
		& \measure[style={fill=red!20}]{} &  &  & \gate[2]{U} & & \measure[style={fill=red!20}]{} & \\
		& \measure[style={fill=red!20}]{} &  & \gate[2]{U} & &  & \measure[style={fill=red!20}]{} & \\
		& \measure[style={fill=red!20}]{} &  &  &  & \gate[2]{U} & \measure[style={fill=red!20}]{} & \\
		& \measure[style={fill=red!20}]{} &  &  &  & & \measure[style={fill=red!20}]{} & 
	\end{quantikz}
\end{align*}
Mathematically, the corresponding moment operator is 
\begin{align}
	\begin{aligned}
		\T_{t,\mathcal{G},\mathrm{WL}}^{U} = & (\T_{t,1}^{\U(2)} \cdots \T_{t,n}^{\U(2)} ) \Big( \frac{1}{2 \vert E(\mathcal{G}) \vert} \sum_{(i,j) \in E(\mathcal{G}) } ( U_{i,j}^{\otimes t} \otimes \bar{U}_{i,j}^{\otimes t} + U_{i,j}^{\dagger \otimes t} \otimes \bar{U}_{i,j}^{\dagger \otimes t} ) \Big)^4 (\T_{t,1}^{\U(2)} \cdots \T_{t,n}^{\U(2)} ) \\
		= & (\T_{t,1}^{\U(2)} \cdots \T_{t,n}^{\U(2)} ) T^4 	(\T_{t,1}^{\U(2)} \cdots \T_{t,n}^{\U(2)} ),
	\end{aligned}
\end{align}
where we add $U^\dagger$ to ensure that the moment operator is Hermitian. It is applied $4$ times in the above diagram. 

Pertinently, we can also consider the \emph{whole-layer gadget model} sketched as follows (already defined and used in proving Theorem \ref{thm:iSWAP2}):
\begin{align*}
	\begin{quantikz}
		& \measure[style={fill=red!20}]{} & \gate[2]{U} & \measure[style={fill=red!20}]{} & & \measure[style={fill=red!20}]{} & & \measure[style={fill=red!20}]{} & & \measure[style={fill=red!20}]{} & \\
		& \measure[style={fill=red!20}]{} &  & \measure[style={fill=red!20}]{} & & \measure[style={fill=red!20}]{} & \gate[2]{U} &  \measure[style={fill=red!20}]{} & & \measure[style={fill=red!20}]{} & \\
		& \measure[style={fill=red!20}]{} &  & \measure[style={fill=red!20}]{} & \gate[2]{U} & \measure[style={fill=red!20}]{} & & \measure[style={fill=red!20}]{} & & \measure[style={fill=red!20}]{} & \\
		& \measure[style={fill=red!20}]{} &  & \measure[style={fill=red!20}]{} &  & \measure[style={fill=red!20}]{} & & \measure[style={fill=red!20}]{} & \gate[2]{U} & \measure[style={fill=red!20}]{} & \\
		& \measure[style={fill=red!20}]{} &  & \measure[style={fill=red!20}]{} &  & \measure[style={fill=red!20}]{} & & \measure[style={fill=red!20}]{} & & \measure[style={fill=red!20}]{} & 
	\end{quantikz}	
\end{align*}
The corresponding moment operator is
\begin{align}
	\T_{t,\mathcal{G},\mathrm{WL}}^{\Gadget} = (\T_{t,1}^{\U(2)} \cdots \T_{t,n}^{\U(2)} ) \Big( \frac{1}{2 \vert E(\mathcal{G}) \vert} \sum_{(i,j) \in E(\mathcal{G}) }  U_{i,j}^{\otimes t} \otimes \bar{U}_{i,j}^{\otimes t} + U_{i,j}^{\dagger \otimes t} \otimes \bar{U}_{i,j}^{\dagger \otimes t} \Big) (\T_{t,1}^{\U(2)} \cdots \T_{t,n}^{\U(2)} ). 
\end{align}

As a simple comparison, we can show that
\begin{align}\label{eq:WLmodel}
	\lambda_t( \T_{t,\mathcal{G},\mathrm{WL}}^{U} ) \geq \lambda_t(\T_{2,\mathcal{G},\mathrm{WL}}^{\Gadget} )
\end{align}
based on the simple intuition that the gadget model uses more 1-qubit Haar random layer. Together with the fact that $\lambda_2(\T_{2,\mathcal{G},\mathrm{WL}}^{\Gadget} ) = \lambda_2(\T_{2,\mathcal{G}}^{\Gadget} )$ verified in proving Theorem \ref{thm:iSWAP1} and Lemma \ref{lemma:Eigenvalues2}, we have
\begin{align}
	 \Delta( \T_{t,\mathcal{G},\mathrm{WL}}^{U} ) \leq \Delta(\T_{2,\mathcal{G},\mathrm{WL}}^{\Gadget} ) = \Delta(\T_{2,\mathcal{G}}^{\Gadget} ). 
\end{align}
In other words, replacing the whole-layer model by the ordinary gadget model does not decrease its convergence speed at least in the formation of 2-designs. A rigorous derivation goes like follows: let
\begin{align}
	P =  (\T_{t,1}^{\U(2)} \cdots \T_{t,n}^{\U(2)} ), \quad T =  \frac{1}{2 \vert E(\mathcal{G}) \vert} \sum_{(i,j) \in E(\mathcal{G}) } ( U_{i,j}^{\otimes t} \otimes \bar{U}_{i,j}^{\otimes t} + U_{i,j}^{\dagger \otimes t} \otimes \bar{U}_{i,j}^{\dagger \otimes t} ).
\end{align}
to simplify the notation. Since $P^2 = P$, we have
\begin{align}
	P T^2 P - (P T P)^2 = P T^2 P - P T P T P = PT (I - P) TP \geq 0
\end{align}
in the sense of positive semidefiniteness. By the same reason,
\begin{align}
	P T^4 P - (P T^2 P)^2 = P T^4 P - P T^2 P T^2 P =  PT^2 (I - P) T^2 P \geq 0 
\end{align}
The first inequality implies that $\lambda_i( (P T^2 P) ) \geq \lambda_i( (P T P)^2 )$ and hence $\lambda_i( (P T^2 P)^2 ) \geq \lambda_i( (P T P)^4 )$. Together with the second one, we have 
\begin{align}
	\lambda_i( P T^4 P ) \geq \lambda_i( (P T^2 P)^2 ) \geq \lambda_i( (P T P)^4 )
\end{align}
which yields \eqref{eq:WLmodel}. 


\subsubsection{Heuristic results for higher designs}\label{sec:3-4-designs}

It is well known that different moments of the probability distribution reveal different characteristics. Conceivably, although using $\iSWAP$ is optimal for 2-designs, it is not clear whether it is still the case for general $t$. Here we provide some numerical results on eigenvalues of 3- and 4-th moment operators of typical gadgets on 2 qubits: 

\begin{table}[H]
	\centering
	\begin{tabular}{|c|c|c|c|c|c|c|c|c|}
		\hline
		&  & $\iSWAP$ &  $\CNOT$ & $\B$ & $\SQiSW$ & $\SQSW$ & $\QFT$ & $\chi$ \\
		\hline
		\multirow{2}{*}{$t = 3$} &  $\lambda_2$ & $\frac{1}{3}$ & $\frac{1}{3}$ & $0$ & $\frac{1}{3}$ & $\frac{1}{6}$ & $\frac{4}{9}$ &  0 \\
		\cline{2-9}
		& $\lambda_{\min}$ & $-\frac{1}{3}$ & $-\frac{1}{3}$ & $-\frac{1}{9}$ & $\frac{1}{6}$ & $0$ & $-\frac{2}{3}$ &  0 \\
		\hline
		\multirow{2}{*}{$t = 4$} & $\lambda_2$ & $\frac{1}{3}$ & $\frac{1}{3}$ & $\frac{1}{18}$ & $\frac{1}{3}$ & $\frac{29}{120}$ & $\frac{4}{9}$ & $0.0144$ \\
		\cline{2-9}
		& $\lambda_{\min}$ & $-\frac{1}{3}$ & $-\frac{1}{3}$ & $-\frac{1}{9}$ & $-\frac{1}{15}$ & $-\frac{1}{3}$ & $-\frac{2}{3}$ & -0.12 \\
		\hline
	\end{tabular}
	\caption{The second largest and smallest eigenvalues of $T_t^{\Gadget}$ with $t = 3,4$ of typical gates families.}
	\label{tab:3/4designs}
\end{table}
  
Based on our previous study of 2-designs, for faster convergence to $t$-designs, it would be desirable to employ gadgets whose 2-local moment operators have small or even negative second eigenvalues. Accordingly, it looks like that the $\B$ gate is a better choice in the formation of 3-designs. While when $t = 4$, the $\chi$ gate would be more better. Due to the scope of this paper, we leave the rigorous analysis on efficient choices of gadgets converging to higher order designs as a future research opportunity. 


\subsection{Robustness guarantees under gate perturbation}\label{sec:Perturbation}

Now we analyze the influence of perturbations of gate coefficients, establishing continuity bounds and robustness guarantees for the spectral gap and convergence efficiency. An evident practical motivation is that deviations are inevitable when implementing the gates in reality, so it would be desirable to have such robustness guarantees which ensure that the efficiency results are ``stable'' under small deviations. Furthermore, note that the bounds only depend on the closeness of KAK coefficients, meaning that they are blind to single-qubit gates and thus apply more broadly than small deviations of the 2-qubit gate. 

For KAK coefficients $k_x,k_y,k_z$, we consider the perturbed version 
\begin{align}
	\tilde{k}_x \in [k_x - \delta,  k_x + \delta], \label{eq:PerturbRange} \\
	\tilde{k}_y \in [k_y - \delta,  k_y + \delta], \\
	\tilde{k}_z \in [k_z - \delta,  k_z + \delta],
\end{align} 
where $\delta > 0$ stipulate the  range of perturbations. Given any connected graph $\mathcal{G}$, we also assume that on different edges of $\mathcal{G}$, there can be different deviations from the desired KAK coefficients $k_x,k_y,k_z$. 

Let $\T_{t,\mathcal{G}}^{\Gadget}$, $\tilde{T}_{t,\mathcal{G}}^{\Gadget}$ be the original $t$-th moment operator and the one obtained after perturbation, respectively. For $t$-design, recall that in Section \ref{sec:foundation}, we demonstrated that any 2-local term $\T_{t,(i,j)}^{\Gadget}$ and $T_t^{\Gadget}$ can be determined by its action on
\begin{align}
	\mathrm{span}\{ u_{0,i},...,u_{D-1,i} \} \otimes \mathrm{span}\{ u_{0,j},...,u_{D-1,j} \} \tag{22$^\ast$}
\end{align}
where $D$ is the dimension of commutant as defined in Eq.~\eqref{eq:1qDimension} as the tensor products of bases 1-qubit Haar projectors $\T_{t,i}^{\U(1)}$ and $\T_{t,j}^{\U(1)}$. Like Eq.~\eqref{eq:MatrixEntries}, the matrix representing $\T_{t,(i,j)}^{\Gadget}$ under this basis can be obtained by computing its entries:
\begin{align}\label{eq:MatrixEntries2}
	\frac{1}{2} \Big[ \tr \Big( (u_{s_1,i} \otimes u_{r_1,j})^\dagger U_{i,j}^{\otimes t} (u_{s_2,i} \otimes u_{r_2,j}) U_{i,j}^{\dagger \otimes t} \Big) 
	+ \tr \Big( (u_{s_1,i} \otimes u_{r_1,j})^\dagger U_{i,j}^{\dagger \otimes t} (u_{s_2,i} \otimes u_{r_2,j}) U_{i,j}^{\otimes t } \Big) \Big].
\end{align} 
Since $U_{i,j}$ is determined by the KAK coefficients, so are these matrix entries. Unlike the case of 2-designs however, it is far more complicated to read off the entries explicitly as in Eq.~\eqref{eq:Gadget4x4}. 

Consequently, it is difficult to compare the spectral gaps of $\T_{t,\mathcal{G}}^{\Gadget}$ and $\tilde{\T}_{t,\mathcal{G}}^{\Gadget}$ analytically. Here we only intend to evaluate the order of the spectral gap perturbation under fluctuations of the KAK coefficients. To this end, we transform $\T_{t,(i,j)}^{\Gadget}$ and $\T_{t,\mathcal{G}}^{\Gadget}$ into transition matrices and employ Markov chain methods discussed in Section \ref{sec:Markov}. We provide an incomplete attempt in the following which motivates the final proof given in Theorem~\ref{prop:Perturb}. To begin with, we first take a collection of independent permutations, denoted by $\sigma_0 = I,...,\sigma_{D-1} \in \text{End}((\mathbb{C}^2)^{\otimes t})$. Then we apply Gram--Schmidt process to get an orthonormal basis $v_0,...,v_{D-1}$ such that
\begin{align}
	(\sigma_0,...,\sigma_{D-1}) = (v_0,...,v_{D-1}) G
\end{align}
with $G$ the transformation matrix. Let $g \neq 0$ be the norm of the last column vector of $G$. Then we find an orthogonal matrix $O$ which rotates that column vector to $\frac{g}{\sqrt{D}}(1,...,1)^T$. Let $A = OG$. It is an invertible matrix and by definition, 
\begin{align}
	(\sigma_0,...,\sigma_{D-1}) = (v_0,...,v_{D-1}) O^T A = (v_0,...,v_{D-1}) O^T \begin{pmatrix} \ast & \cdots & \ast & \frac{g}{\sqrt{D}} \\ \vdots & \ddots & \vdots & \vdots \\ \ast & \cdots & \ast & \frac{g}{\sqrt{D}} \end{pmatrix}.
\end{align}
We set
\begin{align}
	(u_0,...,u_{D-1}) = (v_0,...,v_{D-1}) O^T.
\end{align}
Obviously, $\{u_0,...,u_{D-1}\}$ is still orthonormal. More importantly, their sum $\frac{g}{\sqrt{D}}\sum_{r=0}^{D-1} u_r = \sigma_{D-1}$. Since $ \sigma_{D-1} \otimes \sigma_{D-1}$ is a permutation on $((\mathbb{C}^2)^{\otimes 2})^{\otimes t}$, we have
\begin{align}
	\begin{aligned}
		\T_{t,(i,j)}^{\Gadget}( \sum_{r,s} u_{r,i} \otimes u_{s,j} )
		& = \T_{t,(i,j)}^{\Gadget}\Big( (\sum_{r=0}^{D-1} u_{r,i}) \otimes (\sum_{s=0}^{D-1} u_{s,j}) \Big)
		= \T_{t,(i,j)}^{\Gadget}( \frac{D}{g^2} \sigma_{D-1} \otimes \sigma_{D-1}  ) \\
		& = \frac{D}{g^2} \sigma_{D-1} \otimes \sigma_{D-1} = \sum_{r,s} u_{r,i} \otimes u_{s,j}
	\end{aligned}
\end{align}  
which implies that the matrix representing $\T_{t,(i,j)}^{\Gadget}$ has entries whose row sums are identical to $1$. This fact also holds for $\tilde{\T}_{t,(i,j)}^{\Gadget}$ because they share the same unit eigenspace. However, there is no guarantee that the matrix entries, at least those of $\tilde{\T}_{t,(i,j)}^{\Gadget}$, are nonnegative, which is necessary for the applicability of Markov chain techniques. We will deal with the issue in the following theorem. 

\begin{theorem}\label{prop:Perturb}
	Given perturbation range $\delta$ of the KAK coefficients as defined in \eqref{eq:PerturbRange}, we have
	\begin{align}
		(1 - \epsilon) \Delta( \T_{t,\mathcal{G}}^{\Gadget} ) \leq \Delta( \tilde{\T}_{t,\mathcal{G}}^{\Gadget} ),
	\end{align}
	where $\epsilon = O(\delta)$ and also depends on the order $t$, but independent of the number $n$ of system size and the graph structure $\mathcal{G}$.
\end{theorem}
\begin{proof}
	Given any number $n$ of qubits and any graph structure $\mathcal{G}$, for a fixed $t$, let
	\begin{align}
		P = \frac{1}{2}( \T_{t,(i,j)}^{\Gadget} + I), \quad \tilde{P} = \frac{1}{2}( \tilde{\T}_{t,(i,j)}^{\Gadget} + I ).
	\end{align}
	Obviously, the eigenvalues of $\tilde{P}$ are bounded within $(0,1]$ as long as the gadget is able to form approximate $t$-designs. Then it is proved in Ref.\cite{stochastic1984} that there exists a symmetric doubly stochastic matrix with this prescribed positive spectrum. That is, there is a basis under which $\tilde{P}$ is symmetric and doubly stochastic.
	
	Since $P, \tilde{P}$ shares the same unit eigenspace, each row or column sum of $P$ is also equal to one. Entries of $P$ under this basis are real and bounded within $[-1,1]$. The reason is, $P,\tilde{P}$ are all real and symmetric under standard matrix units spanning $\text{End}((\mathbb{C}^2)^{\otimes t})$. Since the matrix transforming $\tilde{P}$ into a symmetric and doubly stochastic can be real. Applying it to $\tilde{P}$, we can see the first fact. Since eigenvalues of $P$ are also bounded within $(0,1]$, these entries cannot be outside the interval. As a caveat, only matrix entries of $\tilde{P}$ are ensured to be nonnegative and this is sufficient in the following argument.
	
	Let $\T_{t,(i,j)}^{\Gadget}(r,s)$ denote the $(r,s)$-entry of $\T_{t,(i,j)}^{\Gadget}$, as analogous to parameters $a,b,c$ defined in Eq~.\eqref{eq:abc}. Then we consider
	\begin{align}
		\frac{1}{1 - \epsilon} = \max_{\tilde{k}_x, \tilde{k}_y, \tilde{k}_z, } \Big\{ \frac{1 + \T_{t,(i,j)}^{\Gadget}(r,s)}{1 + \tilde{\T}_{t,(i,j)}^{\Gadget}(r,s)}, \ 1 \Big\}
	\end{align}
	Since each row or column sum of either $P$ or $\tilde{P}$ is equal to one and since entries of $\tilde{P}$ are all nonnegative,
	\begin{align}
		\langle (I-P)v,v \rangle = \frac{1}{2}\sum_{r \neq s} P(r,s)(v_r - v_s)^2 
		\leq \frac{1}{1-\epsilon} \langle (I-\tilde{P})v,v \rangle = \frac{1}{1-\epsilon} \frac{1}{2}\sum_{r \neq s} \tilde{P}(r,s)(v_r - v_s)^2. 
	\end{align}
	In the sense of positive semidefiniteness, this implies that 
	\begin{align}\label{eq:Perturb-PSD}
		& I - P \leq \frac{1}{1-\epsilon} (I - \tilde{P}) \Leftrightarrow
		I - \T_{t,(i,j)}^{\Gadget} \leq \frac{1}{1-\epsilon} (I - \tilde{\T}_{t,(i,j)}^{\Gadget}) \\
		\implies & I - \T_{t,(i,j)}^{\Gadget} = \frac{1}{\vert E(\mathcal{G}) \vert} \sum_{(i,j) \in E(\mathcal{G})}  (I - \T_{t,(i,j)}^{\Gadget}) \notag \\
		& \hspace{15mm} \leq \frac{1}{1-\epsilon} \frac{1}{\vert E(\mathcal{G}) \vert} \sum_{(i,j) \in E(\mathcal{G})}  (I - \tilde{\T}_{t,(i,j)}^{\Gadget}) = \frac{1}{1-\epsilon} (I - \tilde{\T}_{t,\mathcal{G}}^{\Gadget}).
	\end{align}
	Therefore,
	\begin{align}
		(1 - \epsilon) \Delta( \T_{t,\mathcal{G}}^{\Gadget} ) \leq \Delta( \tilde{\T}_{t,\mathcal{G}}^{\Gadget} ).
	\end{align}
	By definition, $\epsilon$ is given by the perturbation of matrix entries $\tilde{T}_{t,(i,j)}^{\Gadget})$. It is evaluated locally on 2 qubits and hence independent of the system size $n$, the choice of the edge $(i,j)$ and the graph $\mathcal{G}$. It depends on $t$ and can be written as $O(\delta)$ after taking Taylor expansions.
\end{proof}

The core inequality \eqref{eq:Perturb-PSD} proved in Theorem~\ref{prop:Perturb} refines \eqref{eq:Ensemble-PSD} for cases of gadgets since we are able to compute the constants $\epsilon$ here as long as we have an explicit basis. Specifically, in the case of 2-design, let the parameters defined in Eq.~\eqref{eq:abc} be denoted as $\tilde{a}, \tilde{b}, \tilde{c}$ after perturbation. We can obtain a more refined analysis on the scaling of $\epsilon$ through explicit expansions of $a,b,c$. To be precise, let
\begin{align}
	\frac{1}{1 - \epsilon} = \max_{	\tilde{k}_x, \tilde{k}_y, \tilde{k}_z } \Big\{ b/\tilde{b}, \ c/\tilde{c}, \ 1 \Big\}.
\end{align}
By Eq.~\eqref{eq:abc}, the Taylor expansion formula implies that
\begin{align}
	\tilde{b} = & b + \frac{2}{9}\Big( \sin(4 k_x) (\cos(4 k_y) + \cos(4 k_z) ) \delta k_x  \\
	& + \sin(4 k_y) ( \cos(4 k_x) + \cos (4 k_z)) \delta k_y + 
	\sin(4 k_z) (\cos(4 k_x) + \cos (4 k_y)) \delta k_z \Big) + O(\delta^2), \notag  \\
	\tilde{c} = & c + \frac{1}{12}\Big( \sin(4 k_x) (2 - \cos(4 k_y) - \cos(4 k_z) )  \delta k_x  \\
	& + \sin(4 k_y) (2 - \cos(4 k_x) - \cos (4 k_z)) \delta k_y + \sin(4 k_z) (2 - \cos(4 k_x) - \cos (4 k_y)) \delta k_z \Big) + O(\delta^2). \notag
\end{align}
For the special cases when $(k_x,k_y,k_z) = (\frac{\pi}{4}, \frac{\pi}{4}, 0)$ corresponding to $\iSWAP$ and $(k_x,k_y,k_z) = (\frac{\pi}{4}, 0, 0)$ corresponding to $\CNOT$, the total derivatives of $b$ and $c$ vanish and hence $\epsilon = O(\delta^2)$.

We further make a plot to show how the gap changes with the KAK coefficients. We vary the value of $k_z$ from $0$ to $\frac{\pi}{4}$, with a step size of $\frac{\pi}{40}$. For each value of $k_z$, we plot the gap and optimal probability for different value of $k_x$ and $k_y$. From Fig.~\ref{fig:kak-gap} we observe that the gap changes smoothly as a function of the KAK parameters, which gives evidence for the robustness under gate perturbation.

\begin{figure}[h]
	\centering
	\includegraphics[width=\textwidth]{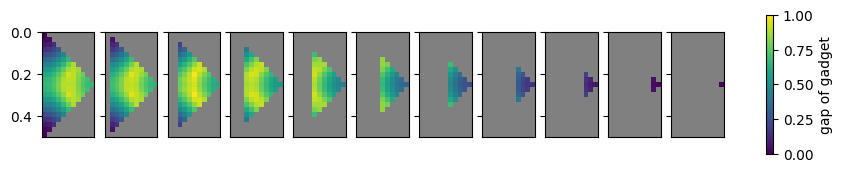}
	\includegraphics[width=\textwidth]{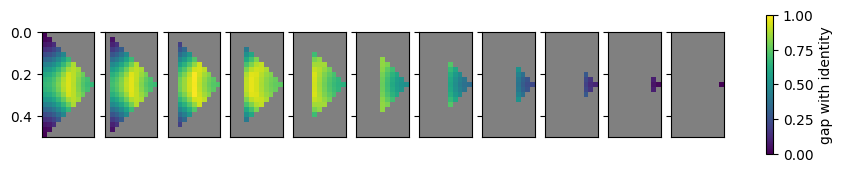}
    \includegraphics[width=\textwidth]{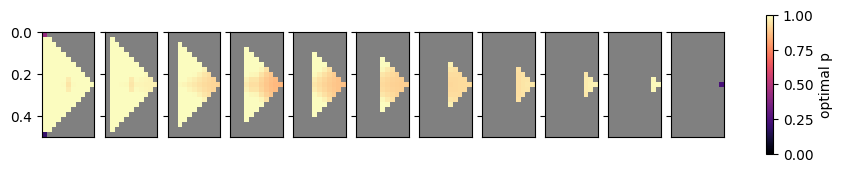}
	\caption{{The gap of $T_t^{\Gadget}$ at $t=3$ for all 2-qubit gates in the Weyl chamber. Here, we plotted the gap corresponding to the gate, as well as the maximum gap obtained by picking the gate with probability $p$ and the identity operator with probability $1-p$, optimized over all $p$. The value of $p$ corresponding to the optimal gap is also plotted. The optimal $p$ is not continuous on the corners, as the gate becomes the identity with these KAK parameters, and the gap is always 0 regardless of $p$.}}
    \label{fig:kak-gap}
\end{figure}


\subsection{Some conjectures and further implications}\label{sec:conjectures}

Here we collect various interesting conjectures that emerge from our analysis and have not been fully resolved. These are regarding the unique properties of $\iSWAP$ in the formation of unitary 2-design based on Theorem \ref{thm:iSWAP1}, Theorem \ref{thm:iSWAP2} and our numerical results presented in Table \ref{tab:differentGraphs} and \ref{tab:brickwork}:

\begin{conjecture}
	Suppose $n \geq 3$. With respect to any fixed connected graph $\mathcal{G}$, for any 2-local unitary circuit ensemble $\mathcal{E}$ provided that $\T_{2,\mathcal{G}}^{\mathcal{E}}$ is Hermitian, 
	\begin{align}
		\Delta( \T_{2,\mathcal{G}}^{\mathcal{E}} ) \leq \Delta( \T_{2,\mathcal{G}}^{\iSWAP} ). 
	\end{align}
\end{conjecture} 

Based on our results in Table \ref{tab:differentGraphs}, we further conjecture that

\begin{conjecture}
	When $n \geq 5$, given any gadget associated with 2-local gates family with the parameter $a = \frac{5}{9}$ (see Eq.~\eqref{eq:abc}) on the complete graph $K_n$,
	\begin{align}
		\Delta( \T_{2,K_n}^{\Gadget} ) = \Delta( \T_{2,K_n}^{\iSWAP} ). 
	\end{align}
	For instance, gadgets associated with gates $\iSWAP$, $\B$ or $\CNOT$ satisfy the requirement. All of them are anticipated to attain the fastest convergence speed towards unitary 2-designs for sufficiently large $n$.
\end{conjecture}

This conjecture can be proven if a stronger version of Lemma \ref{lemma:iSWAP1}, that when $n \geq 5$, the second largest eigenvalue of 
\begin{align}
	\begin{aligned}
		A = & \frac{2}{n(n-1)} \Big( \frac{\sqrt{3}}{4} \Big( (n-1) \rho(X) - \frac{1}{2}(\rho(X) \rho(Z) + \rho(Z) \rho(X) ) \Big) \\
		& \hspace{15mm} + \frac{1}{4} \Big( (n-1) \rho(Z) + \rho(Z) \rho(Z) \Big) \Big) 
		- \Big( 1 + \frac{1}{2(n-1)} \Big) I^{\otimes n}
	\end{aligned}
\end{align}
can always be found within $W_{(n)}$ rather than all other $\SU(2)$ irreps with fewer total spins, holds. It should be noted that the comparison in \eqref{eq:BoundComparison} fails here because we cannot include the second term $-\frac{2}{3(n-1)}$ from \eqref{eq:Bound2}. Numerically, $\lambda_2(A \vert_{W_{(n)}})$ and $\lambda_2(A \vert_{W_{(n-1,1)}})$ tend to be more and more closer asymptotically and thus a more refined method is necessary to rigorously prove the conjecture. We leave this for future work.

Finally, beyond random circuits built on graphs, our numerical results in Section \ref{sec:BW-WL} may also indicates that:
\begin{conjecture}
	When $n \geq 9$, using the brickwork model
	\begin{align}
		\Delta( \T_{t,\mathrm{BW}}^{\Gadget} ) \leq \Delta( \T_{t,\mathrm{BW}}^{\iSWAP} ).
	\end{align}
	That is, in the generation of 2-designs, using the gadget associated with $\iSWAP$ is still optimal for the brickwork model.
\end{conjecture}


\section{Multiqubit gates and hypergraphs}\label{sec:multiqubit}

In this section, we extend the consideration to more general gates with higher locality and provide some cursory discussion. We first provide a general upper bound on the spectral gap of moment operators defined using $r$-local gates. Then we numerically compare the efficiency of several 3-qubit gates in generating 2-designs on a 3-qubit system.

\subsection{General bound on gap}\label{sec:r-local}

Analogous to our estimation at the end of Section \ref{sec:graphs}, we first prove a general lower bound on the $r$-local circuits on hypergraphs. Precisely, let $\mathcal{E}$ be a generic ensemble defined on some hypergraph $\mathcal{G}$ whose edges are tuples $(i_1,...,i_r)$ of distinct sites. Given any site $i_s$, its node degree is thus defined by all nodes as long as they are contained in one tuple standing for an edge. Let $n$ denote the site for which it has the minimal node degree $d_{\min}$ in the hypergraph. Then by the same method in deriving \eqref{eq:TrivialBound} 
\begin{align}\label{eq:TrivialBound2} 
	\begin{aligned}
		\lambda_2( \T_{t,\mathcal{G}}^{\mathcal{E}} ) 
		& = \frac{1}{\vert E(\mathcal{G}) \vert} \lambda_2\Big( \sum_{(i_1,...,i_r) \in E(\mathcal{G}) \setminus N(n)} \T_{t,(i_1,...,i_r)}^{\mathcal{E}}  +  \sum_{(i_1,...,i_r) \in N(n)} \T_{t,(i_1,...,i_r)}^{\mathcal{E}} \Big) \\ 
		& \geq \frac{ \vert E(\mathcal{G}) \setminus N(n) \vert - \vert N(n) \vert }{\vert E(\mathcal{G}) \vert}  
		= 1 - \frac{2 d_{\min}}{\vert E(\mathcal{G}) \vert}.
	\end{aligned}
\end{align}
By definition, $\vert E(\mathcal{G}) \vert \geq \frac{d_{\min}}{r} n$ and hence the RHS from above is universally lower bounded by $1 - \frac{2r}{n}$ (cf. the 2-local case below \eqref{eq:TrivialBound}). As long as the locality is constant, the best scaling of the spectral gap is still inverse linear with respect to $n$ by comparing to the 2-local case. Therefore, there is at most a constant speed up on the convergence.

\subsection{Results of important 3-qubit gates}\label{sec:3-local}

Due to the lack of KAK-type decomposition, it is not immediately clear how to carry out a systematic analysis for higher-locality gates, let alone identifying the optimal solutions. Here we specifically study various important 3-qubit gates through the matrix representations of the corresponding moment operators and eigenvalues as we did in Section \ref{sec:3-4-designs}. 

\begin{example}\label{example:3-local}
	We consider a 3-local unitary ensemble, which we still call an ironed gadget, defined by sampling (1) three 1-local Haar random unitaries on three qubits, (2) a 3-local unitary and (3) three more 1-local Haar random unitaries. The scheme gives the following $t$-th moment operator: 
	\begin{align}
		\T_{t,(i,j,k)}^{\Gadget} = \T_{t,i}^{\U(2)} \T_{t,j}^{\U(2)} \T_{t,k}^{\U(2)} \frac{1}{2} \Big( U^{\otimes t} \otimes \bar{U}^{\otimes t} + U^{\dagger \otimes t} \otimes U^{T \otimes t} \Big) \T_{t,i}^{\U(2)} \T_{t,j}^{\U(2)} \T_{t,k}^{\U(2)},
	\end{align}
	By our earlier discussion in Example \ref{example:LocalOperator}, when $t = 2$, unit eigenbasis of $\T_{t,i}^{\U(2)}$ can be given by Eq.~\eqref{eq:1qBasis2}. Effectively, $\T_{t,(i,j,k)}^{\Gadget}$ is determined by its action on the eight-dimensional subspace $\{u_0, u_1\}^{\otimes 3}$. For instance, $\T_{2,(i,j,k)}^{\CCZ}$ can be represented by:
	\begin{align}\label{eq:CCZ}
		T_2^{\CCZ} = 
		\begin{pNiceMatrix}
			\Block[borders={bottom,right,tikz=dashed}]{1-1}{} 1 & 0 & 0 & 0 & 0 & 0 & 0 & 0 \\
			0 & \Block[borders={bottom,top,right,left,tikz=dashed}]{3-3}{} \frac{1}{2} & 0 & 0 & \frac{1}{6 \sqrt{3}} & 0 & \frac{1}{6 \sqrt{3}} & \frac{1}{18} \\
			0 & 0 & \frac{1}{2} & 0 & \frac{1}{6 \sqrt{3}} & \frac{1}{6 \sqrt{3}} & 0 & \frac{1}{18} \\
			0 & 0 & 0 & \frac{1}{2} & 0 & \frac{1}{6 \sqrt{3}} & \frac{1}{6 \sqrt{3}} & \frac{1}{18} \\
			0 & \frac{1}{6 \sqrt{3}} & \frac{1}{6 \sqrt{3}} & 0 & \Block[borders={bottom,top,right,left,tikz=dashed}]{3-3}{} \frac{4}{9} & \frac{1}{18} & \frac{1}{18} & \frac{1}{3 \sqrt{3}} \\
			0 & 0 & \frac{1}{6 \sqrt{3}} & \frac{1}{6 \sqrt{3}} & \frac{1}{18} & \frac{4}{9} & \frac{1}{18} & \frac{1}{3 \sqrt{3}} \\
			0 & \frac{1}{6 \sqrt{3}} & 0 & \frac{1}{6 \sqrt{3}} & \frac{1}{18} & \frac{1}{18} & \frac{4}{9} & \frac{1}{3 \sqrt{3}} \\
			0 & \frac{1}{18} & \frac{1}{18} & \frac{1}{18} & \frac{1}{3 \sqrt{3}} & \frac{1}{3 \sqrt{3}} & \frac{1}{3 \sqrt{3}} & \Block[borders={left,top,tikz=dashed}]{1-1}{}  \frac{11}{18}
		\end{pNiceMatrix},
	\end{align}
	where the dashed lines highlight the actions of $\T_{2,(i,j,k)}^{\CCZ}$ on binary strings of $\{u_0, u_1\}^{\otimes 3}$ with the same Hamming weights: from top to bottom and left to right, we arrange the basis elements as
	\begin{align}
		\begin{aligned}
			& \ket{u_0,u_0,u_0}, \quad \ket{u_1,u_0,u_0}, \quad \ket{u_0,u_1,u_0}, \quad \ket{u_0,u_0,u_1}, \\
			& \ket{u_1,u_1,u_0}, \quad \ket{u_0,u_1,u_1}, \quad \ket{u_1,u_0,u_1}, \quad \ket{u_1,u_1,u_1}.
		\end{aligned}
	\end{align} 
	Basis elements sharing the same Hamming weight are invariant under the action of $\rho(Z)$ defined in Eq.~\eqref{eq:SU(2)action}. Eq.~\eqref{eq:CCZ}, as well as any summation of local moment operators acting on a connected hypergraph of $n$ qubits, breaks the invariance to ensure that the corresponding moment operator admits exactly two unit eigenvalues (because $\dim \text{Comm}_2(\U(2^n)) = 2$). The eigenvalues of the above matrix are $(1, 1, \frac{5}{9}, \frac{5}{9}, \frac{4}{9},\frac{2}{9}, \frac{1}{3}, \frac{1}{3} )$. As a reminder, the gadget associated with $\Toffoli$ has the same moment operator as above.
	
	As a comparison, $\frac{1}{3}( \T_{2,(i,j)}^{\iSWAP} + \T_{2,(j,k)}^{\iSWAP} + \T_{2,(k,l)}^{\iSWAP} )$ can be represented by
	\begin{align}\label{eq:iSWAP-CCZ}
		T_2^{\iSWAP'} =
		\begin{pNiceMatrix}
			\Block[borders={bottom,right,tikz=dashed}]{1-1}{} 1 & 0 & 0 & 0 & 0 & 0 & 0 & 0 \\
			0 & \Block[borders={bottom,top,right,left,tikz=dashed}]{3-3}{} \frac{1}{3} & \frac{1}{9} & \frac{1}{9} & \frac{2}{9 \sqrt{3}} & 0 & \frac{2}{9 \sqrt{3}} & 0 \\
			0 & \frac{1}{9} & \frac{1}{3} & \frac{1}{9} & \frac{2}{9 \sqrt{3}} & \frac{2}{9 \sqrt{3}} & 0 & 0 \\
			0 & \frac{1}{9} & \frac{1}{9} & \frac{1}{3} & 0 & \frac{2}{9 \sqrt{3}} & \frac{2}{9 \sqrt{3}} & 0 \\
			0 & \frac{2}{9 \sqrt{3}} & \frac{2}{9 \sqrt{3}} & 0 & \Block[borders={bottom,top,right,left,tikz=dashed}]{3-3}{} \frac{5}{27} & \frac{1}{9} & \frac{1}{9} & \frac{4}{9 \sqrt{3}} \\
			0 & 0 & \frac{2}{9 \sqrt{3}} & \frac{2}{9 \sqrt{3}} & \frac{1}{9} & \frac{5}{27} & \frac{1}{9} & \frac{4}{9 \sqrt{3}} \\
			0 & \frac{2}{9 \sqrt{3}} & 0 & \frac{2}{9 \sqrt{3}} & \frac{1}{9} & \frac{1}{9} & \frac{5}{27} & \frac{4}{9 \sqrt{3}} \\
			0 & 0 & 0 & 0 & \frac{4}{9 \sqrt{3}} & \frac{4}{9 \sqrt{3}} & \frac{4}{9 \sqrt{3}} &  \Block[borders={left,top,tikz=dashed}]{1-1}{} \frac{5}{9}
		\end{pNiceMatrix}.
	\end{align} 
	Taking a closer look at Eqs.~\eqref{eq:CCZ} and \eqref{eq:iSWAP-CCZ}, we find that moment operators of 2-local gadgets are able to transit one basis element, except $\ket{u_0,u_0,u_0}$, to another one different by at most one Hamming weight. While 3-local gadgets transit basis elements different by at most two Hamming weights. The eigenvalues of the above matrix are $(1, 1, \frac{5}{9}, \frac{8}{27}, \frac{8}{27}, 0, 0, -\frac{1}{27} )$.
	
	This explains the intuition, from the perspective of Markov chain theory~\cite{Levin2009}, that ensembles defined by 3-local gadgets have the opportunity to mix/converge faster to 2-designs. However, as mentioned after Eq.~\eqref{eq:TrivialBound2}, there can be at most a constant speed up. For example, one can check, like \eqref{eq:iSWAP-U(4)}, that
	\begin{align}
		\frac{19}{12} (I - T_2^{\CCZ}) \geq ( I - T_2^{\iSWAP'} ) \geq  \frac{4}{5} (I - T_2^{\CCZ}).
	\end{align} 
	Therefore, on any hypergraph $\mathcal{H G}$, 
	\begin{align}
		\frac{19}{12} \Delta(\T_{2,\mathcal{H G}}^{\CCZ} ) \geq \Delta(\T_{2,\mathcal{H G}}^{\iSWAP} ) \geq  \frac{4}{5} \Delta(\T_{2,\mathcal{H G}}^{\CCZ} ).
	\end{align}
\end{example} 

We can compute other examples for typical 3-local gates like $\CSWAP$, $\CiSWAP$, $\Peres$ and $\Margolus$. The matrices representing 3-local second moment operators with their eigenvalues are listed in Table~\ref{tab:3-local}: 

\begin{table}[]
	\centering
	\resizebox{\columnwidth}{!}{%
		\begin{tabular}{c|c|c|c}
			\hline\hline
			Gate & Matrix representation & Spectrum & Spectral gap \\ 
			\hline\hline
			$\CSWAP$ & 	$T_2^{\CSWAP} = \resizebox{25em}{10em}{  
			$\begin{pmatrix}
				1 & 0 & 0 & 0 & 0 & 0 & 0 & 0 \\
				0 & \frac{1}{2} & 0 & 0 & 0 & 0 & 0 & \frac{1}{6} \\
				0 & 0 & \frac{1}{4} & \frac{1}{4} & \frac{1}{4\sqrt{3}} & 0 & \frac{1}{4\sqrt{3}} & 0 \\
				0 & 0 & \frac{1}{4} & \frac{1}{4} & \frac{1}{4\sqrt{3}} & 0 & \frac{1}{4\sqrt{3}} & 0 \\
				0 & 0 & \frac{1}{4 \sqrt{3}} & \frac{1}{4\sqrt{3}} & \frac{1}{4} & 0 & \frac{1}{4} & \frac{1}{3 \sqrt{3}} \\
				0 & 0 & 0 & 0 & 0 & \frac{2}{3} & 0 & \frac{1}{3 \sqrt{3}} \\
				0 & 0 & \frac{1}{4 \sqrt{3}} & \frac{1}{4\sqrt{3}} & \frac{1}{4} & 0 & \frac{1}{4} & \frac{1}{3 \sqrt{3}} \\
				0 & \frac{1}{6} & 0 & 0 & \frac{1}{3 \sqrt{3}} & \frac{1}{3 \sqrt{3}} & \frac{1}{3 \sqrt{3}} & \frac{11}{18}
			\end{pmatrix}$ }$ 
		    & $(1, 1, \frac{1}{18}(10 + \sqrt{10}), \frac{1}{6}(2 + \sqrt{2}), \frac{1}{18}(10 - \sqrt{10}),  \frac{1}{6}(2 - \sqrt{2}), 0, 0 )$ & 0.269 \\ 
		    \hline
			$\CiSWAP$ & $T_2^{\CiSWAP} = \resizebox{25em}{10em}{  
				$\begin{pmatrix}
					1 & 0 & 0 & 0 & 0 & 0 & 0 & 0 \\
					0 & \frac{1}{2} & 0 & 0 & 0 & 0 & 0 & \frac{1}{6} \\
					0 & 0 & \frac{1}{4} & \frac{1}{12} & \frac{1}{4\sqrt{3}} & \frac{1}{6\sqrt{3}} & \frac{1}{12\sqrt{3}} & \frac{1}{18} \\
					0 & 0 & \frac{1}{12} & \frac{1}{4} & \frac{1}{12\sqrt{3}} & \frac{1}{6\sqrt{3}} & \frac{1}{4\sqrt{3}} & \frac{1}{18} \\
					0 & 0 & \frac{1}{4\sqrt{3}} & \frac{1}{12\sqrt{3}} & \frac{1}{4} & \frac{1}{18} & \frac{7}{36} & \frac{7}{18\sqrt{3}} \\
					0 & 0 & \frac{1}{6\sqrt{3}} & \frac{1}{6\sqrt{3}} & \frac{1}{18} & \frac{5}{9} & \frac{1}{18} & \frac{2}{9\sqrt{3}} \\
					0 & 0 & \frac{1}{12\sqrt{3}} & \frac{1}{4\sqrt{3}} & \frac{7}{36} & \frac{1}{18} & \frac{1}{4} & \frac{7}{18\sqrt{3}} \\
					0 & \frac{1}{6} & \frac{1}{18} & \frac{1}{18} & \frac{7}{18\sqrt{3}} & \frac{2}{9\sqrt{3}} & \frac{7}{18\sqrt{3}} & \frac{31}{54}
				\end{pmatrix}$ }$ 
			& $(1, 1, \frac{1}{6}(2 + \sqrt{2}), \frac{1}{54}(20 + \sqrt{22}), \frac{1}{54}(20 - \sqrt{22}), \frac{2}{9}, \frac{1}{6}(2 - \sqrt{2}), 0)$ & 0.431 \\
			\hline
			$\Peres$ & 	$T_2^{\Peres} = \resizebox{25em}{10em}{  
				$\begin{pmatrix}
					1 & 0 & 0 & 0 & 0 & 0 & 0 & 0 \\
					0 & \frac{1}{3} & 0 & 0 & 0 & 0 & 0 & \frac{2}{9} \\
					0 & 0 & \frac{1}{3} & 0 & \frac{2}{3\sqrt{3}} & 0 & 0 & 0 \\
					0 & 0 & 0 & \frac{1}{3} & 0 & 0 & \frac{2}{3\sqrt{3}} & 0 \\
					0 & 0 & \frac{2}{3\sqrt{3}} & 0 & \frac{1}{9} & 0 & \frac{2}{9} & \frac{4}{9\sqrt{3}} \\
					0 & 0 & 0 & 0 & 0 & \frac{5}{9} & 0 & \frac{4}{9\sqrt{3}} \\
					0 & 0 & 0 & \frac{2}{3\sqrt{3}} & \frac{2}{9} & 0 & \frac{1}{9} & \frac{4}{9\sqrt{3}} \\
					0 & \frac{2}{9} & 0 & 0 & \frac{4}{9\sqrt{3}} & \frac{4}{9\sqrt{3}} & \frac{4}{9\sqrt{3}} & \frac{13}{27}
				\end{pmatrix} $  }$
			& $(1, 1, \frac{1}{27}(11 + 2\sqrt{10}), \frac{5}{9}, \frac{1}{9}(1 + 2\sqrt{2}),  \frac{1}{27}(11 - 2\sqrt{10}), \frac{1}{9}(1 - 2\sqrt{2}), -\frac{1}{3} )$ & 0.358 \\
			\hline
			$\Margolus$ & $T_2^{\Margolus} = \resizebox{25em}{10em}{  
				$\begin{pmatrix}
					1 & 0 & 0 & 0 & 0 & 0 & 0 & 0 \\
					0 & \frac{1}{2} & 0 & 0 & 0 & 0 & \frac{1}{3\sqrt{3}} & \frac{1}{9} \\
					0 & 0 & \frac{1}{2} & 0 & \frac{1}{6\sqrt{3}} & \frac{1}{6\sqrt{3}} & 0 & \frac{1}{18} \\
					0 & 0 & 0 & \frac{1}{6} & 0 & \frac{1}{6\sqrt{3}} & \frac{1}{2\sqrt{3}} & \frac{1}{18} \\
					0 & 0 & \frac{1}{6\sqrt{3}} & 0 & \frac{1}{6} & \frac{1}{18} & \frac{1}{9} & \frac{11}{18\sqrt{3}} \\
					0 & 0 & \frac{1}{6\sqrt{3}} & \frac{1}{6\sqrt{3}} & \frac{1}{18} & \frac{1}{3} & \frac{1}{18} & \frac{4}{9\sqrt{3}} \\
					0 & \frac{1}{3\sqrt{3}} & 0 & \frac{1}{2\sqrt{3}} & \frac{1}{9} & \frac{1}{18} & \frac{5}{18} & \frac{5}{18\sqrt{3}} \\
					0 & \frac{1}{9} & \frac{1}{18} & \frac{1}{18} & \frac{11}{18\sqrt{3}} & \frac{4}{9\sqrt{3}} & \frac{5}{18\sqrt{3}} & \frac{13}{27}
				\end{pmatrix}$ }$ 
		    & $(1, 1, 0.547, 0.407, 0.318, 0.215, -0.141, -0.087)$ & 0.453 \\
			\hline
		\end{tabular}
	}
	\caption{Matrix representations of $T_2^{\Gadget}$ for some typical 3-local gates.}
	\label{tab:3-local}
\end{table}

To summarize, the 3-local gadget associated with $\Margolus$ has the largest spectral gap and is thus supposed to possess the best efficiency among the typical 3-qubit gates studied above. It could also be better than gates built from 2-qubit gates like $\iSWAP$.


\section{Clifford + phase gate} \label{sec:large-n}

Here we consider a practically motivated case of gate sets composed of Clifford gates and a diagonal phase gate. Remarkably, we establish various analytical solutions. Our distribution $\nu$ has probability $p$ and $1-p$ for Clifford gates and diagonal gates respectively. To be more precise, with probability $p/2$ we choose uniformly randomly from all $n$-qubit Clifford gates, and with probability $(1-p)/4$ we choose the diagonal gate $G \equiv \operatorname{diag}(1,e^{i\phi})$ {acting on the first qubit}. Note that since SWAP gates are contained in the Clifford group, our gate set is still universal. Similarly, with probability $(1-p)/4$ we choose the diagonal gate $G^\dagger$ acting only on the first qubit. Finally with probability $1/2$ we choose identity. Since Clifford group forms exact 3-designs on qubits~\cite{webb2015clifford,zhu2016clifford,zhu2017multiqubit} with efficient sampling methods~\cite{CliffordSampling2014,CliffordSampling2021}, we here focus on the cases of 4- and 5-designs. {For any $n\ge 3$ for $t=4$, and $n\ge 4$ for $t=5$, we found that $\phi=\pi/4$ (i.e. the gate $G$ is the $T$ gate) is one of the optimal solutions for the fastest convergence, and we derived an analytical result for the optimal value of $p$. The condition $n\ge t-1$ is required for the states invariant under Clifford group to be linearly independent.}

\subsection{Case of 4-design} 

As mentioned in Section \ref{sec:framework}, the 4-th moment $\T_4^{U(2^n)}$ is a projector onto the space spanned by permutations from symmetric group $S_4$ acting on $((\mathbb{C}^2)^{\otimes n})^{\otimes 4}$. Explicitly, they can be written as $\ket{\Psi_{\pi}} \equiv \ket{\psi_{\pi}}^{\otimes n}$ where
\begin{align}
	|\psi_\pi\> \equiv \frac{1}{4}\sum_{i \in \{0,1\}^4 } |i\> \otimes |\pi(i)\>
\end{align}
and $\pi \in S_4$ is a permutation. The dimension of this space is $4!=24$ for any $n\ge 2$. 

It is known~\cite{GNW17} that the moment operator $T_4^{\text{Cliff}}$, averaged over Clifford group, is a projector onto a 30 dimensional subspace. Besides the 24 states mentioned above, there are 6 other states defined in the following way. We define $s_1,\ldots,s_4$ as the 8-bit strings $10011001,01010101$, $11110000,00001111$. The state $\ket{T}$ is then defined as a uniform superposition of all bit strings in the subspace of $\F_2^8$ spanned by $\{s_1,\ldots,s_4\}$
\begin{align}
	\ket{T} = \frac{1}{4}\sum_{x_i \in \{0,1\}} \ket{x_1 s_1 + \ldots + x_4 s_4}
\end{align}
where the addition inside bracket is done bit-wise. The extra states we want are $\ket{\Phi_\pi} \equiv \ket{\phi_{\pi}}^{\otimes n}$ with $\pi \in S_4$ and
\begin{align}
	\ket{\phi_\pi} \equiv (\pi \otimes I) \ket{T}.
\end{align}

Note that in~\cite{GNW17} the permutations could be applied to the last 4 qubits as well, but we have
\begin{align}
	(\pi \otimes \pi)\ket{T} = \ket{T} 
	\implies (\sigma \otimes \pi)\ket{T} = (\pi^{-1}\sigma  \otimes I)\ket{T} 
\end{align}
in our case, so we can set the second permutation to be identity. To get the counting right, note that $\ket{\phi_\pi} = \ket{T}$ when $\pi \in \{(1), (12)(34),(13)(24),(14)(23)\}$, so the number of states equals to the number of cosets of the subgroup above, which is $4!/4=6$. The 30 states $|\psi_\pi\>$ and $|\phi_{\pi}\>$ are linearly independent, given that $n \ge t-1 = 3$~\cite{GNW17}. 

Now we consider the action of $G^{\otimes 4,4} = G^{\otimes 4} \otimes \bar{G}^{\otimes 4}$ on $|\phi_\pi\>$ where $G=\operatorname{diag}(1,e^{i\theta})$ is the diagonal gate. The component $\ket{11110000}$ and $\ket{00001111}$ will get a phase of $e^{4i\theta}$ and $e^{-4i\theta}$ respectively, while other components are unaffected. Therefore
\begin{align}
	\frac{G^{\otimes 4,4}+(G^\dagger)^{\otimes 4,4}}{2} \ket{\phi_\pi} = \ket{\phi_\pi} + \frac{\cos(4\theta)-1}{4} (\ket{11110000} + \ket{00001111}).
\end{align}
We want to find a linear combination of $\ket{\phi_\pi}$ and $\ket{11110000} + \ket{00001111}$ that is orthogonal to $\ket{\phi_\pi}$, which is
\begin{align}
	\ket{\phi_\pi'} \equiv \frac{1}{\sqrt{7}}\left( \ket{\phi_\pi} - 2( \ket{11110000} + \ket{00001111} )\right).
\end{align}
Now we define 
\begin{align}
	\ket{\Phi_\pi'} \equiv \ket{\phi_\pi'} \otimes \ket{\phi_\pi}^{\otimes (n-1)}
\end{align}
and it can be seen that the all the states $\ket{\Psi_\pi}$, $\ket{\Phi_\pi}$ and $	\ket{\Phi_\pi'}$ are approximately orthogonal in the large $n$ limit.

By construction of our single-step distribution $\nu$,
\begin{align}
	\T_4^{\text{Cliff}+\phi} = \frac{p}{2} T_4^{\text{Cliff}} + \frac{p-1}{4}(G_1^{\otimes 4,4}+(G_1^\dagger)^{\otimes 4,4}) + \frac{I}{2}
\end{align}
where $G_1$ means the gate $G = \operatorname{diag}(1,e^{i\theta})$ acting on the first qubit. It is not hard to see that the space $P$ spanned by the set of states $S \equiv \{ \ket{\Psi_\pi} \} \cup \{ \ket{\Phi_\pi} \} \cup \{ \ket{\Psi_\pi'} \}$ is an invariant subspace of $T_4^{\mathcal{E}}$. The states $\ket{\Psi_\pi}$ are obviously unit eigenstates for $T_4^{\text{Cliff}+\phi}$, while on each pair of $\ket{\Phi_\pi}$ and $\ket{\Phi_\pi'}$ the operator $T_4^{\text{Cliff}+\phi}$ is represented as
\begin{align}\label{eq:matrix2d}
    \begin{pmatrix}
        \frac{1}{2} + \frac{1}{2}\left(p+(1-p)\frac{c+7}{8}\right) &  -\frac{\sqrt{7}(p-1)}{16}(c-1) \\
        -\frac{\sqrt{7}(p-1)}{16}(c-1) &
        \frac{1}{2} + \frac{(1-p)(1+7c)}{16}
    \end{pmatrix},
\end{align}
where $c = \cos4\theta$. With a direct diagonalization we can find the gap $\Delta(p,c)$ in this subspace,
\begin{align} \label{eq:gap-4}
    \Delta(p,c) = \frac{1}{4} - \frac{c}{4}(1-p)-\frac{1}{8}\sqrt{4 c^2 (p-1)^2-2 c (p-4) (p-1)+2 p (p-1)+4},
\end{align}
and by taking the derivative w.r.t. $p$ we can find the maximum gap within this subspace and the optimal value of $p$,
\begin{align}
	\tilde \Delta(c) = \frac{1-c}{8 \left(-2 c+ \sqrt{14(1-c)}+4\right)},\quad \tilde p(c) = \frac{4 c^2+ \sqrt{14(1-c)} c-5 c+1}{4 c^2-2 c+2}.
\end{align}

Their plots against $c$ are shown below.
\begin{figure}[ht]
    \centering
    \includegraphics[width=0.4\textwidth]{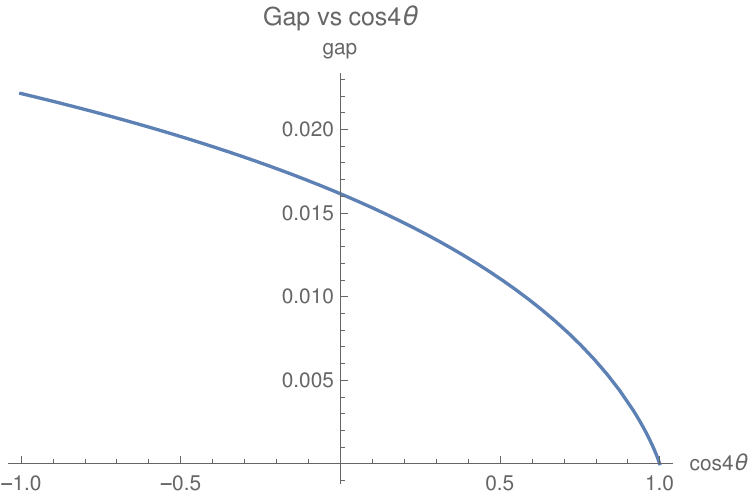}
    \includegraphics[width=0.4\textwidth]{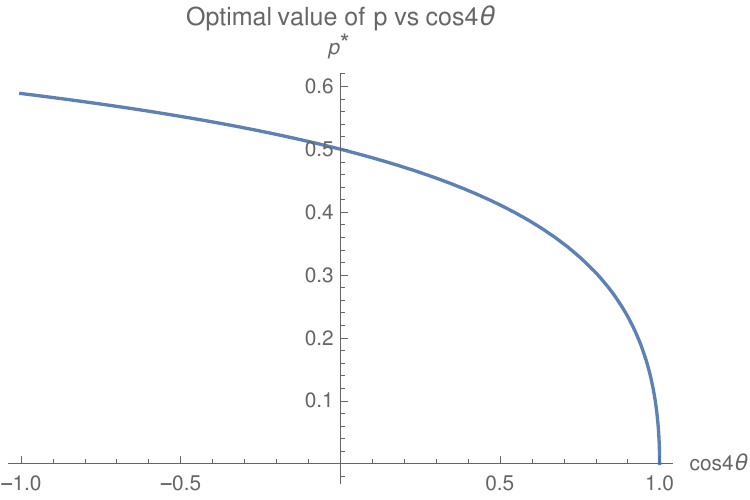}
\end{figure}
The largest gap is obtained at $c=-1$, and the $T$ gate with $\theta = \exp(i\pi/4)$ is one such case.  This is somewhat intuitive because $T$ is the most ``magical'' in this class.

Outside of the subspace $P$ the Clifford integral is 0, so the largest eigenvalue of $T_4^{\text{Cliff}+\phi}$ is at most $1-\frac{p}{2}$. It is easy to show that
\begin{align}
	\Delta(p,c) \ge \frac{p}{2}
\end{align}
for all values of $c$, so $\Delta(p,c)$ is indeed the global spectral gap.

Note that the calculations above are carried out on a set of non-orthogonal basis. When we switch to a set of orthonormal basis, the matrix undergoes a similarity transformation, which has no effect on the eigenvalues. The means that the gap we obtained is exact for all values of $n$.

\subsection{Higher designs}

The $t=5$ case is closely related to the $t=4$ case. We define the strings $s_1',\ldots,s_5'$ to be 1001010010, 0101001010, 0000100001, 0000011110, 1111000000, and the state
\begin{align}
	\ket{T'} = \frac{\sqrt{2}}{8}\sum_{x_i \in\{0,1\}} \ket{x_1 s_1 + \ldots x_5 s_5},
\end{align}
and all the states that span the subspace to which the Clifford average projects are
\begin{align}
	\ket{\Phi_{\sigma,\pi}} \equiv \sigma \otimes \pi \ket{T'}
\end{align}
where $\sigma,\pi \in S_5$. Note that $\ket{T'}$ is just a product state of $\ket{T}$ and a bell pair on 5th and 10th qubit. We could see that the operator $T_5^{\text{Cliff}+\phi}$ should take the same form as the $t=4$, and as a result the gap is given by Eq.~(\ref{eq:gap-4}). Note that in order to ensure the states above are linearly independent, we need $n \ge t-1 = 4$.

For $t$-designs with $t\ge 6$, the method above no longer works, as we will have to diagonalize a high dimensional matrix analytically, which is generally impossible.

\section{Discussion and outlook}

In this work, we explored the stochastic quantum circuit framework and examined the design generation features of various mathematically or experimentally motivated settings, which yield a number of interesting results and insights about the efficiency of different quantum gates and circuit architectures. In particular, we studied the ironed gadget model in depth, leveraging various nontrivial mathematical insights and techniques along the way. This leads to a systematic mathematical  theory for understanding the efficiency of entangling gates and circuit architectures, for instance, identifying certain gates with their exceptional convergence performance such as the newfound $\chi$ and $\iSWAP$ gates.
We expect the results to offer timely value for both theoretical study and experimental design of quantum circuits. 

As mentioned, our framework draws crucial motivation from its relevance to practical limitations of quantum computation architectures. For instance, the implementation of certain gadget models is expected to be directly beneficial for design generation and scrambling experiments. Beyond convergence, they are expected to exhibit all important behaviours of conventional random circuits such as OTOC, entanglement, and complexity dynamics, anticoncentration, etc.
More generally, our efficiency results for gates and circuit topologies are anticipated to provide useful guidelines for experimental efforts, facilitating the design and optimization of quantum experiments. We expect further interplay between these theoretical analyses and experiments to be particularly fruitful.

Our theory gives rise to an abundance of interesting mathematical problems, some of which we have left open for further exploration. A few eminent conjectures from the gadget model are presented in Section~\ref{sec:conjectures}. In particular, it would be interesting to settle the universal optimality conjectures of $\iSWAP$.
Moreover, for high-order designs, the situation such as particularly efficient gates remain less understood. In fact, even on a single qubit, our numerical computations in Section \ref{sec:one-qubit} indicate heterogeneous and complicated behaviours in the choice of efficient gates for different $t$. Rigorous investigations would be interesting for future work. Moreover,  there are other circuit structures that can nontrivially accelerate the convergence speed, e.g., using high dimensional lattices~\cite{harrow2023approximate} or small random circuits~\cite{Schuster2024lowDepth}. It would be interesting to study efficient gate sets based on these models to further improve the performance. Beyond these, extending these results to qudits or continuous variable (CV) models are also valuable directions for future research. Exploring these generalizations could lead to new insights and applications, enhancing our understanding of quantum randomness and circuit architectures.

\section*{Acknowledgements}
We thank Fei Yan for valuable conversations.
ZL and ZWL are supported in part by a startup funding from YMSC, Tsinghua University, and NSFC under Grant No.~12475023.

\bibliography{convergence}
\bibliographystyle{unsrt}

\end{document}